\documentclass[11pt,letterpaper]{article}

\usepackage[margin=1in]{geometry}
\usepackage{setspace}
\usepackage{libertinus}
\input{glyphtounicode}
\usepackage{amsmath}
\usepackage{mathtools}
\usepackage{tikz}
\usepackage{mathdots}
\usepackage{amsfonts}
\usepackage{amsthm}
\usepackage{latexsym}
\usepackage{graphicx}
\usepackage{color}
\usepackage[hidelinks]{hyperref}
\hypersetup{pdftitle={Hermite Brings a Laptop: Analyzing Frieze-Jerrum Rounding Yields Improved Approximations for Clustering Problems},pdfauthor={David Garc\'ia-Soriano and Atsushi Miyauchi}}
\usepackage[linesnumbered,ruled,vlined]{algorithm2e}
\usepackage{booktabs}
\usepackage{arydshln}
\usepackage{bm}
\usepackage{multirow}
\usepackage{makecell}
\usepackage{booktabs}
\usepackage{balance}
\usepackage{subcaption}
\usepackage{fullpage}
\usepackage{authblk}
\usepackage[numbers,sort&compress]{natbib}
\usepackage{amssymb}
\usepackage[T1]{fontenc}

\DeclareMathOperator*{\argmax}{arg\,max}

\DeclareMathOperator{\E}{\mathbb{E}}
\DeclareMathOperator{\Var}{\mathrm{Var}}
\DeclareMathOperator{\Cov}{\mathrm{Cov}}

\newcommand{\maxagree}{\textsc{MaxAgree}}
\newcommand{\maxagreek}{\texorpdfstring{$\textsc{MaxAgree}[K]$}{MaxAgree[K]}}
\newcommand{\mindisagree}{\textsc{MinDisagree}}
\newcommand{\1}{\mathbf{1}}
\newcommand{\poly}{\operatorname{poly}}

\usepackage{xcolor}

\newtheorem{theorem}{Theorem}[section]
\newtheorem{lemma}[theorem]{Lemma}
\newtheorem{corollary}[theorem]{Corollary}
\newtheorem{proposition}[theorem]{Proposition}

\theoremstyle{definition}

\theoremstyle{remark}

\theoremstyle{plain}

\newcommand{\mycomment}[1]{}

\title{
Hermite Brings a Laptop: Analyzing Frieze--Jerrum Rounding Yields Improved Approximations for Clustering Problems
}

\author[1]{David Garc\'ia-Soriano\thanks{david.garcia.soriano@upc.edu}}
\author[2]{Atsushi Miyauchi\thanks{atsushi.miyauchi@intesasanpaolo.com}}

\affil[1]{Universitat Polit\`ecnica de Catalunya, Barcelona, Spain}
\affil[2]{Intesa Sanpaolo, Turin, Italy}

\date{}

\begin{document}
\pagenumbering{gobble}
\hypersetup{pageanchor=false}
\maketitle

\vspace{-2mm}
\begin{abstract}

The Frieze--Jerrum rounding scheme (Algorithmica~1997) is a standard tool for rounding semidefinite programming (SDP) relaxations of graph partitioning and clustering problems, assigning nodes to at most $k$ clusters using $k$ independent Gaussian vectors.
Its analysis hinges on the collision probability $P_k(\rho)$ that two nodes whose SDP vectors have inner product $\rho$ are assigned to the same cluster.  No tractable closed form for $P_k$ is known for $k\geq 4$, making it difficult to certify approximation guarantees and hindering the systematic search for better algorithms.

We develop a Hermite-coefficient certification framework to derive accurate and tractable bounds on $P_k$. Using the Hermite expansion of Gaussian noise stability, we
express $P_k$ as a power series with nonnegative coefficients, reduce these coefficients to one-dimensional Gaussian integrals, and certify finitely many of them, yielding rigorous bounds on $P_k$ over the entire correlation range.
Our framework yields strengthened polynomial-time approximations for
several
clustering problems.

For {\maxagree} Correlation Clustering, we derive a $0.7818$-approximation, the first
improvement in two decades over the $0.7666$ ratio of Swamy (SODA~2004). 
We also improve the best known ratios for its variant with at most $K$ clusters, \maxagreek\ (e.g., from $0.77$ to $0.8151$ for $K=3$).  
For Max $K$-Cut we resolve, via a structural property of the Hermite expansion, a conjecture of de~Klerk, Pasechnik and Warners (J. Comb. Optim.~2004) characterizing the Frieze–Jerrum approximation ratio
for every $K\ge3$; we show that this ratio is tight, and
determine it to within $10^{-6}$ accuracy for $K\leq16$.  
Finally, we reduce the additive approximation error for modularity maximization from $0.42084$ (Kawase, Matsui, and Miyauchi, JCSS~2021) to 
$0.3790$.

On the hardness side, we construct an integrality-gap instance for
\maxagree\ in Gaussian space and analyze it by applying Borell's
isoperimetric inequality to each cluster, yielding an integrality
ratio of at most $0.802$ for the standard SDP relaxation.
Combined with Raghavendra's CSP framework, this shows  that
approximation within $0.802+\epsilon$ is Unique Games-hard for every
$\epsilon>0$, the first such hardness result for this problem.

\end{abstract}

\thispagestyle{empty}

\clearpage
\tableofcontents
\thispagestyle{empty}

\clearpage
\pagenumbering{arabic}
\hypersetup{pageanchor=true}

\section{Introduction}

Semidefinite programming (SDP) has become a fundamental tool for designing approximation algorithms
for combinatorial optimization problems~\cite{Gartner+12,Vazirani01,Williamson+11}.
A common paradigm is to solve an SDP relaxation, interpret the resulting solution as a collection of vectors,
and convert these vectors into a discrete solution by randomized rounding.
One particularly natural rounding scheme, especially for graph partitioning and clustering problems, is the Frieze--Jerrum (FJ) rounding~\cite{Frieze+97}.
Given unit SDP vectors and an integer $k$, the scheme samples $k$ independent Gaussian vectors
and assigns each node to the cluster corresponding to the Gaussian vector with which its SDP vector has the largest inner product.
We refer to this scheme as the $k$-spokes rounding.
Originally introduced for Max $K$-Cut~\cite{Garey+02},
the FJ rounding and closely related rounding schemes have since been used in various graph partitioning, clustering, and coloring problems~\cite{Karger+98,Swamy04}.

For problems whose objectives decompose over edges, the approximation guarantee of the FJ rounding can naturally be analyzed edge by edge.
Consider a pair of nodes whose SDP vectors have inner product $\rho$.
The probability that they are assigned to the same cluster depends only on $\rho$; we call this probability the collision probability and denote it by $P_k(\rho)$.
Consequently, the probability that the two nodes are assigned to different clusters is $1-P_k(\rho)$.
The approximation analysis can then be reduced to comparing these probabilities with the corresponding SDP contributions as functions of $\rho$.
Thus, understanding the behavior of the collision probability $P_k(\rho)$ is a central analytic problem in establishing approximation guarantees for the FJ rounding.

The main difficulty here is that the collision probability of the FJ rounding is analytically hard to handle.
Although explicit formulas are available for $k=2$~\cite{Goemans+95} and $k=3$~\cite{Goemans+04,deKlerk+04},
no comparably tractable closed-form representation is known for general $k$.
One natural way to cope with this absence is to bound the collision probability by analytically tractable functions.
Such bounds can be sufficient to derive nontrivial approximation guarantees, as in the original analysis of \citet{Frieze+97} for Max $K$-Cut;
however, it is generally difficult to determine whether they are sufficiently tight for fine-grained analyses.
A different approach is to modify the rounding scheme itself so that its collision probability becomes more amenable to explicit analysis.
For example, \citet{Newman18} designed a rounding scheme for general $k$ whose collision probability admits a tractable analytic expression.
Such alternative roundings facilitate the analysis, but can yield weaker approximation guarantees than the original FJ rounding.

The lack of accurate and tractable bounds on the collision probability poses a broader obstacle.
It not only makes it difficult to certify the approximation guarantees of the specified $k$-spokes roundings,
but also hinders the systematic search for better algorithms based on different choices of $k$ or their mixtures.
Thus, a better understanding of the collision probability has the potential both to sharpen the analysis of existing algorithms and to guide the design of new ones.
This motivates a basic question:
can one obtain sufficiently accurate and tractable bounds on the collision probability of the FJ rounding without requiring a closed-form representation?

We answer this question affirmatively.
We present a framework, which we call the Hermite-coefficient certification framework,
for deriving accurate and tractable bounds on the collision probabilities of the FJ rounding.
As its name suggests, our approach is based on the Hermite expansion of the collision probability.
We first reduce the computation of each Hermite coefficient, which a priori involves a $k$-dimensional Gaussian integral,
to a finite collection of explicitly structured one-dimensional Gaussian integrals.
Building on this representation, we further develop a certification procedure that rigorously encloses these integrals on a finite domain
and controls the remaining tails explicitly, yielding certified rational lower and upper bounds on the coefficients.
Bounds on finitely many coefficients can then be combined with coefficient nonnegativity to obtain certified bounds on the entire collision probability.
In developing our framework, we also prove a novel structural property of the separation probability $1-P_k(\rho)$.

We utilize this framework to obtain improved approximation guarantees and new insights for several important graph partitioning and clustering problems,
namely \maxagree~\cite{Bansal+04}, \maxagreek~\cite{Giotis+06_j},
Max $K$-Cut~\cite{Garey+02,Frieze+97}, and modularity maximization~\cite{Newman+04,Fortunato10}.
Specifically, for \maxagree, we present a polynomial-time $0.7818$-approximation algorithm,
improving upon the long-standing state-of-the-art approximation ratio of $0.7666$ due to \citet{Swamy04}.
For \maxagreek, we present a polynomial-time approximation algorithm that substantially improves the best known approximation ratios for a range of values of $K$~\cite{Swamy04}.
Both algorithms employ the $k$-spokes roundings for multiple values of $k$, and their analyses are based on carefully chosen mixtures of these roundings.
Both the design and the analysis of these algorithms are enabled by our Hermite-coefficient certification framework.
For Max $K$-Cut, using the aforementioned structural property of the collision probability,
we obtain a simpler characterization of the approximation ratio of the classical $K$-Cut algorithm of \citet{Frieze+97},
prove its tightness for the algorithm, and determine it rigorously to high precision for a range of values of $K$.
In particular, our rigorous lower bounds improve upon the best known explicit approximation ratios.
Our framework also yields an improved additive approximation algorithm for modularity maximization.

We complement these algorithmic results with nearly matching upper bounds
for \maxagree.  Unconditionally, we show that the integrality ratio of the
standard SDP relaxation used by our \maxagree\ algorithm is at most $0.802$,
improving the previous bound $2/(1+\sqrt2)\approx0.8284$~\cite{Charikar+05}.
Assuming the Unique Games Conjecture, we further show that, for every constant
$\epsilon>0$, no polynomial-time algorithm for weighted \maxagree\ can exceed
this integrality ratio by $\epsilon$; in particular, none achieves an
approximation ratio of $0.802+\epsilon$.  To the best of our knowledge, this is
the first UGC-based inapproximability result for \maxagree.  Consequently, the
best approximation ratio achievable in polynomial time under the UGC lies in
$[0.7818,0.802]$.  The integrality gap rests on an explicit instance in
Gaussian space, analyzed with Borell's isoperimetric inequality, and on the
strong-gap theorem of \citet{RaghavendraSteurer09Gaps}; see
Section~\ref{subsec:raghavendra-overview} for an overview.

Finally, we study the limitations of rounding families under the edge-wise analysis.
For any prescribed finite family of rounding schemes, we develop a framework for upper-bounding the best approximation ratio
that can be certified by any mixture over that family under the edge-wise analysis.
Applied to natural broad families of the FJ and hyperplane roundings,
these bounds show that our proposed mixtures for \maxagree\ and \maxagreek\ are nearly optimal within the corresponding family and analysis framework.

\subsection{Hermite-Coefficient Certification Framework}

Here we describe the main ideas underlying our Hermite-coefficient certification framework,
which is based on the Hermite expansion of the collision probability.

The starting point is to express the collision event in Gaussian space.
For each spoke, the event that it is selected corresponds to an argmax region in $\mathbb{R}^k$.
For two SDP vectors with inner product $\rho$, the corresponding vectors of Gaussian inner products are $\rho$-correlated standard Gaussian vectors.
Consequently, $P_k(\rho)$ can be expressed as a sum of the probabilities that two $\rho$-correlated Gaussian vectors both fall into the same argmax region.
In other words, $P_k(\rho)$ is a sum of Gaussian noise stabilities of these regions.
Expanding the indicators in the Hermite basis, which diagonalizes the Gaussian noise operator, gives
\begin{equation*}
P_k(\rho) = \sum_{m=0}^{\infty} c_{k,m}\rho^m,
\end{equation*}
in which each $c_{k,m}$ is a sum of squares of degree-$m$ Hermite coefficients of the argmax indicators.

The main technical challenge is to obtain accurate and tractable bounds on the coefficients $c_{k,m}$.
Their direct definition involves $k$-dimensional Gaussian integrals over argmax regions and is not well suited for certified computation.
By exploiting symmetry and conditioning on the winning Gaussian coordinate,
we reduce the Hermite coefficients of the argmax indicators to a finite collection of explicitly structured one-dimensional Gaussian integrals.
After a suitable normalization and symmetry grouping,
each coefficient $c_{k,m}$ can be recovered as a nonnegative weighted sum of squares of these one-dimensional integrals.

An important subtlety is that the resulting one-dimensional integrals are signed.
Simply restricting such an integral to a bounded interval need not give a valid bound on its square,
because the omitted tails may partially cancel the contribution from the bounded interval.
The explicit Gaussian form of the integrands nevertheless makes rigorous control possible.
We enclose the integral over a bounded interval using validated numerical integration in ball arithmetic
and bound the two tails analytically using Gaussian decay, which dominates the polynomial growth contributed by the Hermite polynomials.
Combining these bounds gives a rigorous enclosure of the value of the full integral, which can then be safely squared.
These interval enclosures are then propagated through the sum-of-squares representation, yielding certified rational lower and upper bounds
\begin{equation*}
\underline{c}_{k,m}
\le
c_{k,m}
\le
\overline{c}_{k,m}.
\end{equation*}
Thus, numerical integration is used only to produce rigorously certified rational data, rather than as evidence for an unverified numerical approximation.

Crucially, to obtain rigorous bounds on the entire collision-probability function, it is enough to certify only finitely many coefficients $c_{k,m}$.
Suppose that lower and upper bounds $\underline{c}_{k,m}$ and $\overline{c}_{k,m}$ are available up to degree $M$.
For nonnegative correlations $\rho$,
a lower bound on $P_k(\rho)$ is obtained by replacing the known coefficients by their lower bounds and discarding the nonnegative unknown tail.
For an upper bound, we replace the known coefficients by their upper bounds and use the property $\sum_{m=0}^\infty c_{k,m}=1$ to bound the total mass of the unknown tail.
For negative correlations $\rho$, the powers of $\rho$ alternate in sign:
a lower bound therefore uses lower coefficient bounds for even degrees and upper coefficient bounds for odd degrees, while an upper bound uses the opposite choices.
In both cases, the remaining unknown tail is controlled in absolute value using the same normalization.
In this way, finitely many certified bounds $\underline{c}_{k,m}$ and $\overline{c}_{k,m}$ yield explicit functions satisfying
\begin{equation*}
\underline{P}_{k,M}(\rho)
\le
P_k(\rho)
\le
\overline{P}_{k,M}(\rho)
\quad
\text{for all }\rho\in[-1,1].
\end{equation*}
This construction extends immediately to mixtures of the $k$-spokes roundings by linearity.
Thus, our framework obtains rigorous control of the entire collision-probability function for any mixture of the $k$-spokes roundings from a finite amount of certified data.

The Hermite representation of the collision probability also yields a novel structural property of the associated separation probability.
Consider the normalized separation probability
\begin{equation*}
R_k(\rho)
:=
\frac{1-P_k(\rho)}{1-\rho}
\qquad
\text{for }\rho\in[-1,1).
\end{equation*}
Using the Hermite expansion and the fact that its coefficients $c_{k,m}$ sum to one, $R_k(\rho)$ can be written as a nonnegative linear combination of the functions
$(1-\rho^m)/(1-\rho)$.
We prove that each of these functions is nondecreasing on $\rho\in[-1/2,1)$.
Coefficient nonnegativity therefore implies that $R_k$ itself is nondecreasing on the same interval, and this property is preserved under arbitrary mixtures of the $k$-spokes roundings.
This monotonicity allows us to reduce certain inequalities involving the separation probability $1-P_k(\rho)$ over an entire interval of correlations, which later appear in the approximation-ratio analysis, to a single endpoint condition, significantly simplifying the analysis.
Moreover, as we explain later, it also leads almost immediately to a proof of a long-standing conjecture concerning the FJ rounding for Max $K$-Cut.

\subsection{Improved Approximations}
The Hermite-coefficient certification framework enables us to rigorously bound the collision probability of the $k$-spokes rounding, despite the absence of a closed-form representation for $k\geq 4$.
We utilize this framework to obtain improved approximation guarantees and novel insights for several important clustering problems, namely \maxagree~\cite{Bansal+04}, \maxagreek~\cite{Giotis+06_j}, Max $K$-Cut~\cite{Garey+02,Frieze+97}, and modularity maximization~\cite{Newman+04,Fortunato10}, all of which will be introduced formally in the subsequent subsections.

To treat the first three problems in a unified manner, we view them as optimization problems on weighted edge-labeled graphs $G=(V,E)$, where each edge $\{i,j\}\in E$ is labeled either positive or negative. \maxagree\ and \maxagreek\ are naturally formulated in this setting, while Max $K$-Cut corresponds to the special case in which all edges are negative.
We introduce a family of SDPs parameterized by $a\in[-1/2,0]$.
For $a\in[-1/2,0]$, define two affine functions
\begin{equation*}
s_a^+(\rho)=\frac{\rho-a}{1-a}
\quad\text{and}\quad
s_a^-(\rho)=\frac{1-\rho}{1-a}
\quad\text{for }\rho\in[a,1].
\end{equation*}
Together with the edge weights $w_{ij}$, these functions represent the SDP contributions of positive and negative edges $\{i,j\}\in E$, respectively.
The parameter $a$ lower bounds the inner product between any pair of SDP vectors and, equivalently, determines an upper bound on the angle between them.
The standard SDP relaxation for \maxagree\ used by \citet*{Charikar+05} and \citet{Swamy04} corresponds to the SDP with $a=0$, whereas the SDP relaxations for \maxagreek\ and Max $K$-Cut used by \citet{Swamy04} and \citet{Frieze+97}, respectively, correspond to the SDP with $a=-1/(K-1)$.

To analyze the approximation ratio of a rounding scheme for the family of SDPs, we formalize the edge-wise analysis (Lemma~\ref{lem:edge-wise-analysis}).
Consider a rounding scheme whose collision probability depends only on the inner product $\rho$ of the two corresponding SDP vectors, and denote this probability by $P(\rho)$.
For an edge $\{i,j\}\in E$ whose corresponding SDP vectors have inner product $\rho$, its expected contribution after rounding is $w_{ij}P(\rho)$ if the edge is positive and $w_{ij}(1-P(\rho))$ if it is negative, whereas its SDP contribution is $w_{ij}s_a^+(\rho)$ in the former case and $w_{ij}s_a^-(\rho)$ in the latter.
Thus, provided that the SDP is a relaxation of the problem and the rounding always produces a feasible solution, an approximation ratio $\gamma$ is certified by the inequalities
\begin{equation*}
P(\rho)\geq \gamma s_a^+(\rho)
\quad\text{and}\quad
1-P(\rho)\geq \gamma s_a^-(\rho)
\quad\text{for all }
\rho\in[a,1].
\end{equation*}
Indeed, these inequalities guarantee that the expected contribution of every edge after rounding is at least a $\gamma$ fraction of its SDP contribution, and summing over all edges yields the desired approximation ratio.
We refer to these two inequalities as the positive-edge and negative-edge inequalities.
Importantly, for a mixture of the $k$-spokes roundings, our monotonicity lemma implies that the negative-edge inequality over the entire interval $[a,1]$ is equivalent to the single endpoint condition $1-P(a)\geq \gamma$.
Consequently, certifying an approximation ratio of $\gamma$ via the edge-wise analysis reduces to verifying the positive-edge inequality over $[a,1]$ and the negative-edge inequality at the single point $a$.

We then describe how our Hermite-coefficient certification framework further reduces the task of certifying the positive-edge and negative-edge inequalities for any mixture of the $k$-spokes roundings to finite, exactly checkable conditions.
The key observation is simple.
Using certified rational lower and upper bounds on the relevant Hermite coefficients up to degree $M$, for some truncation degree $M\in \mathbb{N}_0$, our framework yields certified rational lower and upper bounds $\underline P_M(\rho)$ and $\overline P_M(\rho)$, respectively, on the collision probability $P(\rho)$ of the mixture.
Therefore, it suffices to verify
\begin{equation*}
\underline P_M(\rho)\geq \gamma s_a^+(\rho)
\quad\text{for all } \rho\in[a,1]
\quad\text{and}\quad
1-\overline P_M(a)\geq \gamma.
\end{equation*}
The latter condition, which certifies the negative-edge inequality, can be checked directly in exact rational arithmetic.
The former condition still ranges over infinitely many $\rho\in [a,1]$, but it reduces to the nonnegativity of a rational polynomial on an interval, which we certify in exact rational arithmetic using Bernstein certificates.

For modularity maximization, we use the standard SDP relaxation employed by \citet*{Dinh+15} and \citet*{Kawase+21}.
Although this relaxation is equivalent to our $a=0$ relaxation up to an additive constant in the objective,
         the presence of negative coefficients in the objective means that
         treating the problem as a $\maxagree$ instance
         introduces an additive constant, preserving optimal solutions but not multiplicative approximation guarantees. We therefore analyze the additive error directly.
Moreover, our analysis is not edge-wise either. Nevertheless, certifying the approximation guarantee again reduces to verifying an inequality involving a certified lower bound on the collision probability of the $k$-spokes rounding, allowing us to apply the same Hermite-coefficient certification framework.

In the following subsections, we briefly review the background and existing approximation results for each problem and summarize our contributions.

\subsubsection{\maxagree}
\maxagree\ is a fundamental optimization problem in Correlation Clustering~\citep{Bansal+04},
a clustering model that has been actively studied in both the theoretical computer science and machine learning communities~\cite{Bonchi+22}.
Let $G=(V,E)$ be an edge-labeled simple undirected graph with $V=\{1,\dots,n\}$ and $E=E_+\cup E_-$, where $E_+\cap E_-=\emptyset$.
Each edge $\{i,j\}\in E$ has a weight $w_{ij}>0$, representing the strength of its label.
A positive edge $\{i,j\}\in E_+$ indicates that $i$ and $j$ should belong to the same cluster,
whereas a negative edge $\{i,j\}\in E_-$ indicates that they should belong to different clusters.
In the general weighted \maxagree, the goal is to find a clustering of $V$ that maximizes the total weight of correctly classified edges.
Its complementary problem, \textsc{MinDisagree}, asks for a clustering that minimizes the total weight of incorrectly classified edges.
Since the two objective values sum to the total edge weight, the two problems are equivalent for exact optimization,
and both are known to be NP-hard even on unweighted complete graphs~\cite{Bansal+04}.
A substantial body of work has therefore focused on approximation algorithms for these problems and their variants.

\citet*{Charikar+05} designed a polynomial-time $0.7664$-approximation algorithm for the general weighted \maxagree.
Their algorithm solves a standard SDP relaxation and rounds the SDP solution using the random hyperplane rounding scheme of \citet*{Goemans+95}.
More specifically, they use the $2$-hyperplane and $3$-hyperplane roundings, which produce clusterings with at most four and eight clusters, respectively,
and return the better of the resulting clusterings.
Subsequently, \citet{Swamy04} improved the approximation ratio to $0.7666$ by additionally using the FJ rounding scheme~\cite{Frieze+97}.
In particular, the ratio $0.7666$ is obtained by combining the $2$-hyperplane rounding with the $6$-spokes rounding and returning the better of the two resulting clusterings.
This approximation ratio has remained the state of the art for more than two decades.

\paragraph{Our Results.}
We present a polynomial-time $0.7818$-approximation algorithm for \maxagree\ (Theorem~\ref{thm:maxagree}),
improving upon the long-standing state-of-the-art approximation ratio $0.7666$ due to \citet{Swamy04}.
The algorithm solves the standard SDP relaxation used by \citet*{Charikar+05} and \citet{Swamy04},
and rounds the resulting SDP solution using only the $k$-spokes rounding scheme~\cite{Frieze+97}.
Specifically, it uses the $4$-spokes and $5$-spokes roundings, and returns the better of the resulting clusterings.
Unlike the algorithms of \citet*{Charikar+05} and \citet{Swamy04}, our algorithm does not rely on the random hyperplane rounding.

In our analysis, we consider the algorithm that uses a mixture of the $4$-spokes and $5$-spokes roundings with probabilities $0.364$ and $0.636$, respectively.
Since our proposed algorithm performs at least as well as this mixture, it suffices to prove the desired approximation ratio of $0.7818$ for the mixture.
By our edge-wise analysis and its reduction to finite, exactly checkable conditions, this reduces to verifying the hypotheses for the positive-edge and negative-edge inequalities.
We verify the hypothesis for the positive-edge inequality in exact rational arithmetic using Bernstein certificates,
using the certified lower bound $\underline P_M(\rho)$ with truncation degree $M=24$.
The negative-edge inequality, on the other hand, can be shown analytically.

\subsubsection{\maxagreek}
\maxagreek\ is the variant of \maxagree\ in which the number of clusters is restricted to at most a fixed integer $K\geq 2$~\cite{Giotis+06_j}.
Again, its complementary problem, $\textsc{MinDisagree}[K]$, has also been studied,
and both problems are known to be NP-hard even on unweighted complete graphs~\cite{Shamir+04}.

\citet{Giotis+06_j} showed that for $K=2$,
the general weighted \maxagreek\ admits the Goemans--Williamson approximation ratio $\alpha_{\mathrm{GW}}\approx 0.8785$ known for Max Cut.
This guarantee is optimal under the UGC, since that optimality is known for Max Cut~\cite{Khot+07}, which is the special case where $E_+=\emptyset$ and $E_-=E$.
For $3\leq K\leq 5$, \citet{Swamy04} gave a $0.77$-approximation algorithm
that applies the $1$-hyperplane and $2$-hyperplane roundings to the solution of an SDP relaxation that differs slightly from the one used for \maxagree.
For $K\geq 6$, the aforementioned $0.7666$-approximation algorithm for \maxagree\ achieves the same approximation ratio,
since it always outputs a clustering with at most six clusters.
These approximation ratios for \maxagreek\ with $K\geq 3$, too, have remained the state of the art for more than two decades.

\paragraph{Our Results.}
We present a polynomial-time algorithm that improves the best known approximation ratios for $K=3,\dots,16$ (Theorem~\ref{thm:maxagree-k}).
The improvements achieved by our algorithm are substantial and considerably larger than the improvement obtained for \maxagree.
The most notable case is $K=3$, for which our algorithm achieves an approximation ratio of $0.8151$,
substantially improving upon the state-of-the-art ratio of $0.77$ due to \citet{Swamy04}.
Our algorithm solves the SDP relaxation used by \citet{Swamy04}, applies the $k$-spokes roundings for all $k=1,\dots,K$ to the SDP solution,
and returns the clustering with the largest objective value.
As in our analysis for \maxagree, for each $K=3,\dots, 16$, we analyze an appropriate mixture of these roundings
and verify the corresponding hypotheses for the positive-edge and negative-edge inequalities in exact rational arithmetic, using Bernstein certificates for the former.

\subsubsection{Max \texorpdfstring{$K$}{K}-Cut}
Max $K$-Cut is a classical problem in combinatorial optimization:
given a simple undirected graph $G=(V,E)$ with a weight $w_{ij}>0$ for each edge $\{i,j\}\in E$,
we are asked to find a clustering of $V$ into at most $K$ clusters that maximizes the total weight of edges across clusters.
Max $K$-Cut generalizes Max Cut, which corresponds to $K=2$, and can also be viewed as the special case of \maxagreek\ in which $E_+=\emptyset$ and $E_-=E$.
Note that Max $K$-Cut has a trivial expected approximation ratio of $1-1/K$,
since randomly partitioning $V$ into $K$ (possibly empty) clusters cuts $1-1/K$ fraction of the edges.

\citet*{Frieze+97} proposed the $K$-Cut algorithm, which uses their $K$-spokes rounding to the solution of a novel SDP relaxation,
and proved that the algorithm admits an approximation ratio of
\begin{equation*}
\alpha^\mathrm{FJ}_K \coloneqq \min_{\rho\in [-1/(K-1),1)}\frac{K(1-P_K(\rho))}{(K-1)(1-\rho)}.
\end{equation*}
While the exact value of $\alpha^\mathrm{FJ}_K$ is not trivial for a general $K$,
\citet{Frieze+97} derived explicit lower bounds on $\alpha^\mathrm{FJ}_K$ for several values of $K$.
For $K=2$, their algorithm and analysis reduce to those of \citet{Goemans+95} for Max Cut, so $\alpha^\mathrm{FJ}_2 = \alpha_{\mathrm{GW}}\approx 0.8785$.
\citet{Frieze+97} also showed
$\alpha^\mathrm{FJ}_3\geq 0.832718$, $\alpha^\mathrm{FJ}_4\geq 0.850304$,
$\alpha^\mathrm{FJ}_5\geq 0.874243$, $\alpha^\mathrm{FJ}_{10}\geq 0.926642$, and $\alpha^\mathrm{FJ}_{100}\geq 0.990625$,
based on the Hermite expansion of $P_K(\rho)$ and computations of low-degree coefficients.
Later, \citet{Goemans+04} and \citet*{deKlerk+04} further advanced the investigation.
For $K=3$, \citet{Goemans+04} obtained the exact value of $\alpha^{\mathrm{FJ}}_3 \approx 0.836008$ by analyzing an equivalent rounding scheme.
Their analysis indirectly yields a closed-form expression of $P_3(\rho)$ and shows that the minimum in the definition of $\alpha^\mathrm{FJ}_3$ is achieved at the left endpoint $\rho=-1/2$.
For $4\leq K\leq 10$, \citet*{deKlerk+04} obtained numerical estimates of $\alpha^\mathrm{FJ}_K$ and conjectured that the minimum is achieved at the left endpoint $\rho=-1/(K-1)$.
\citet{Frieze+97} also demonstrated the asymptotic bound
\begin{equation*}
\alpha^\mathrm{FJ}_K = 1-\frac{1}{K} + (1+o(1))2K^{-2}\log K,
\end{equation*}
which was later shown to be optimal under the UGC~\cite{Khot+07}.
\citet{Newman18} extended the results of \citet{Goemans+04} to general $K$ and gave explicit approximation ratios.

\paragraph{Our Results.}
The nature of our contribution here is different from those for the other problems:
we do not present a new algorithm, but significantly advance our understanding of the approximation guarantee of the $K$-Cut algorithm.
We first prove that for every $K\geq 3$, the minimum in the definition of $\alpha^\mathrm{FJ}_K$ is achieved at the left endpoint $\rho=-1/(K-1)$, yielding the much simpler representation $\alpha^\mathrm{FJ}_K = 1-P_K(-1/(K-1))$,
and thereby resolving affirmatively the conjecture posed by \citet*{deKlerk+04} (Theorem~\ref{thm:max-k-cut-approx-ratio}).
Using our monotonicity lemma for the separation probability, the proof is almost immediate, illustrating the power of this structural result.
Based on this characterization, we further prove that $\alpha^\mathrm{FJ}_K$ is tight for the $K$-Cut algorithm (Theorem~\ref{thm:max-k-cut-FJ-tightness}), meaning that no better approximation ratio can be certified for this algorithm.
Such tightness was previously known only for $K=2$~\cite{Karloff99} and $K=3$~\cite{Goemans+04}.

In view of this tightness result, it becomes particularly important to determine $\alpha^\mathrm{FJ}_K$ precisely.
Using the above characterization together with our Hermite-coefficient certification framework, we establish rigorous lower and upper bounds $\underline{\alpha}^\mathrm{FJ}_K$ and $\overline{\alpha}^\mathrm{FJ}_K$ for $K=3,\dots,16$ (Theorem~\ref{thm:max-k-cut}).
Both bounds reduce to exactly checkable inequalities involving certified bounds on the collision probability.
For every $K=3,\dots,16$, the gap between $\underline{\alpha}^\mathrm{FJ}_K$ and $\overline{\alpha}^\mathrm{FJ}_K$ is at most $10^{-6}$, thereby pinning down the tight approximation ratio $\alpha^\mathrm{FJ}_K$ to within $10^{-6}$.
In particular, our lower bounds improve upon the previous state-of-the-art explicit approximation ratios due to \citet{Frieze+97} and \citet{Newman18} for $K=4,\dots,16$.

\subsubsection{Modularity Maximization}
Community detection is a well-established network-analysis primitive whose goal is to find a partitioning of a graph that naturally divides it,
especially without specifying the number of clusters in advance.
The quality function known as modularity, introduced by \citet{Newman+04}, is perhaps the most widely-used measure for community detection~\cite{Fortunato10}.
For a given clustering of nodes, modularity quantifies the excess of the fraction of edges within clusters over the expected fraction under a random graph model with the same degree distribution.
Modularity maximization asks for a clustering maximizing the modularity value and is known to be NP-hard~\cite{Brandes+08}.

While most existing algorithms for modularity maximization are heuristic~\cite{Fortunato10}, a few works have developed algorithms with provable approximation guarantees.
\citet{DasGupta+13} gave a polynomial-time $\epsilon$-additive approximation algorithm for dense graphs, where $\epsilon>0$ is an arbitrary constant, using an algorithmic version of the regularity lemma~\cite{Frieze+96}.
That is, the algorithm returns a feasible solution whose objective value is at least the optimum minus $\epsilon$.
\citet*{Dinh+15} presented a polynomial-time $0.4672$-additive approximation algorithm applicable to arbitrary instances.
Their algorithm is essentially identical to the algorithm of \citet*{Charikar+05} for \maxagree.
Later, \citet*{Kawase+21} improved the analysis and showed that the additive approximation error is at most
$\cos\left(\frac{3-\sqrt{5}}{4}\pi\right) - \frac{1+\sqrt{5}}{8}<0.42084$.
The focus on additive approximation is motivated by a strong inapproximability result for multiplicative approximation:
for every constant $\alpha>0$, no polynomial-time $\alpha$-approximation algorithm for modularity maximization exists, unless $\mathrm{P}=\mathrm{NP}$~\cite{Dinh+15}.

\paragraph{Our Results.}
We present a polynomial-time $0.3790$-additive approximation algorithm for modularity maximization (Theorem~\ref{thm:modularity}), improving upon the state-of-the-art additive approximation error of $0.42084$ due to \citet*{Kawase+21}.
Our algorithm solves the SDP relaxation used by \citet*{Dinh+15} and \citet*{Kawase+21} and rounds the SDP solution using the $7$-spokes rounding rather than a hyperplane rounding.
Our analysis is not edge-wise, but certifying the additive approximation error reduces to verifying an inequality involving a certified lower bound on the collision probability of the $7$-spokes rounding, made possible by our Hermite-coefficient certification framework.
Finally, we show that the $7$-spokes rounding is nearly optimal, under this analysis, within the family of mixtures consisting of the $k$-spokes roundings for $1\leq k\leq 32$ and the $k$-hyperplane roundings for all $k\geq 1$.
Specifically, no such mixture can certify an additive approximation error smaller than $0.3785$ under the same analysis.

\subsection{Inapproximability of \maxagree\ under the UGC}
\label{subsec:raghavendra-overview}
It was shown by~\citet*{Charikar+05} that the weighted \maxagree\ is NP-hard to approximate within
$79/80+\varepsilon=0.9875+\varepsilon$
for every constant $\varepsilon>0$. \citet{note_inapp_maxagree}
extended this threshold to unweighted instances, assuming
$\mathrm{NP}\nsubseteq\mathrm{RP}$.  Under this assumption, this remains the best known
inapproximability bound.

A smaller upper bound is known for a very limited class of algorithms, namely those
based on rounding
the standard SDP relaxation and measuring their guarantee against the fractional solution.  Let $\operatorname{SDP}_{\mathrm{std}}$
denote the relaxation used in Section~\ref{sec:maxagree}, namely
$\operatorname{SDP}(0)$ in~\eqref{eq:maxagree_sdp}, which requires
$\langle v_i,v_j\rangle\geq0$ for every pair of nodes, and let
$\alpha_{\mathrm{SDP}}$ be its integrality ratio, the infimum of
$\operatorname{OPT}/\operatorname{SDP}_{\mathrm{std}}$ over all instances.
The previous best upper bound was
$\alpha_{\mathrm{SDP}}\leq2/(1+\sqrt2)\approx0.8284$, given by the
\textsf{++-} triangle of~\citet*{Charikar+05}.
However, this does not rule out the possibility of better approximations achieved by other means.

\paragraph{Our Results.}
We prove in Section~\ref{sec:raghavendra-gap}
 two complementary facts: a tighter integrality ratio bound for the standard relaxation $\operatorname{SDP}_{\mathrm{std}}$, and a unique-games hardness result for going
beyond this ratio.
To the best of our knowledge, the latter is the first UGC-based inapproximability result for \maxagree.

First, unconditionally, $\alpha_{\mathrm{SDP}}\leq0.802$.  Together with the
rounding of Section~\ref{sec:maxagree}, this places
$\alpha_{\mathrm{SDP}}$ in the interval $[0.7818,0.802]$, so our rounding
analysis is nearly tight.

Second, under the UGC of~\citet{ugc_khot}, it is NP-hard to approximate
weighted \maxagree\ within $\alpha_{\mathrm{SDP}}+\varepsilon$ for any
constant $\varepsilon>0$, and in particular within $0.802+\varepsilon$.
In other words, unless the UGC fails or $\mathrm{NP}\subseteq\mathrm{BPP}$,
$\operatorname{SDP}_{\mathrm{std}}$ is an optimal relaxation for \maxagree:
no polynomial-time algorithm, whether based on a relaxation or not, beats
its integrality ratio by a constant.
General theory
does not reveal the value of $\alpha_{\mathrm{SDP}}$: the gap-computation
algorithm of~\citet[Theorem~1.3]{RaghavendraSteurer09Round} can in principle
find an $\varepsilon$-additive approximation
thereof,
but its runtime of
$\exp(\exp(\poly(1/\varepsilon)))$ makes this computation infeasible.

Our approach is as follows.  We first reduce \maxagree\ to a CSP over a fixed alphabet, relate its
canonical SDP  in the sense of~\citet{Raghavendra08} to $\operatorname{SDP}_{\mathrm{std}}$, and then prove the
integrality-gap bound using a Gaussian construction.

\smallskip\noindent\emph{Reduction to a CSP.}
To work with CSPs
over a fixed finite alphabet, we
pass to $\maxagree[k]$,
its restriction to at most $k$ clusters.  This is a $2$-CSP over the
alphabet $[k]$: a positive edge has payoff
$\operatorname{EQ}(a,b)=\mathbf{1}[a=b]$, and a negative edge has payoff
$\operatorname{NEQ}(a,b)=\mathbf{1}[a\neq b]$.  Merging the clusters of any
clustering into $k$ random buckets keeps every positive agreement and, in
expectation, a $1-1/k$ fraction of the negative ones
(Lemma~\ref{lem:hash-clusters}).  Hence the optima of $\maxagree[k]$ and
\maxagree\ are within a factor $1-1/k$ of each other; as we explain next, so
are their SDPs.

\smallskip\noindent\emph{The canonical SDP.}
Viewed as a CSP, the canonical SDP
of $\maxagree[k]$, denoted $\operatorname{SDP}_k$, asks for two objects that must be consistent
with each other.  The first is a vector $u_{i,a}$ for every node $i$ and
label $a$; in an integral solution, $u_{i,a}=u_0$ if $i$ receives label $a$,
and $u_{i,a}=0$ otherwise.  The second is a \emph{local distribution}
$\mu_{ij}$, that is, a genuine joint distribution of the two labels
$(x_i,x_j)\in[k]^2$, for every pair $ij$ that carries a constraint; since Section~\ref{sec:raghavendra-gap} treats every
pair of nodes as a constraint, possibly of zero weight, this means every
pair.  The two are tied by
requiring that inner products equal local probabilities,
$\langle u_{i,a},u_{j,b}\rangle=\Pr_{\mu_{ij}}[x_i=a,\ x_j=b]$, and similarly
for the first moments $\langle u_{i,a},u_0\rangle$.  Thus the SDP must
exhibit, for every pair of nodes viewed in isolation, a genuine random
labeling of that pair, and a positive semidefinite Gram matrix glues these
local views together; but the views need not come from any global
distribution over labelings.

For \maxagree, the payoffs
only depend on whether two labels are equal.  After symmetrizing a solution over
all label permutations, the only relevant statistic of $\mu_{ij}$ is
therefore $y_{ij}=\Pr_{\mu_{ij}}[x_i=x_j]$, and the canonical SDP becomes, after some manipulations, the
program~\eqref{eq:reduced-basic-sdp}: $y_{ii}=1$, $y_{ij}\geq0$, and
$kY-J_n\succeq0$, where $J_n$ is the all-ones matrix.  In this form, the
constraint $\langle v_i,v_j\rangle\geq0$ of
$\operatorname{SDP}_{\mathrm{std}}$ is precisely the requirement that the pair
$ij$ admit a local distribution.  The only difference between the two
programs is the constraint $kY-J_n\succeq0$ in place of $Y\succeq0$.  It can be thought of as
the pigeonhole principle in SDP form: it rules out, for instance, $k+1$
nodes with pairwise $y_{ij}=0$, since $k$ labels cannot keep $k+1$
nodes pairwise apart.  Its effect vanishes as $k\to\infty$:
Lemma~\ref{lem:basic-sdp-sandwich} shows that
\[
  \left(1-\frac1k\right)\operatorname{SDP}_{\mathrm{std}}
  \leq\operatorname{SDP}_k
  \leq\operatorname{SDP}_{\mathrm{std}},
\]
so $\operatorname{SDP}_{\mathrm{std}}$ is the limit of the canonical SDPs of
the problems $\maxagree[k]$.

Raghavendra's theorem~\citep{Raghavendra08}, together with this comparison
and the reduction to $k$ clusters, gives the UGC-hardness statement of
Theorem~\ref{thm:maxagree-sdp-gap-hardness}; we prove this in
Section~\ref{subsec:ugc-gap-transfer}.

\smallskip\noindent\emph{The Gaussian integrality gap.}
It remains to prove the unconditional bound
$\alpha_{\mathrm{SDP}}\leq0.802$ in
Theorem~\ref{thm:maxagree-sdp-gap-hardness}.
Our integrality gap is built from an auxiliary instance in Gaussian space.
Its basic pattern is that of the Max-Cut gap of~\citet{FeigeSchechtman02}:
the SDP vectors are essentially the points of the space, and an
isoperimetric inequality bounds every integral solution.  Two difficulties
absent from Max-Cut prevent a direct adaptation.  First, a cut has only two
sides, whereas a clustering can have any number of cells of unequal
measures, and the most noise-stable such partitions are not known; natural
candidates are even known not to be optimal~\citep{HeilmanMosselNeeman16}.
Second, the Goemans--Williamson SDP only asks for unit vectors, whereas
$\operatorname{SDP}_{\mathrm{std}}$ requires nonnegative inner products on
every pair of nodes, including pairs of zero weight, and the natural
vectors violate this on half of the independent pairs.

The auxiliary instance has node set $\mathbb R^d$ and a constraint
distribution that puts negative edges between independent Gaussians and
positive edges between $(3/4)$-correlated Gaussians.  This can be formalized
by two measures on $\mathbb R^{2d}$, one for each edge type, with total mass
one.  The negative constraints encourage shrinking the Gaussian volume of the
clusters.  The agreement of any partition of Gaussian space with the positive
constraints is determined by Gaussian noise stability.  Borell's isoperimetric
inequality~\citep{Borell85} asserts that among sets of a prescribed measure,
halfspaces maximize the noise stability.  To get around the first
difficulty, we apply it to each cell separately, which is lossy but avoids
the unknown optimal partitions.  Combining this with a quadratic
upper bound on the noise stability of halfspaces, and choosing the
weights of the two edge types so as to cancel the term from the negative
agreements that depends on the cluster measures, we bound the agreement of
every clustering by $1203/1750\approx0.6874$.  This bound concerns integral
solutions only, so it does not depend on the relaxation.

On the fractional side, the natural unit vectors $r_x=x/\|x\|$ suggest the
value $6/7$, but they meet the second difficulty.  We replace them by the
shifted unit vectors
$k^{-1/2}e_0+(1-1/k)^{1/2}r_x$, with $e_0\perp\mathbb R^d$, whose inner
products $\frac1k+\bigl(1-\frac1k\bigr)\langle r_x,r_y\rangle$ are negative
only on a set of pairs of negligible total constraint weight.  We remove
the edge weights on these pairs.  After discretization, this gives an
explicit instance $J$ on which the weaker
relaxation that imposes $\langle v_i,v_j\rangle\geq0$ only on pairs carrying
weight has value close to $6/7$ for large $k$, while every clustering still
has agreement at most about $0.6874$.  For this weaker relaxation, $J$ is thus
an explicit instance with integrality ratio close to $0.802$.

For $\operatorname{SDP}_{\mathrm{std}}$ itself, $J$ does not suffice to establish the gap:
$\operatorname{SDP}_{\mathrm{std}}$ imposes $\langle v_i,v_j\rangle\geq0$ on
every pair of nodes, including those whose constraints were removed.
The instance that
witnesses $\alpha_{\mathrm{SDP}}\leq0.802$ is therefore a different one,
produced from $J$ by the strong-gap theorem
of~\citet{RaghavendraSteurer09Gaps}.  This theorem converts any gap instance
of the weaker relaxation (Raghavendra's canonical SDP, which asks
for local distributions only on the pairs of positive weight) into another
instance that preserves the bounds on the SDP value and integral optimum,
up to arbitrarily small additive errors, even for relaxations with consistent
local distributions on all sets of boundedly many nodes, including $\operatorname{SDP}_{\mathrm{std}}$.  This detour is the
only step of the integrality gap that relies on Raghavendra's work.

For the numerical hardness bound, the auxiliary instance $J$ already
suffices, by the gap-to-hardness transfer described above.  Under the UGC,
this hardness bound already implies $\alpha_{\mathrm{SDP}}\leq0.802$: since
$\operatorname{SDP}_{\mathrm{std}}$ can be solved in polynomial time, its
integrality ratio cannot exceed a UG-hardness threshold unless
$\mathrm P=\mathrm{NP}$.  The strong-gap theorem is needed only to remove this
assumption; in effect, it replaces the conjectured hardness of Unique Games
by a Unique Games instance whose low value is proven.

\subsection{Limitations of Rounding Families under the Edge-Wise Analysis}

The improved approximation ratios for \maxagree\ and \maxagreek\ above naturally raise the question of how much further one can go by choosing better mixtures of rounding schemes while retaining the same edge-wise analysis. We address this question by developing a framework that, for any prescribed finite family of rounding schemes, certifies an upper bound on the best approximation ratio achievable by mixtures over that family within the edge-wise analysis.

Our framework is based on LP duality and certified bounds on collision probabilities. The best ratio certifiable by mixtures over a given rounding family can be formulated as a semi-infinite LP. We obtain a finite relaxation by restricting the edge-wise inequalities to finitely many rational test points and replacing the collision probabilities by certified lower and upper bounds. A rational feasible solution to the dual of this finite relaxed LP then yields a rigorous upper bound on the optimal value of the original semi-infinite LP. The final certificate is checked entirely in exact rational arithmetic.

For \maxagree, we apply this framework to the family consisting of the $k$-spokes roundings for $1\leq k\leq 32$ and the $k$-hyperplane roundings for $1\leq k\leq 5$. We show that no mixture over this family can certify an approximation ratio larger than $0.78187$ under the edge-wise analysis. Since our algorithm achieves $0.7818$, this leaves only tiny room for improvement within this family and analysis framework. We also apply the framework to the subfamily underlying Swamy's algorithm, consisting of the $6$-spokes and $2$-hyperplane roundings, and obtain an upper bound of $0.77222$. Thus, within the same framework, attaining our ratio of $0.7818$, or any comparable improvement over Swamy's guarantee, requires moving beyond Swamy's subfamily.

For \maxagreek, for every $K=3,\ldots,16$, we consider the family of all $k$-spokes roundings with $1\leq k\leq K$, precisely the family used by our algorithm. The resulting certified upper bounds are very close to the corresponding approximation ratios of our algorithm, showing that these guarantees are also nearly optimal among mixtures over this family within the edge-wise analysis framework.

\section{Hermite-Coefficient Certification Framework}\label{sec:theory}

In this section, we develop our Hermite-coefficient certification framework.
We first formulate the collision probability of the $k$-spokes rounding in Gaussian space and derive its Hermite expansion.
We then establish a structural property of the associated separation probability,
before showing how to obtain certified bounds on finitely many coefficients and use them to derive rigorous lower and upper bounds on the entire collision-probability function.

\subsection{Preliminaries}

We recall a few standard facts about
Gaussian space and Hermite expansions; see \cite[Chapter~11]{o2014analysis}.
We write $\mathbb{N}_0=\{0,1,2,\ldots\}$.
Let $\gamma_d=\mathcal{N}(\bm{0}_d,I_d)$ denote the standard Gaussian
probability measure on $\mathbb{R}^d$, with density
\[\varphi_d(x) = (2\pi)^{-d/2}\exp\left(-\frac{\|x\|_2^2}{2}\right).\]
We write $\varphi=\varphi_1$ in dimension one.

Let $L^2(\gamma_d)$ be the space of square-integrable functions with respect to $\gamma_d$.
For $f,g\in L^2(\gamma_d)$, write
$\langle f,g\rangle
=
\E_{Z\sim \mathcal{N}(\bm{0}_d,I_d)}[f(Z)g(Z)],$
and let $\|f\|_2^2=\langle f,f\rangle$.

We will repeatedly use the following elementary property of Gaussian linear
forms.  If $G\sim\mathcal{N}(\bm{0}_d,I_d)$ and $v,w\in\mathbb{R}^d$, then
$v\cdot G$ and $w\cdot G$ are centered, jointly Gaussian random variables with
\[
\Var(v\cdot G)=\|v\|_2^2
\quad \text{and}\quad
\Cov(v\cdot G,w\cdot G)=v\cdot w.
\]
In particular, if $v$ and $w$ are unit vectors with $v\cdot w=\rho$, then
$(v\cdot G,w\cdot G)$ is a pair of $\rho$-correlated standard Gaussians.
Independent copies $G_1,\ldots,G_k$ of $G$ yield independent pairs
$(v\cdot G_p,w\cdot G_p)$, $p\in[k]$, since each pair is a function of
$G_p$ alone.

More generally, for $\rho\in[-1,1]$, we say that
$Z,Z'\in\mathbb{R}^d$ are $\rho$-correlated standard Gaussian vectors if each
coordinate pair $(Z_r,Z'_r)$ is a pair of $\rho$-correlated standard Gaussians
and the coordinate pairs are mutually independent
\cite[Definition~11.10]{o2014analysis}.  Equivalently, if $Z$ is standard
Gaussian, then one may sample $Z'$ as
$Z'=\rho Z+\sqrt{1-\rho^2}\,G,$
where $G\sim\mathcal{N}(\bm{0}_d,I_d)$ is independent of $Z$.
We denote the joint distribution of $(Z,Z')$ by $\nu_\rho$, leaving the
dimension implicit.
For $\rho\in[-1,1]$, the Gaussian noise operator $\mathrm{U}_\rho$
\cite[Definition~11.12]{o2014analysis} is defined by
\begin{equation*}%
\mathrm{U}_\rho f(z)
=
\E_{G\sim\gamma_d}
\left[
f\left(\rho z+\sqrt{1-\rho^2}\,G\right)
\right]
=
\E_{(Z,Z')\sim\nu_\rho}\left[f(Z')\mid Z=z\right].
\end{equation*}

We also recall the Hermite expansion in Gaussian space. We can express the derivatives of any order of the density function $\varphi$ as a polynomial multiple of $\varphi$ itself.
In dimension 1, the probabilist's (monic) Hermite polynomial
$H_j\colon \mathbb{R}\to\mathbb{R}$ of degree $j$ is defined by~\cite[Equation~(11.10)]{o2014analysis}:
\begin{equation}\label{eq:hermite-derivative-definition}
H_j(z)=\frac{(-1)^j}{\varphi(z)}
\frac{\mathrm{d}^j}{\mathrm{d}z^j}\varphi(z),
\quad z\in\mathbb{R},\quad j\in\mathbb{N}_0.
\end{equation}
These polynomials satisfy $H_0(z)=1$, $H_1(z)=z$, and the
recurrence~\cite[Exercise~11.10(b)]{o2014analysis}
\begin{equation}\label{eq:hermite-recurrence}
H_{j+1}(z)=zH_j(z)-jH_{j-1}(z),\qquad j\ge1.
\end{equation}
In particular, every $H_j$ has integer coefficients.
We use the normalized Hermite polynomials
$h_j(z)=\frac{H_j(z)}{\sqrt{j!}}$.
For example,
$h_0(z)=1$, $h_1(z)=z$, and $h_2(z)=\frac{z^2-1}{\sqrt{2}}$.
These polynomials are obtained by applying
Gram--Schmidt orthogonalization to the monomials
$\{z^j\}_{j=0}^{\infty}$ under the Gaussian inner product~\cite[Sec. 11.2]{o2014analysis}.

For a multi-index
$\alpha=(\alpha_1,\ldots,\alpha_d)\in\mathbb{N}_0^d$ and $z=(z_1,\dots, z_d)\in \mathbb{R}^d$, define
\begin{equation*}%
h_\alpha(z)=\prod_{r=1}^d h_{\alpha_r}(z_r)
\quad \text{and}\quad
|\alpha|=\sum_{r=1}^d \alpha_r.
\end{equation*}
The polynomials $(h_\alpha)_{\alpha\in\mathbb{N}_0^d}$ form a complete
orthonormal basis for $L^2(\gamma_d)$
\cite[Propositions~11.30 and~11.33]{o2014analysis}.  Hence every
$f\in L^2(\gamma_d)$ admits a unique Hermite expansion
\begin{equation}\label{eq:hermite-expansion}
f=\sum_{\alpha\in\mathbb{N}_0^d}\hat f(\alpha)h_\alpha
\quad \text{with}\quad
\hat f(\alpha)=\langle f,h_\alpha\rangle,
\end{equation}
where the series converges in $L^2(\gamma_d)$
\cite[Definition~11.34]{o2014analysis}.
For $m\in\mathbb{N}_0$, let
\begin{equation*}%
f^{(m)}
=
\sum_{|\alpha|=m}\hat f(\alpha)h_\alpha
\end{equation*}
be the degree-$m$ Hermite component of $f$, and define the level-$m$ weight
\begin{equation}\label{eq:hermite-level-weight}
W_m(f)
=
\|f^{(m)}\|_2^2
=
\sum_{|\alpha|=m}\hat f(\alpha)^2.
\end{equation}
The Hermite basis diagonalizes the Gaussian noise operator, giving the
noise-stability identity
\begin{equation}\label{eq:hermite-noise-stability}
\E_{(Z,Z')\sim\nu_\rho}[f(Z)f(Z')]
=
\sum_{\alpha\in\mathbb{N}_0^d}\rho^{|\alpha|}\hat f(\alpha)^2
=
\sum_{m=0}^{\infty}\rho^m W_m(f)
\end{equation}
for every $f\in L^2(\gamma_d)$
\cite[Proposition~11.37]{o2014analysis}.

\subsection{Frieze--Jerrum Rounding and Collision Probabilities}

Let $V=\{1,\dots, n\}$. Let $\bm{v}_1,\dots,\bm{v}_n \in\mathbb{R}^n$ be unit vectors.
(The ambient dimension may be taken to be $n$ without loss of generality:
any $n$ vectors span a space of dimension at most $n$, and an isometric
embedding into $\mathbb{R}^n$ preserves their inner products.) For an integer $k\geq 1$, the $k$-spokes rounding generates $k$ independent Gaussian vectors $\bm{g}_1,\dots,\bm{g}_k \sim \mathcal{N}(\bm{0}_n,I_n)$,
and for each $p=1,\dots, k$, constructs a cluster $C_p\leftarrow \{i\in V: p\in \argmax_{q\in [k]}\bm{v}_i\cdot \bm{g}_{q}\}$.
Ties occur with probability zero and may be broken arbitrarily.
Finally, it returns a clustering $\mathcal{C}=\{C_1,\dots, C_k\}$ of $V$.
For reference, this procedure is summarized in Algorithm~\ref{alg:k-spokes}.
For $k=1$, this is the trivial rounding that places all nodes in a single cluster, and hence the collision probability is equal to $1$.
\begin{algorithm}[h]
\SetKwInput{Input}{Input}
\caption{$k\text{-spokes}((\bm{v}_i)_{i\in V})$}\label{alg:k-spokes}
\Input{$\bm{v}_1,\dots, \bm{v}_n\in \mathbb{R}^n$}
Generate $k$ independent Gaussian vectors $\bm{g}_1,\dots,\bm{g}_k\sim \mathcal{N}(\bm{0}_n,I_n)$\;
\textbf{for} $p=1,\dots, k$ \textbf{do} $C_p\leftarrow \{i\in V: p\in \argmax_{q\in [k]}\bm{v}_i\cdot \bm{g}_{q}\}$\;
\Return $\mathcal{C}=\{C_1,\dots, C_k\}$.
\end{algorithm}

Fix $i,j\in V$ and let $\rho=\bm{v}_i\cdot \bm{v}_j$.
For each $p\in[k]$, define
\begin{equation}\label{eq:spoke-linear-forms}
L_p=\bm{v}_i\cdot \bm{g}_p
\quad \text{and}\quad
R_p=\bm{v}_j\cdot \bm{g}_p.
\end{equation}
By the Gaussian linear-form facts recalled above, $L_p$ and $R_p$ are
standard Gaussian random variables and
$\Cov(L_p,R_p)=\bm{v}_i\cdot \bm{v}_j=\rho$.
Thus $(L_p,R_p)$ is a pair of $\rho$-correlated standard Gaussians.  Moreover,
since $\bm{g}_1,\ldots,\bm{g}_k$ are independent, the pairs
$(L_1,R_1),\ldots,(L_k,R_k)$
are independent.  Therefore the probability that $i$ and $j$ are assigned to
the same cluster by $k$-spokes depends only on $\rho$.  We denote this
probability by
\begin{equation}\label{eq:Pk-def}
P_k(\rho)
=
\Pr\left[
\argmax_{p\in[k]} L_p
=
\argmax_{p\in[k]} R_p
\right].
\end{equation}
By symmetry of the spokes, the probability that both vertices select a
given spoke is the same for all $p\in[k]$.

\subsection{Hermite Expansion of Collision Probabilities}\label{subsec:hermite}

Let $L=(L_1,\ldots,L_k)$ and $R=(R_1,\ldots,R_k)$ as in \eqref{eq:spoke-linear-forms}.  As observed above, $L$
and $R$ are $\rho$-correlated standard Gaussian vectors in $\mathbb{R}^k$.
For each $p\in[k]$, define
\[
A_p=\left\{x\in\mathbb{R}^k: x_p=\max_{q\in[k]}x_q\right\},
\]
and let $f_p=\1_{A_p}$.  The sets $A_p$ overlap only on Gaussian measure-zero
boundaries, so by \eqref{eq:Pk-def} and symmetry,
\begin{equation*}
\begin{aligned}
P_k(\rho)
= \sum_{p=1}^k \Pr[L\in A_p \text{ and } R\in A_p]
= \sum_{p=1}^k \E[f_p(L)f_p(R)]
 = k\E[f_1(L)f_1(R)].
\end{aligned}
\end{equation*}
Each $A_p$ is an intersection of halfspaces, hence measurable, and
$f_p\in L^2(\gamma_k)$.
Applying the noise-stability identity \eqref{eq:hermite-noise-stability} to $f_p$, we get
\begin{equation*}%
\E[f_p(L)f_p(R)]
=
\sum_{m=0}^{\infty}\rho^m W_m(f_p).
\end{equation*}
Therefore
\begin{equation}\label{eq:Pk-hermite-series}
P_k(\rho)
=
\sum_{p=1}^k\sum_{m=0}^{\infty}\rho^m W_m(f_p)
=
\sum_{m=0}^{\infty}c_{k,m}\rho^m,
\end{equation}
where, for every $m\in\mathbb{N}_0$,
\begin{equation}\label{eq:ckm-def}
c_{k,m}
=
\sum_{p=1}^k W_m(f_p)
=
kW_m(f_1)
=
k\sum_{|\alpha|=m}\hat f_1(\alpha)^2
\ge0.
\end{equation}
Here the level weights $W_m(f_p)$ are equal because coordinate
permutations map the functions $f_p$ to one another and preserve Hermite degree.

At $\rho=0$, the Gaussian vectors $L$ and $R$ are independent.  By symmetry,
both $\argmax_{p\in[k]}L_p$ and $\argmax_{p\in[k]}R_p$ are uniformly
distributed on $[k]$, and hence
\begin{equation}\label{eq:Pk-zero}
P_k(0)=c_{k,0}=\frac1k.
\end{equation}
At $\rho=1$, we have $L=R$, so the two argmax indices agree almost surely.
Thus
\begin{equation}\label{eq:Pk-one}
P_k(1)=\sum_{m=0}^{\infty}c_{k,m}=1.
\end{equation}

\subsection{A Structural Property of Separation Probabilities}\label{subsec:theory-structure}

We call $1-P_k(\rho)$ the separation probability.
Using the Hermite representation of the collision probability,
we prove the following structural property:
\begin{lemma}\label{lem:negative-ratio-monotonicity}
Fix $k\geq 1$. Define
\begin{equation*}
R_k(\rho)=\frac{1-P_k(\rho)}{1-\rho}\quad \text{for }\rho\in [-1,1).
\end{equation*}
Then, $R_k$ is nondecreasing on $\rho\in [-1/2,1)$.
\end{lemma}

\begin{proof}
Recall from equations~\eqref{eq:Pk-hermite-series} and \eqref{eq:Pk-one} that
$P_k(\rho)=\sum_{m=0}^\infty c_{k,m}\rho^m$ and $P_k(1)=\sum_{m=0}^\infty c_{k,m}=1$.
Hence,
\begin{equation*}
R_k(\rho)
=\sum_{m=0}^\infty c_{k,m}\frac{1-\rho^m}{1-\rho}
=\sum_{m=1}^\infty c_{k,m}\frac{1-\rho^m}{1-\rho}
=\sum_{m=1}^{\infty}c_{k,m}q_m(\rho),
\end{equation*}
where we put $q_m(\rho)=\frac{1-\rho^m}{1-\rho}
=\sum_{j=0}^{m-1}\rho^j$ for $m\geq1$.

The derivative of $q_m$ is $q'_m(\rho)=\sum_{j=1}^{m-1}j\rho^{j-1}$.
For $\rho\geq-1/2$, pairing consecutive terms gives
\begin{equation*}
(2j-1)\rho^{2j-2}+2j\rho^{2j-1}
=\rho^{2j-2}\bigl((2j-1)+2j\rho\bigr)
\geq (j-1)\rho^{2j-2}\geq0.
\end{equation*}
For odd $m$, this pairs all terms; for even $m$, the remaining term
$(m-1)\rho^{m-2}$ is also nonnegative.
Thus every $q_m$ is nondecreasing on $[-1/2,1)$.
Since $c_{k,m}\geq0$, so is every partial sum of
$\sum_{m\geq1}c_{k,m}q_m$, and hence its pointwise limit $R_k$.
\end{proof}

This lemma enables us to reduce the negative-edge inequality arising in the edge-wise analysis formalized in Section~\ref{subsec:approximations-general} from an entire interval of correlations to a single endpoint condition, significantly simplifying the analysis.
It also leads almost immediately to a proof of a long-standing conjecture concerning the FJ rounding for Max $K$-Cut (Theorem~\ref{thm:max-k-cut-approx-ratio}).

Lemma~\ref{lem:negative-ratio-monotonicity} applies directly to any mixture of the $k$-spokes roundings:
\begin{corollary}\label{cor:negative-ratio-monotonicity}
Let $\mathcal{K}$ be a set of positive integers.
Consider the mixture of the $k$-spokes roundings with probabilities $\bm{\lambda}=\left(\lambda_k\right)_{k\in \mathcal{K}}$,
where $\lambda_k \geq 0$ and $\sum_{k\in \mathcal{K}} \lambda_k =1$.
For two unit vectors $\bm{v}_i$ and $\bm{v}_j$ with $\bm{v}_i\cdot \bm{v}_j=\rho$,
let $P(\rho)$ be the probability that $i$ and $j$ are assigned to the same cluster by this mixture.
Define
\begin{equation*}
R(\rho)=\frac{1-P(\rho)}{1-\rho}\quad \text{for }\rho\in [-1,1).
\end{equation*}
Then, $R$ is nondecreasing on $\rho\in [-1/2,1)$.
\end{corollary}

\subsection{Certified Bounds on the Power-Series Coefficients}\label{subsec:theory-one-dimensional}

Here we outline the main ideas behind obtaining certified bounds on the coefficients $c_{k,m}$.
The complete derivation and certification procedure are given in Section~\ref{sec:details}.

A priori, each coefficient is a sum of squares of $k$-dimensional Gaussian integrals over
argmax regions. By symmetry, it suffices to consider the region where coordinate
$1$ is largest.  There the indicator factors as
$\prod_{q=2}^k \1[Z_q\le Z_1]$, so conditioning on $Z_1=x$ separates the
remaining coordinates. Integrating~\eqref{eq:hermite-derivative-definition}
gives the identity

\[
  \int_{-\infty}^x \varphi(y)h_a(y)\,\mathrm{d}y
  =
  \begin{cases}
  \Phi(x), & a=0,\\[2mm]
  -\frac{\varphi(x)h_{a-1}(x)}{\sqrt a}, & a\ge1,
  \end{cases}
\]
so each original
$k$-dimensional Hermite coefficient becomes a single signed one-dimensional
integral with an explicit integrand. Here $\Phi$ denotes the standard Gaussian
cumulative distribution function.

For the exact certificate, it is convenient to clear the square-root normalizations and work instead with unnormalized probabilists' Hermite polynomials. Let $H_a$ denote the unnormalized probabilists' Hermite polynomial, set $K_0(x)=\Phi(x)$, and set $K_b(x)=-\varphi(x)H_{b-1}(x)$ for $b\ge1$.
For a
multiset $\lambda$ of $k-1$ nonnegative integers, write $|\lambda|$ for the
sum of its elements and
\[
N(\lambda)=\frac{(k-1)!}{\prod_{b\ge0} n_b(\lambda)!},
\]
where $n_b(\lambda)$ is the multiplicity of $b$ in $\lambda$.
Products indexed by $b\in\lambda$ are understood to count elements with multiplicity.
Define
\begin{equation}\label{eq:main-Balambda-def}
B_{a,\lambda}
=
\int_{-\infty}^{\infty}
\varphi(x)H_a(x)\prod_{b\in \lambda}K_b(x)\,\mathrm{d}x.
\end{equation}
The core certificate statement is the following.
The first part is proved in
Section~\ref{subsec:details-normalization-symmetry},
while the second part is proved in Section~\ref{subsec:details-assembling}.

\begin{proposition}\label{prop:main-one-dimensional-certificate}
For every $k\ge2$ and $m\in\mathbb{N}_0$,
\begin{equation*}%
c_{k,m}
=
k
\sum_{a=0}^m
\sum_{\substack{\lambda:\ |\lambda|=m-a\\
\lambda\text{ is a multiset of size }k-1}}
N(\lambda)\,
\frac{B_{a,\lambda}^2}
{a!\prod_{b\in \lambda} b!}.
\end{equation*}
Moreover, suppose that for every pair $(a,\lambda)$ in the sum we have a
certified rational interval $I_{a,\lambda}\ni B_{a,\lambda}$, and define
$\underline b_{a,\lambda}=\min_{x\in I_{a,\lambda}}x^2$ and $\overline b_{a,\lambda}=\max_{x\in I_{a,\lambda}}x^2$.
Then
\[
\underline c_{k,m}
=
k
\sum_{a=0}^m
\sum_{\substack{\lambda:\ |\lambda|=m-a\\
\lambda\text{ is a multiset of size }k-1}}
N(\lambda)\,
\frac{\underline b_{a,\lambda}}
{a!\prod_{b\in \lambda} b!}
\]
and
\[
\overline c_{k,m}
=
k
\sum_{a=0}^m
\sum_{\substack{\lambda:\ |\lambda|=m-a\\
\lambda\text{ is a multiset of size }k-1}}
N(\lambda)\,
\frac{\overline b_{a,\lambda}}
{a!\prod_{b\in \lambda} b!}
\]
are certified rational lower and upper bounds, respectively, on $c_{k,m}$.
\end{proposition}

The integrals in
\eqref{eq:main-Balambda-def} have an explicit form.  If $r(\lambda)$ is
the number of positive elements of $\lambda$ and $s(\lambda)$ is the number of
zero elements, then
\begin{equation*}%
B_{a,\lambda}
=
\int_{-\infty}^{\infty}
(-1)^{r(\lambda)}
\varphi(x)^{r(\lambda)+1}\Phi(x)^{s(\lambda)}S_{a,\lambda}(x)\,\mathrm{d}x,
\end{equation*}
where $S_{a,\lambda}$ is an explicit integer polynomial.
The integrals are signed, and the coefficient formula involves their
squares, so simply truncating the tails yields neither a valid lower bound
nor a valid upper bound in general.
Instead, for each $B_{a,\lambda}$ we certify an interval
\[
I_{a,\lambda}
=
[L_{a,\lambda}-\varepsilon_{a,\lambda},
 U_{a,\lambda}+\varepsilon_{a,\lambda}],
\]
where $[L_{a,\lambda},U_{a,\lambda}]$ encloses the finite-window integral and
$\varepsilon_{a,\lambda}$ is an explicit bound on the two tails.  Proposition~\ref{prop:main-one-dimensional-certificate}
then turns these interval enclosures into rational lower and upper bounds $\underline c_{k,m}\le c_{k,m}\le \overline c_{k,m}$.

\subsection{Certified Bounds on Collision Probabilities}

Finally, we present lower and upper bounds on the collision probability $P_k(\rho)$ obtainable from finitely many certified lower and upper bounds on $c_{k,m}$.

\begin{proposition}\label{prop:theory-Pk-LB-UB}
Fix $k\geq 1$ and $M\in\mathbb{N}_0$.
Suppose that for each $m=0,\dots,M$, we have certified lower and upper bounds $\underline c_{k,m}\leq c_{k,m}\leq \overline c_{k,m}$
(for $k=1$, since $P_1\equiv 1$, the exact values $c_{1,0}=1$ and $c_{1,m}=0$ for $m\geq 1$ may be used).
Define
\begin{equation*}
\underline P_{k,M}(\rho)
=
\begin{cases}
\displaystyle
\sum_{m=0}^{M}\underline c_{k,m}\rho^m,
& \rho\in[0,1],
\\[2mm]
\displaystyle
\sum_{\substack{0\leq m\leq M\\ m\ \mathrm{even}}}
\underline c_{k,m}\rho^m
+
\sum_{\substack{0\leq m\leq M\\ m\ \mathrm{odd}}}
\overline c_{k,m}\rho^m
-
|\rho|^{M+1}
\left(
1-\sum_{m=0}^{M}\underline c_{k,m}
\right),
& \rho\in[-1,0],
\end{cases}
\end{equation*}
and
\begin{equation*}
\overline P_{k,M}(\rho)
=
\begin{cases}
\displaystyle
\sum_{m=0}^{M}\overline c_{k,m}\rho^m
+
\rho^{M+1}
\left(
1-\sum_{m=0}^{M}\underline c_{k,m}
\right),
& \rho\in[0,1],
\\[2mm]
\displaystyle
\sum_{\substack{0\leq m\leq M\\ m\ \mathrm{even}}}
\overline c_{k,m}\rho^m
+
\sum_{\substack{0\leq m\leq M\\ m\ \mathrm{odd}}}
\underline c_{k,m}\rho^m
+
|\rho|^{M+1}
\left(
1-\sum_{m=0}^{M}\underline c_{k,m}
\right),
& \rho\in[-1,0].
\end{cases}
\end{equation*}
Then, it holds that
\begin{equation*}
\underline P_{k,M}(\rho)
\leq
P_k(\rho)
\leq
\overline P_{k,M}(\rho)
\quad
\text{for all $\rho\in[-1,1]$}.
\end{equation*}
\end{proposition}

\begin{proof}
First, let $\rho\in[0,1]$.
Using the nonnegativity of $c_{k,m}$ and $\underline c_{k,m}\leq c_{k,m}$, we have
\begin{equation*}
P_k(\rho)
\geq
\sum_{m=0}^{M}c_{k,m}\rho^m
\geq
\sum_{m=0}^{M}\underline c_{k,m}\rho^m
=
\underline P_{k,M}(\rho).
\end{equation*}
Moreover, using $\rho\leq 1$, the nonnegativity of $c_{k,m}$, and $\sum_{m=0}^{\infty}c_{k,m}=1$, we have
\begin{equation*}
\sum_{m=M+1}^{\infty}c_{k,m}\rho^m
\leq
\rho^{M+1}\sum_{m=M+1}^{\infty}c_{k,m}
=
\rho^{M+1}
\left(
1-\sum_{m=0}^{M}c_{k,m}
\right)
\leq
\rho^{M+1}
\left(
1-\sum_{m=0}^{M}\underline c_{k,m}
\right).
\end{equation*}
Together with $c_{k,m}\leq\overline c_{k,m}$, this gives $P_k(\rho) \leq \overline P_{k,M}(\rho)$ for $\rho\in [0,1]$.

Next, let $\rho\in[-1,0]$.
For every even $m\leq M$, we have
$\underline c_{k,m}\rho^m \leq c_{k,m}\rho^m \leq \overline c_{k,m}\rho^m$.
For every odd $m\leq M$, since $\rho\leq 0$, we have
$\overline c_{k,m}\rho^m \leq c_{k,m}\rho^m \leq \underline c_{k,m}\rho^m$.
Moreover, using $|\rho|\leq 1$, the nonnegativity of $c_{k,m}$, and $\sum_{m=0}^{\infty}c_{k,m}=1$, we have
\begin{equation*}
\left|
\sum_{m=M+1}^{\infty}c_{k,m}\rho^m
\right|
\leq
|\rho|^{M+1}\sum_{m=M+1}^{\infty}c_{k,m}
=
|\rho|^{M+1}
\left(
1-\sum_{m=0}^{M}c_{k,m}
\right)
\leq
|\rho|^{M+1}
\left(
1-\sum_{m=0}^{M}\underline c_{k,m}
\right).
\end{equation*}
Hence, $\underline P_{k,M}(\rho) \leq P_k(\rho) \leq \overline P_{k,M}(\rho)$ for $\rho \in [-1,0]$.
\end{proof}

Note that Proposition~\ref{prop:theory-Pk-LB-UB} applies directly to any mixture of the $k$-spokes roundings:
\begin{corollary}\label{cor:theory-P-LB-UB}
Fix $M\in \mathbb{N}_0$.
Consider the mixture of the $k$-spokes roundings as in Corollary~\ref{cor:negative-ratio-monotonicity}.
Suppose that for each $k\in \mathcal{K}$ and each $m=0,\dots, M$, we have certified lower and upper bounds $\underline{c}_{k,m}\leq c_{k,m}\leq \overline{c}_{k,m}$.
Recall $\underline P_{k,M}(\rho)$ and $\overline P_{k,M}(\rho)$ defined in Proposition~\ref{prop:theory-Pk-LB-UB}.
Define
\begin{equation*}
\underline P_M(\rho) =\sum_{k\in\mathcal K}\lambda_k\underline P_{k,M}(\rho)\quad
\text{and}\quad \overline P_M(\rho) =\sum_{k\in\mathcal K}\lambda_k\overline P_{k,M}(\rho).
\end{equation*}
Then
\begin{equation*}
\underline P_M(\rho)\le P(\rho)\le \overline P_M(\rho) \quad\text{for all } \rho\in[-1,1].
\end{equation*}
\end{corollary}

\section{Derivation of the Certified Bounds on the Power-Series Coefficients}\label{sec:details}

This section fully details the derivation of the certified bounds on the coefficients $c_{k,m}$.
We first derive the coefficient formula by reducing the
Hermite coefficients over argmax regions to one-dimensional Gaussian integrals.
We then explain the normalization and symmetry grouping used to make the
certificate finite and efficient, prove the signed-tail bounds needed to enclose
the integrals rigorously, and finally describe the ball-arithmetic implementation.

\subsection{Reduction to One-Dimensional Gaussian Integrals}
Throughout, let
\[
Z=(Z_1,\ldots,Z_k)\sim\mathcal N(\bm 0_k,I_k),
\]
and recall the argmax indicator $f_1=\1_{A_1}$ from
Section~\ref{subsec:hermite}.  By~\eqref{eq:ckm-def}, it suffices to
compute the coefficients $\hat f_1(\alpha)$.
The following lemma gives the one-dimensional formula used in
Proposition~\ref{prop:main-one-dimensional-certificate}; the intuition is
summarized in Section~\ref{subsec:theory-one-dimensional}.

\begin{lemma}\label{lem:one-dimensional-hermite-coefficients}
Let $\alpha=(\alpha_1,\ldots,\alpha_k)\in\mathbb{N}_0^k$.  Define
\[
J_a(x)=\int_{-\infty}^x \varphi(y)h_a(y)\,\mathrm{d}y,
\quad a\in\mathbb{N}_0.
\]
Then
\begin{equation}\label{eq:f1-hat-one-dimensional}
\hat f_1(\alpha)
=
\int_{-\infty}^{\infty}
\varphi(x)h_{\alpha_1}(x)
\prod_{q=2}^k J_{\alpha_q}(x)\,\mathrm{d}x.
\end{equation}
Moreover,
\begin{equation}\label{eq:Ja-formula}
J_0(x)=\Phi(x)
\quad \text{and}\quad
J_a(x)=
-\frac{\varphi(x)h_{a-1}(x)}{\sqrt a}
\quad\text{for }\, a\ge1.
\end{equation}
\end{lemma}

\begin{proof}
By definition \eqref{eq:hermite-expansion},
\[
\hat f_1(\alpha)
=
\E\left[
\1[Z_1\ge Z_2\land \cdots \land Z_1\ge Z_k]
\prod_{q=1}^k h_{\alpha_q}(Z_q)
\right],
\]
where ties have probability zero.  Conditioning on $Z_1=x$ and using the
independence of $Z_2,\ldots,Z_k$ gives
\[
\hat f_1(\alpha)
=
\int_{-\infty}^{\infty}
\varphi(x)h_{\alpha_1}(x)
\prod_{q=2}^k
\left(
\int_{-\infty}^x \varphi(y)h_{\alpha_q}(y)\,\mathrm{d}y
\right)
\mathrm{d}x.
\]
This proves \eqref{eq:f1-hat-one-dimensional}.

It remains to prove \eqref{eq:Ja-formula}.  The case $a=0$ is just
$J_0(x)=\Phi(x)$.
For $a\ge1$
we use the identity
\[
\frac{\mathrm{d}}{\mathrm{d}y}\left(\varphi(y)H_{a-1}(y)\right)
=
-\varphi(y)H_a(y),
\]
which follows directly from the definition in~\eqref{eq:hermite-derivative-definition}.
Integrating, we get
\[
\int_{-\infty}^x \varphi(y)H_a(y)\,\mathrm{d}y
=
-\varphi(x)H_{a-1}(x),
\]
since $\varphi(x)H_{a-1}(x)\to0$ as $x\to-\infty$.
Dividing by $\sqrt{a!}$, we obtain
~\eqref{eq:Ja-formula}.
\end{proof}

Combining Lemma~\ref{lem:one-dimensional-hermite-coefficients} with
\eqref{eq:ckm-def} gives an exact coefficient formula:
\begin{equation*}%
c_{k,m}
=
k\sum_{|\alpha|=m}
\left(
\int_{-\infty}^{\infty}
\varphi(x)h_{\alpha_1}(x)
\prod_{q=2}^k J_{\alpha_q}(x)\,\mathrm{d}x
\right)^2.
\end{equation*}
This is already a reduction to one-dimensional integrals, but for a
computer-assisted proof it is worth eliminating irrational constants
 and exploiting symmetry to reduce
the number of distinct integrals.

\subsection{Normalization and Grouping by Symmetry}\label{subsec:details-normalization-symmetry}
\paragraph{Removing the square-root normalizations.}
Define
\[
K_0(x)=\Phi(x)
\quad \text{and}\quad
K_a(x)=-\varphi(x)H_{a-1}(x)
\quad \text{for } a\ge1,
\]
and for $\alpha=(\alpha_1,\ldots,\alpha_k)$, define
\begin{equation*}%
B_\alpha
=
\int_{-\infty}^{\infty}
\varphi(x)H_{\alpha_1}(x)
\prod_{q=2}^k K_{\alpha_q}(x)\,\mathrm{d}x.
\end{equation*}

\begin{lemma}\label{lem:unnormalized-coefficient-formula}
For every $k\ge2$ and $m\in\mathbb{N}_0$,
\begin{equation}\label{eq:ckm-unnormalized-formula}
c_{k,m}
=
k\sum_{|\alpha|=m}
\frac{B_\alpha^2}{\alpha_1!\alpha_2!\cdots\alpha_k!}.
\end{equation}
\end{lemma}

\begin{proof}
For $a=0$, we have $J_0=K_0$.  For $a\ge1$, equation
\eqref{eq:Ja-formula} gives
\[
J_a(x)
=
-\frac{\varphi(x)H_{a-1}(x)}{\sqrt{a!}}
=
\frac{K_a(x)}{\sqrt{a!}}.
\]
Also $h_{\alpha_1}=\frac{H_{\alpha_1}}{\sqrt{\alpha_1!}}$. Therefore, by \eqref{eq:f1-hat-one-dimensional},
\[
\hat f_1(\alpha)
=
\frac{B_\alpha}
{\sqrt{\alpha_1!\alpha_2!\cdots\alpha_k!}}.
\]
Substituting this into \eqref{eq:ckm-def} proves the desired equality.
\end{proof}

\paragraph{Grouping by symmetry.}
The integrand of $B_\alpha$ depends on $(\alpha_2,\ldots,\alpha_k)$ only
through the \emph{multiset} of these indices, because the product
$\prod_{q=2}^k K_{\alpha_q}(x)$ is invariant under permutations.
We can therefore evaluate one integral per multiset and multiply by the
number of orderings.  Let $\lambda$ be a multiset of the remaining $k-1$ nonnegative
integers.  We write $|\lambda|$ for the sum of its elements, and
$n_b(\lambda)$ for the multiplicity of $b$ in $\lambda$.  Define
\[
N(\lambda)
=
\frac{(k-1)!}{\prod_{b\ge0}n_b(\lambda)!},
\]
the number of distinct orderings of $\lambda$, and, for
$a\in\mathbb{N}_0$,
\begin{equation*}%
B_{a,\lambda}
=
\int_{-\infty}^{\infty}
\varphi(x)H_a(x)\prod_{b\in \lambda}K_b(x)\,\mathrm{d}x.
\end{equation*}
We now prove the first part of Proposition~\ref{prop:main-one-dimensional-certificate} in Section~\ref{subsec:theory-one-dimensional},
stating that for every $k\ge2$ and $m\in\mathbb{N}_0$,
\begin{equation}\label{eq:ckm-grouped-formula}
c_{k,m}
=
k
\sum_{a=0}^m
\sum_{\substack{\lambda\colon |\lambda|=m-a\\
\lambda\text{ is a multiset of size }k-1}}
N(\lambda)\,
\frac{B_{a,\lambda}^2}
{a!\prod_{b\in \lambda} b!}.
\end{equation}

\begin{proof}[Proof of the first part of Proposition~\ref{prop:main-one-dimensional-certificate}]
In \eqref{eq:ckm-unnormalized-formula}, fix the first coordinate
$\alpha_1=a$.  The remaining coordinates
$(\alpha_2,\ldots,\alpha_k)$ contribute to $B_\alpha$ only through their
multiset $\lambda$, because the product
$\prod_{q=2}^k K_{\alpha_q}(x)$ is invariant under permutation of these
coordinates.  For a fixed multiset $\lambda$, the number of corresponding
ordered tuples is $N(\lambda)$, and each such tuple has the same
denominator $a!\prod_{b\in \lambda}b!$.  Summing over all $a$ and all such
multisets proves the formula.
\end{proof}

The savings here are substantial: for $2\le k\le 32$ and $1\le m\le 24$,
formula \eqref{eq:ckm-unnormalized-formula} would require about
$9.93\times 10^{15}$ integrals (one per multi-index $\alpha$),
whereas \eqref{eq:ckm-grouped-formula} requires only
$8.10\times 10^{5}$ integrals (one per pair $(a,\lambda)$).
For $2\le k\le 16$ and $1\le m\le 32$, the corresponding numbers are about
$3.35\times 10^{12}$ and $1.66\times 10^{6}$, respectively.
These are precisely the two parameter ranges for the reusable coefficient data described in Section~\ref{subsec:coefficient-implementation}.

In summary: unwinding the definition of
$K_b$, the integrand of $B_{a,\lambda}$ is
\begin{equation}\label{eq:Ba-lambda-integrand}
(-1)^{r(\lambda)}\,
\varphi(x)^{r(\lambda)+1}\,
\Phi(x)^{s(\lambda)}\,
S_{a,\lambda}(x)
\quad \text{with}\quad
S_{a,\lambda}(x)
=
H_a(x)\prod_{b\in\lambda\colon b\ge1}H_{b-1}(x),
\end{equation}
where
\[
r(\lambda)=|\{b\in\lambda:b\ge1\}|,
\quad
s(\lambda)=|\{b\in\lambda:b=0\}|,
\quad \text{and}\quad
r(\lambda)+s(\lambda)=k-1,
\]
and $S_{a,\lambda}$ is a polynomial with integer coefficients.  Every
integral we need is thus of one fixed, simple shape: a power of the
Gaussian density, times a power of the Gaussian CDF, times an explicit
integer polynomial.

\subsection{Rigorous Enclosures: Finite Window and Signed Tails}

We now explain how the integrals $B_{a,\lambda}$ are certified. One difficulty is that
the integrands are \emph{signed}. If the integrands were nonnegative, truncating the
integral to a bounded interval would immediately give a lower bound.
For a signed integral, however, the omitted tails may either reinforce or cancel the
contribution from the bounded interval. Since
\eqref{eq:ckm-grouped-formula} involves squares, we therefore first \emph{enclose}
each $B_{a,\lambda}$ in an interval and only then derive lower and upper bounds on
its square, via the following elementary lemma.

\begin{lemma}\label{lem:square-lower-bound}
For an interval $[L,U]$, define
\begin{equation*}
\operatorname{sqLB}([L,U])
=
\begin{cases}
0, & 0\in[L,U],\\
\min\{L^2,U^2\}, & 0\notin[L,U],
\end{cases}
\quad
\text{and}\quad
\operatorname{sqUB}([L,U])
=
\max\{L^2,U^2\}.
\end{equation*}
If $x\in[L,U]$, then
\[
\operatorname{sqLB}([L,U])
\le x^2
\le
\operatorname{sqUB}([L,U]).
\]
\end{lemma}

To enclose $B_{a,\lambda}$ itself, fix a window radius $R>0$ and split
\begin{equation}\label{eq:B-split}
B_{a,\lambda}
=
B_{a,\lambda}^{(R)}
+
T_{a,\lambda}^{(R)},
\quad \text{where}\quad
B_{a,\lambda}^{(R)}
=
\int_{-R}^R
\varphi(x)H_a(x)\prod_{b\in\lambda}K_b(x)\,\mathrm{d}x
\end{equation}
and $T_{a,\lambda}^{(R)}$ is the contribution of
$(-\infty,-R]\cup[R,\infty)$.  The two parts are treated by different
means:
\begin{itemize}
\item the finite-window integral $B_{a,\lambda}^{(R)}$ is enclosed in an
interval $[L_{a,\lambda},U_{a,\lambda}]$ by validated numerical
integration in ball arithmetic, as described in Section~\ref{subsec:coefficient-implementation};
\item the tail is bounded in absolute value,
$\left|T_{a,\lambda}^{(R)}\right|\le \varepsilon_{a,\lambda}$, by the explicit
analytic estimate of Lemma~\ref{lem:coefficient-tail-bound} below.
\end{itemize}
Together these give
$B_{a,\lambda}\in
[L_{a,\lambda}-\varepsilon_{a,\lambda},\,
U_{a,\lambda}+\varepsilon_{a,\lambda}]$,
and Lemma~\ref{lem:square-lower-bound} then gives lower and upper bounds on $B_{a,\lambda}^2$.

The tail estimate exploits the specific shape
\eqref{eq:Ba-lambda-integrand} of the integrand: the factor
$\Phi(x)^{s(\lambda)}$ is bounded by $1$; the factor
$\varphi(x)^{r(\lambda)+1}$ is a Gaussian kernel with an explicitly
enhanced decay rate $\mu=r(\lambda)+1$; and the polynomial factor
$S_{a,\lambda}$ has explicit integer coefficients, so its absolute value
is bounded coefficient-by-coefficient.  What remains are incomplete
Gaussian moments, which can be bounded rigorously using the two-term recurrence below.
For $\mu>0$ and $R>0$, define
\[
G_j(\mu,R)
=
\int_R^\infty x^j\exp\left(-\frac{\mu x^2}{2}\right)\mathrm{d}x .
\]

\begin{lemma}\label{lem:coefficient-tail-bound}
Let $S_{a,\lambda}(x)=\sum_{j=0}^{D}s_jx^j$, where $D$ is the degree of $S_{a,\lambda}$.  For every $a,\lambda$ and
every $R>0$,
\begin{equation}\label{eq:tail-bound}
\left|T_{a,\lambda}^{(R)}\right|
\le
2(2\pi)^{-(r(\lambda)+1)/2}
\sum_{j=0}^{D}|s_j|\,G_j(r(\lambda)+1,R).
\end{equation}
Moreover,
\[
G_0(\mu,R)
\le
\frac{\exp(-\mu R^2/2)}{\mu R},
\quad
G_1(\mu,R)
=
\frac{\exp(-\mu R^2/2)}{\mu},
\]
and
\begin{equation}\label{eq:Gj-recurrence}
G_j(\mu,R)
=
\frac{R^{j-1}\exp(-\mu R^2/2)}{\mu}
+
\frac{j-1}{\mu}G_{j-2}(\mu,R)
\quad \text{for }\, j\ge 2.
\end{equation}
\end{lemma}

\begin{proof}
By \eqref{eq:Ba-lambda-integrand} and $0\le\Phi(x)\le1$,
\[
\left|
\varphi(x)H_a(x)\prod_{b\in\lambda}K_b(x)
\right|
\le
\varphi(x)^{r(\lambda)+1}|S_{a,\lambda}(x)|.
\]
Using
\[
\varphi(x)^{r(\lambda)+1}
=
(2\pi)^{-(r(\lambda)+1)/2}
\exp\left(-\frac{(r(\lambda)+1)x^2}{2}\right)
\]
and
$|S_{a,\lambda}(x)|\le\sum_{j=0}^{D}|s_j||x|^j$,
we get
\[
\begin{aligned}
\left|T_{a,\lambda}^{(R)}\right|
&\le
\int_{|x|\ge R}
\varphi(x)^{r(\lambda)+1}|S_{a,\lambda}(x)|\,\mathrm{d}x\\
&\le
(2\pi)^{-(r(\lambda)+1)/2}
\sum_{j=0}^{D}|s_j|
\int_{|x|\ge R}
|x|^j
\exp\left(-\frac{(r(\lambda)+1)x^2}{2}\right)\mathrm{d}x\\
&=
2(2\pi)^{-(r(\lambda)+1)/2}
\sum_{j=0}^{D}|s_j|\,G_j(r(\lambda)+1,R).
\end{aligned}
\]
This proves \eqref{eq:tail-bound}.

It remains to justify the formulas for $G_j$.  For $G_0$, since
$x/R\ge1$ for $x\ge R$, we have the standard estimate~\cite[page~98]{grimmet_prob}
\[
G_0(\mu,R)
=
\int_R^\infty e^{-\mu x^2/2}\,\mathrm{d}x
\le
\frac1R\int_R^\infty x e^{-\mu x^2/2}\,\mathrm{d}x
=
\frac{e^{-\mu R^2/2}}{\mu R}.
\]
For $G_1$,
\[
G_1(\mu,R)
=
\int_R^\infty x e^{-\mu x^2/2}\,\mathrm{d}x
=
\frac{e^{-\mu R^2/2}}{\mu}.
\]
For $j\ge2$, integration by parts gives
\[
\begin{aligned}
G_j(\mu,R)
&=
\int_R^\infty x^{j-1}\cdot x e^{-\mu x^2/2}\,\mathrm{d}x\\
&=
\left.x^{j-1}\frac{e^{-\mu x^2/2}}{-\mu}\right|_{x=R}^{x=\infty}
+
\frac{j-1}{\mu}
\int_R^\infty x^{j-2}e^{-\mu x^2/2}\,\mathrm{d}x\\
&=
\frac{R^{j-1}e^{-\mu R^2/2}}{\mu}
+
\frac{j-1}{\mu}G_{j-2}(\mu,R).
\end{aligned}
\]
This proves the recurrence.
\end{proof}

In particular, for each fixed pair $(a,\lambda)$, the bound in \eqref{eq:tail-bound}
decays exponentially in $R^2$, up to a polynomial factor in $R$.
Thus, the tail contribution can be made arbitrarily small by increasing
the window radius $R$; the choice of $R$ used in our implementation is
specified in Section~\ref{subsec:coefficient-implementation}.

\subsection{Assembling the Certificate}\label{subsec:details-assembling}

Combining the pieces yields certified rational lower and upper bounds on $c_{k,m}$.
Fix a maximum degree $M\in \mathbb{N}_0$.
Fix $k\ge 2$ and $m\in \{0,\dots, M\}$.
For each $a\in\{0,\ldots,m\}$ and multiset $\lambda$
of size $k-1$ with $|\lambda|=m-a$, compute an interval
$[L_{a,\lambda},U_{a,\lambda}]\ni B_{a,\lambda}^{(R)}$ and a bound
$\varepsilon_{a,\lambda}\ge|T_{a,\lambda}^{(R)}|$ using
Lemma~\ref{lem:coefficient-tail-bound}. Then set
\[
I_{a,\lambda}
=
[L_{a,\lambda}-\varepsilon_{a,\lambda},\,
 U_{a,\lambda}+\varepsilon_{a,\lambda}]
\]
and
\[
\underline b_{a,\lambda}
=
\operatorname{sqLB}(I_{a,\lambda})\quad \text{and}\quad
\overline b_{a,\lambda}
=
\operatorname{sqUB}(I_{a,\lambda}).
\]
Then define $\underline c_{k,0}=\overline c_{k,0}=1/k$ and, for $1\le m\le M$,
\begin{equation}\label{eq:ckm-certified-lower-bound}
\underline c_{k,m}
=
k
\sum_{a=0}^m
\sum_{\substack{\lambda:\ |\lambda|=m-a\\
\lambda\text{ is a multiset of size }k-1}}
N(\lambda)\,
\frac{\underline b_{a,\lambda}}
{a!\prod_{b\in\lambda} b!}
\end{equation}
and
\begin{equation}\label{eq:ckm-certified-upper-bound}
\overline c_{k,m}
=
k
\sum_{a=0}^m
\sum_{\substack{\lambda:\ |\lambda|=m-a\\
\lambda\text{ is a multiset of size }k-1}}
N(\lambda)\,
\frac{\overline b_{a,\lambda}}
{a!\prod_{b\in\lambda} b!}.
\end{equation}

We now prove the second part of Proposition~\ref{prop:main-one-dimensional-certificate} in Section~\ref{subsec:theory-one-dimensional},
stating that if the interval $I_{a,\lambda}$ is rational,
$\underline c_{k,m}$ and $\overline c_{k,m}$ are certified rational lower and upper bounds, respectively, on $c_{k,m}$.

\begin{proof}[Proof of the second part of Proposition~\ref{prop:main-one-dimensional-certificate}]
By \eqref{eq:B-split} and Lemma~\ref{lem:coefficient-tail-bound}, for every
pair $(a,\lambda)$ appearing in \eqref{eq:ckm-grouped-formula}, we have
$B_{a,\lambda} \in I_{a,\lambda}$.
Hence, Lemma~\ref{lem:square-lower-bound} gives
$\underline b_{a,\lambda}
\le B_{a,\lambda}^2
\le \overline b_{a,\lambda}$.
Substituting these bounds into \eqref{eq:ckm-grouped-formula} term by term,
and noting that all weights
$kN(\lambda)/(a!\prod_{b\in\lambda} b!)$
are positive, yields
$\underline c_{k,m} \le c_{k,m} \le \overline c_{k,m}$.
Finally, if $L_{a,\lambda}$, $U_{a,\lambda}$, and
$\varepsilon_{a,\lambda}$ are rational, then the endpoints of
$I_{a,\lambda}$ are rational, and hence so are
$\underline b_{a,\lambda}$ and $\overline b_{a,\lambda}$.
It follows from their defining formulas that
$\underline c_{k,m}$ and $\overline c_{k,m}$ are rational as well.
\end{proof}

Although the above argument also applies when $m=0$, no numerical certification is needed in that case.
By \eqref{eq:Pk-zero}, $c_{k,0}=1/k$ exactly.
Hence, in our certificates, we set $\underline c_{k,0} = \overline c_{k,0} = 1/k$,
and apply the certification procedure above only for $m\ge 1$.

\subsection{Implementation in Ball Arithmetic with FLINT/Arb}\label{subsec:coefficient-implementation}

The implementation mirrors the certification procedure developed above and is
independent of any particular application.  Given a finite collection of
spoke parameters $k\ge 2$ and a maximum degree $M$, it produces certified
rational bounds $\underline{c}_{k,m} \le c_{k,m} \le \overline{c}_{k,m}$ for $m=0,\ldots,M$.
The resulting bounds are stored as reusable rational data and are subsequently
used by the application-specific computations in Sections~\ref{sec:approximations} and~\ref{sec:UB}, without
repeating the numerical integration.

The computation consists of three types of operations:
\begin{enumerate}
    \item \emph{exact combinatorics}: constructing the integer-coefficient
    Hermite polynomials $H_a$ and the polynomials $S_{a,\lambda}$, as well as
    the multiplicities $N(\lambda)$ and the factorial denominators;

    \item \emph{certified analytic computation}: enclosing the finite-window
    integrals $B^{(R)}_{a,\lambda}$ and evaluating the analytic tail bounds
    from Lemma~\ref{lem:coefficient-tail-bound};

    \item \emph{exact rational postprocessing}: outward-rounding the certified
    enclosures to rational intervals, applying the square bounds of
    Lemma~\ref{lem:square-lower-bound}, and assembling the coefficient bounds
    $\underline{c}_{k,m}$ and $\overline{c}_{k,m}$ according to~\eqref{eq:ckm-certified-lower-bound} and~\eqref{eq:ckm-certified-upper-bound}.
\end{enumerate}

\paragraph{Ball arithmetic.}
For the certified analytic computations, we use
Arb~\cite{Johansson17}, now part of FLINT,\footnote{Arb was merged into
FLINT in version~3.  We access it through the Python bindings
\texttt{python-flint}
(\url{https://github.com/flintlib/python-flint}).}
together with its complex-ball counterpart ACB.
Ball arithmetic represents a real or complex quantity by an
arbitrary-precision enclosure and propagates rounding and approximation errors
through the computation, so that the resulting ball is guaranteed to contain
the exact value.

For the finite-window integral $B^{(R)}_{a,\lambda}$, we use the validated
numerical integration routine \texttt{acb.integral}, based on ball
arithmetic~\cite{Johansson18}.  The integrand in~\eqref{eq:Ba-lambda-integrand} extends to
an entire function of the complex variable $z$: the polynomial
$S_{a,\lambda}(z)$ is entire, as are
$\exp(-z^2/2)$ and
\[
    \Phi(z)
    =
    \frac{1}{2}
    \left(
        1+\operatorname{erf}\left(\frac{z}{\sqrt{2}}\right)
    \right).
\]
It can therefore be evaluated on complex balls, and the validated integrator
returns an enclosure containing $B^{(R)}_{a,\lambda}$.  The analytic tail
bound in Lemma~\ref{lem:coefficient-tail-bound} is evaluated separately using real Arb balls and the
recurrence~\eqref{eq:Gj-recurrence}.

\paragraph{Pipeline.}
At the beginning of a run, the implementation constructs
$H_0,\ldots,H_M$ by the recurrence~\eqref{eq:hermite-recurrence}
using exact integer polynomial arithmetic.  For each requested $k\ge2$, the
degree-zero coefficient is handled without numerical computation: $\underline{c}_{k,0} = \overline{c}_{k,0} = 1/k$.
For each $m=1,\ldots,M$, the implementation enumerates exactly the pairs
$(a,\lambda)$ appearing in~\eqref{eq:ckm-grouped-formula}. For every such pair, it performs the
following steps:
\begin{enumerate}
    \item enclose $B^{(R)}_{a,\lambda}$ using validated ACB integration on
    $[-R,R]$, and outward-round the lower and upper endpoints of the resulting
    ball to $\delta$ decimal digits, obtaining a rational interval
    $[L_{a,\lambda},U_{a,\lambda}]$;

    \item evaluate the tail bound~\eqref{eq:tail-bound} in Arb using the recurrence~\eqref{eq:Gj-recurrence}, and
    outward-round its upper endpoint to a rational number
    $\varepsilon_{a,\lambda}$;

    \item form the rational enclosure
        $I_{a,\lambda}
        =
        [L_{a,\lambda}-\varepsilon_{a,\lambda},
         U_{a,\lambda}+\varepsilon_{a,\lambda}]$
    and compute
        $\underline{b}_{a,\lambda}
        =
        \operatorname{sqLB}(I_{a,\lambda})$ and
        $\overline{b}_{a,\lambda}
        =
        \operatorname{sqUB}(I_{a,\lambda})$
    in exact rational arithmetic;

    \item add the corresponding weighted contributions to
    $\underline{c}_{k,m}$ and $\overline{c}_{k,m}$ using
    \eqref{eq:ckm-certified-lower-bound} and~\eqref{eq:ckm-certified-upper-bound}, again in exact rational arithmetic.
\end{enumerate}
Thus, all non-exact numerical evaluation in the coefficient generator is
performed inside certified Arb/ACB ball arithmetic.  Once the resulting
enclosures have been outward-rounded to rational endpoints, all remaining
coefficient computations are exact over $\mathbb{Q}$.

\paragraph{Parameters and coefficient data.}
For the computations used in this paper, we take the finite-window radius
$R=8$, use a working precision of $256$ bits, and outward-round the ball
enclosures to $\delta=40$ decimal digits.  We generate two reusable
collections of coefficient bounds:
\[
    2\le k\le32\quad \text{and}\quad 0\le m\le24
\]
and
\[
    2\le k\le16\quad \text{and}\quad 0\le m\le32.
\]
Together, these collections contain all Hermite-coefficient bounds required
by the certified computations in Sections~\ref{sec:approximations} and~\ref{sec:UB}.
The exact rational bounds, together with the parameters used to generate them,
are stored in machine-readable files.\footnote{The source code, coefficient
data, and verification scripts for all computer-assisted results in this paper
are available at \url{https://github.com/atsushi-miyauchi/hermite-brings-a-laptop}.}
For reference, Table~\ref{tab:certified-coefficients} displays outward-rounded decimal
representations of a small subset of these exact data; the computer-assisted
proofs use the exact rational values rather than the displayed decimals.

\begin{table*}[t]
\centering
\caption{Certified lower and upper bounds $\underline{c}_{k,m}$ and $\overline{c}_{k,m}$ on the coefficients
$c_{k,m}$ for $k=2,\ldots,8$ and $m=0,\ldots,8$. Upper bounds are rounded upward and lower bounds downward
to six digits after the decimal point. The values $\underline{c}_{k,0}=\overline{c}_{k,0}=1/k$ are given by \eqref{eq:Pk-zero}.}
\label{tab:certified-coefficients}
\setlength{\tabcolsep}{5.0pt}
\renewcommand{\arraystretch}{1.0}
\begin{tabular}{cccccccccc}
\toprule
$m$ & $0$ & $1$ & $2$ & $3$ & $4$
    & $5$ & $6$ & $7$ & $8$ \\
\midrule

$\overline{c}_{2,m}$
& $1/2$ & $0.318310$ & $0.000000$ & $0.053052$ & $0.000000$
& $0.023874$ & $0.000000$ & $0.014211$ & $0.000000$ \\
$\underline{c}_{2,m}$
& $1/2$ & $0.318309$ & $0.000000$ & $0.053051$ & $0.000000$
& $0.023873$ & $0.000000$ & $0.014210$ & $0.000000$ \\
\midrule
$\overline{c}_{3,m}$
& $1/3$ & $0.358099$ & $0.056994$ & $0.044763$ & $0.023748$
& $0.018465$ & $0.013299$ & $0.010741$ & $0.008651$ \\
$\underline{c}_{3,m}$
& $1/3$ & $0.358098$ & $0.056993$ & $0.044762$ & $0.023747$
& $0.018464$ & $0.013298$ & $0.010740$ & $0.008650$ \\
\midrule
$\overline{c}_{4,m}$
& $1/4$ & $0.353205$ & $0.101322$ & $0.038489$ & $0.037527$
& $0.017396$ & $0.019264$ & $0.011064$ & $0.011858$ \\
$\underline{c}_{4,m}$
& $1/4$ & $0.353204$ & $0.101321$ & $0.038488$ & $0.037526$
& $0.017395$ & $0.019263$ & $0.011063$ & $0.011857$ \\
\midrule
$\overline{c}_{5,m}$
& $1/5$ & $0.338122$ & $0.133341$ & $0.036609$ & $0.044452$
& $0.019404$ & $0.021277$ & $0.013170$ & $0.012704$ \\
$\underline{c}_{5,m}$
& $1/5$ & $0.338121$ & $0.133340$ & $0.036608$ & $0.044451$
& $0.019403$ & $0.021276$ & $0.013169$ & $0.012703$ \\
\midrule
$\overline{c}_{6,m}$
& $1/6$ & $0.321163$ & $0.156593$ & $0.037745$ & $0.047484$
& $0.022757$ & $0.021578$ & $0.015655$ & $0.012820$ \\
$\underline{c}_{6,m}$
& $1/6$ & $0.321162$ & $0.156592$ & $0.037744$ & $0.047483$
& $0.022756$ & $0.021577$ & $0.015654$ & $0.012819$ \\
\midrule
$\overline{c}_{7,m}$
& $1/7$ & $0.304732$ & $0.173734$ & $0.040710$ & $0.048361$
& $0.026501$ & $0.021254$ & $0.017931$ & $0.012855$ \\
$\underline{c}_{7,m}$
& $1/7$ & $0.304731$ & $0.173733$ & $0.040709$ & $0.048360$
& $0.026500$ & $0.021253$ & $0.017930$ & $0.012854$ \\
\midrule
$\overline{c}_{8,m}$
& $1/8$ & $0.289520$ & $0.186543$ & $0.044728$ & $0.048062$
& $0.030169$ & $0.020807$ & $0.019802$ & $0.013029$ \\
$\underline{c}_{8,m}$
& $1/8$ & $0.289519$ & $0.186542$ & $0.044727$ & $0.048061$
& $0.030168$ & $0.020806$ & $0.019801$ & $0.013028$ \\

\bottomrule
\end{tabular}
\end{table*}

\section{Improved Approximations}\label{sec:approximations}
In this section, we utilize our Hermite-coefficient certification framework to obtain improved approximation guarantees for several important clustering problems, namely \maxagree~\cite{Bansal+04}, \maxagreek~\cite{Giotis+06_j}, Max $K$-Cut~\cite{Garey+02,Frieze+97}, and modularity maximization~\cite{Newman+04,Fortunato10}. To this end, we first consider a general clustering problem on weighted edge-labeled graphs that contains the first three problems as special cases, and introduce a family of SDP relaxations. We then formalize the standard edge-wise analysis, reducing the task of certifying a desired approximation ratio to verifying the positive-edge and negative-edge inequalities, and describe how these inequalities can be rigorously verified using finite Hermite bounds, Bernstein certificates, and exact rational arithmetic. We subsequently apply this framework to the first three problems. Although modularity maximization does not fit directly into this unified framework, its approximation analysis relies on the same certified bounds on the collision probabilities.

\subsection{General Clustering Problem, SDP Relaxations, and Edge-Wise Analysis}\label{subsec:approximations-general}
We consider the following general clustering problem.
Let $G=(V,E)$ be an edge-labeled undirected graph with $V=\{1,\dots,n\}$ and $E=E_+\cup E_-$, where $E_+\cap E_-=\emptyset$.
Each edge $\{i,j\}\in E$ has a weight $w_{ij}>0$, representing the strength of its label.
For a clustering $\mathcal{C}$ of $V$, let $\mathcal{C}(i)$ denote the cluster containing $i$.
We define the objective function
\begin{equation*}
\mathsf{val}(\mathcal{C})=\sum_{\{i,j\}\in E_+}w_{ij}\,\1[\mathcal{C}(i)=\mathcal{C}(j)] + \sum_{\{i,j\}\in E_-}w_{ij}\,\1[\mathcal{C}(i)\neq \mathcal{C}(j)],
\end{equation*}
which captures the intuition that
the endpoints of a positive edge should be assigned to the same cluster, whereas the endpoints of a negative edge should be assigned to different clusters.
Let $\mathcal{F}$ be a family of feasible clusterings of $V$.
Given $G=(V,E)$ together with the edge weights, the general problem is to find a clustering $\mathcal{C}\in\mathcal{F}$ that maximizes $\mathsf{val}(\mathcal{C})$.
We denote the optimal value of this problem by $\mathrm{OPT}$.

We next introduce a family of SDP relaxations. For $a\in [-1/2,0]$, define the two affine functions
\begin{equation*}
s^+_a(\rho) = \frac{\rho-a}{1-a}\quad \text{and} \quad s^-_a(\rho) = \frac{1-\rho}{1-a} \qquad \text{for } \rho\in[a,1].
\end{equation*}
Note that $s^+_a(1)=s^-_a(a)=1$ and $s^+_a(a)=s^-_a(1)=0$.
We use $s^+_a(\rho)$ and $s^-_a(\rho)$, together with the edge weights $w_{ij}$, as the SDP contributions of positive and negative edges, respectively.
The parameter $a$ lower bounds the inner product between any pair of SDP vectors and, equivalently, determines an upper bound on the angle between them.
Specifically, we consider the following SDP:
\begin{alignat}{4}\label{eq:maxagree_sdp}
&\mathrm{SDP}(a)\colon &\, &\text{maximize} &\quad &\sum_{\{i,j\}\in E_+}w_{ij}\, s^+_a&&(\bm{v}_i\cdot \bm{v}_j) + \sum_{\{i,j\}\in E_-}w_{ij}\, s^-_a(\bm{v}_i\cdot \bm{v}_j)\\
&                           &   &\text{subject to} &&\bm{v}_i\cdot \bm{v}_i = 1 &&\text{for all } i\in V,\notag\\
&                           &   &                  &&\bm{v}_i\cdot \bm{v}_j \geq a &&\text{for all } i,j\in V,\notag\\
&                           &   &                  &&\bm{v}_i\in \mathbb{R}^n &&\text{for all } i\in V\notag.
\end{alignat}
We denote by $\mathrm{OPT}_\mathrm{SDP}$ the optimal value of this SDP.

To analyze the approximation ratio of a rounding scheme for the family of SDPs, we formalize the standard edge-wise analysis:
\begin{lemma}\label{lem:edge-wise-analysis}
Consider any special case of the general clustering problem defined above.
Let $a\in [-1/2,0]$, and suppose that $\mathrm{SDP}(a)$ is a relaxation of the problem.
Let $(\bm{v}^*_i)_{i\in V}$ be an optimal solution to $\mathrm{SDP}(a)$.\footnote{A feasible solution of $\mathrm{SDP}(a)$ with objective value within an additive $\varepsilon$ of $\mathrm{OPT}_\mathrm{SDP}$ can be computed in time polynomial in the input size and $\log(1/\varepsilon)$.
The approximation guarantees therefore hold up to an arbitrarily small loss (for multiplicative guarantees, we first normalize the total edge weight). The stated numerical
    constants can in fact be attained exactly, because the certified inequalities hold with positive slack (in the case of \maxagree,  after slightly perturbing the analysis
            mixture.)}
Suppose that a mixture of the $k$-spokes roundings applied to $(\bm{v}^*_i)_{i\in V}$ always returns a feasible solution.
If, for some $\gamma \geq 0$, its collision probability $P(\rho)$ satisfies
\begin{equation}\label{eq:edge-inequalities}
P(\rho)\geq \gamma s^+_a(\rho) \quad \text{for all }\rho\in[a,1]\qquad and \qquad 1-P(a)\geq \gamma,
\end{equation}
then the rounding achieves an expected approximation ratio of $\gamma$.
\end{lemma}

\begin{proof}
Since $\mathrm{SDP}(a)$ is a relaxation, $\mathrm{OPT}_\mathrm{SDP}\geq \mathrm{OPT}$.
For $\{i,j\}\in E_+$, the positive-edge inequality gives
\begin{equation*}
P(\bm{v}^*_i\cdot\bm{v}^*_j)\geq \gamma\, s^+_a(\bm{v}^*_i\cdot\bm{v}^*_j).
\end{equation*}
For $\{i,j\}\in E_-$ with $\bm{v}^*_i\cdot\bm{v}^*_j<1$, by Corollary~\ref{cor:negative-ratio-monotonicity} and the assumption,
\begin{equation*}
\frac{1-P(\bm{v}^*_i\cdot\bm{v}^*_j)}{1-\bm{v}^*_i\cdot\bm{v}^*_j}
\geq \frac{1-P(a)}{1-a}
\geq \frac{\gamma}{1-a},
\end{equation*}
and thus
\begin{equation*}
1-P(\bm{v}^*_i\cdot\bm{v}^*_j) \geq \gamma\, s^-_a(\bm{v}^*_i\cdot\bm{v}^*_j),
\end{equation*}
which also holds if $\bm{v}^*_i\cdot\bm{v}^*_j=1$.
Let $\mathcal{C}_\mathrm{ALG}$ be the output of the (randomized) rounding. Then,
\begin{align*}
\E[\mathsf{val}(\mathcal{C}_\mathrm{ALG})]
&= \sum_{\{i,j\}\in E_+}w_{ij}P(\bm{v}^*_i\cdot\bm{v}^*_j) + \sum_{\{i,j\}\in E_-}w_{ij}\bigl(1-P(\bm{v}^*_i\cdot\bm{v}^*_j)\bigr)\\
&\geq \gamma \sum_{\{i,j\}\in E_+}w_{ij}\, s^+_a(\bm{v}^*_i\cdot\bm{v}^*_j) + \gamma \sum_{\{i,j\}\in E_-}w_{ij}\, s^-_a(\bm{v}^*_i\cdot\bm{v}^*_j)\\
&= \gamma \, \mathrm{OPT}_{\mathrm{SDP}} \geq \gamma \, \mathrm{OPT}.
\end{align*}
This completes the proof.
\end{proof}

We refer to the two inequalities in~\eqref{eq:edge-inequalities} as the positive-edge and negative-edge inequalities.
The above lemma formalizing the standard edge-wise analysis reduces the task of certifying a desired approximation ratio for a rounding scheme with a collision probability $P(\rho)$ to verifying the positive-edge inequality for all $\rho \in [a,1]$ and the negative-edge inequality at the single point $\rho=a$.

\subsection{Finite Certification of the Positive-Edge and Negative-Edge Inequalities}
We now explain how our Hermite-coefficient certification framework further reduces the task of certifying the positive-edge and negative-edge inequalities to finite, exactly checkable conditions.
Recall that Corollary~\ref{cor:theory-P-LB-UB} constructs lower and upper bounds on the collision probability $P(\rho)$ of any mixture of the $k$-spokes roundings from finitely many certified lower and upper bounds on Hermite coefficients.
Using this, we obtain the following machinery:

\begin{lemma}\label{lem:machinery}
Let $M\in \mathbb{N}_0$. Recall $\underline P_M(\rho)$ and $\overline P_M(\rho)$ defined in Corollary~\ref{cor:theory-P-LB-UB}.
Let $a\in [-1/2,0]$ and $\gamma\geq 0$.
\begin{itemize}
\item If
\begin{equation}\label{eq:hypothesis-positive}
\underline P_M(\rho)\ge \gamma s^+_a(\rho)\quad \text{for all }\rho\in [a,1],
\end{equation}
then the positive-edge inequality $P(\rho)\ge \gamma s^+_a(\rho)$ holds for all $\rho\in [a,1]$;
\item If
\begin{equation}\label{eq:hypothesis-negative}
1- \overline P_M(a)\ge \gamma,
\end{equation}
then the negative-edge inequality $1-P(a)\ge \gamma$ holds.
\end{itemize}
\end{lemma}

\begin{proof}
Corollary~\ref{cor:theory-P-LB-UB} guarantees $\underline P_M(\rho)\le P(\rho)\le \overline P_M(\rho)$ for all $\rho \in[-1,1]$.
Hence, the hypotheses~\eqref{eq:hypothesis-positive} and \eqref{eq:hypothesis-negative} directly ensure the desired inequalities.
\end{proof}

Given rational bounds on $c_{k,m}$, checking the hypothesis~\eqref{eq:hypothesis-negative} for the negative-edge inequality can be done in exact rational arithmetic.
In contrast, checking the hypothesis~\eqref{eq:hypothesis-positive} for the positive-edge inequality still requires an effort, since the inequality needs to be verified for infinitely many $\rho\in[a,1]$. However, we can do this in exact rational arithmetic using Bernstein certificates. Indeed, given rational bounds on $c_{k,m}$, the task reduces to verifying that the piecewise polynomial $\underline{P}_M(\rho)-\gamma s_a^+(\rho)$
is nonnegative over $\rho\in[a,1]$. More precisely, by the definition of $\underline{P}_M$, this amounts to checking the nonnegativity of a polynomial with rational coefficients on each of the intervals $[a,0]$ and $[0,1]$, with the former omitted when $a=0$. These polynomial inequalities can be verified in exact rational arithmetic using Bernstein certificates~\cite{bernstein_basis}:

\begin{lemma}\label{lem:bernstein-nonnegativity}
Let $q$ be a polynomial of degree at most $d$ and write
\[
q(\rho)
=
\sum_{i=0}^{d} \beta_i\binom{d}{i}
\left(\frac{\rho-u}{v-u}\right)^{i}
\left(\frac{v-\rho}{v-u}\right)^{d-i}
\]
in the Bernstein basis of degree $d$ on an interval $[u,v]$.  If
$\beta_i\ge0$ for all $i$, then $q\ge0$ on $[u,v]$.
\end{lemma}

\begin{proof}
Each Bernstein basis polynomial is a product of nonnegative factors on
$[u,v]$.
\end{proof}

Note that the converse need not hold, since the Bernstein enclosure can be
strict: a polynomial may be positive even though one of its Bernstein
coefficients is negative.  Subdivision removes this obstruction: Bernstein coefficients on shrinking subintervals converge
to values of $q$, so if $q>0$ on $[u,v]$, then recursively bisecting any
subinterval with a negative coefficient terminates with a partition of
$[u,v]$ on which Lemma~\ref{lem:bernstein-nonnegativity} applies
everywhere.  The resulting subdivision, together with the (rational)
Bernstein coefficients on each piece, constitutes a finite, exactly
checkable certificate of \eqref{eq:hypothesis-positive}.

In our implementation, the application-specific verifiers take as input the
exact rational coefficient bounds generated as described in Section~\ref{subsec:coefficient-implementation}.
Consequently, all computations in these verifiers are carried out in exact
integer or rational arithmetic, with no further floating-point or
ball-arithmetic computation.

\subsection{\maxagree}\label{sec:maxagree}
\maxagree\ is the special case of the general clustering problem introduced in Section~\ref{subsec:approximations-general}
obtained by taking $\mathcal{F}$ to be the family of all clusterings of $V$.
Here, we use $\mathrm{SDP}(a)$ with $a=0$, which coincides with the standard SDP relaxation used by \citet*{Charikar+05} and \citet{Swamy04}.
Our algorithm applies both the $4$-spokes and $5$-spokes roundings to an optimal solution to $\mathrm{SDP}(0)$ and returns the clustering with the larger objective value.
The algorithm is displayed in Algorithm~\ref{alg:maxagree}.
\begin{algorithm}[h]
\SetKwInput{Input}{Input}
\caption{Our algorithm for \maxagree}\label{alg:maxagree}
\Input{$V$, $E_+$, $E_-$, and $w_{ij}>0$ for $\{i,j\}\in E_+\cup E_-$}
Solve $\mathrm{SDP}(0)$ and obtain an optimal solution $(\bm{v}^*_i)_{i\in V}$\;
$\mathcal{C}_4\leftarrow \text{$4$-spokes}((\bm{v}^*_i)_{i\in V})$ and $\mathcal{C}_5\leftarrow \text{$5$-spokes}((\bm{v}^*_i)_{i\in V})$\;
\Return $\argmax_{\mathcal{C}\in \{\mathcal{C}_4,\mathcal{C}_5\}}\textsf{val}(\mathcal{C})$.
\end{algorithm}

Our goal in this section is to prove that Algorithm~\ref{alg:maxagree} has an expected approximation ratio of $0.7818$.
To this end, we consider the mixture of the $4$-spokes and $5$-spokes roundings with probabilities $0.364$ and $0.636$, respectively.
Since our algorithm only improves the objective value over this mixture, it suffices to prove an approximation ratio of $0.7818$ for the mixture.

Let $P(\rho)$ be the collision probability of the mixture,
and for $M\in \mathbb{N}_0$, let $\underline P_M(\rho)$ denote the certified lower bound associated with this mixture, as defined in Corollary~\ref{cor:theory-P-LB-UB}.
By Lemma~\ref{lem:machinery}, the positive-edge inequality follows if the hypothesis~\eqref{eq:hypothesis-positive} holds for some $M\in\mathbb{N}_0$ with $\gamma=0.7818$.
We verify the hypothesis in exact rational arithmetic using Bernstein certificates:
\begin{proposition}\label{prop:maxagree-positive}
With the exact rational lower bounds on $c_{k,m}$ generated as described in Section~\ref{subsec:coefficient-implementation},
\begin{equation*}
\underline P_{24}(\rho)\ge 0.7818\, s^+_0(\rho)= 0.7818 \rho \quad \text{for all } \rho\in[0,1].
\end{equation*}
\end{proposition}
For the negative-edge inequality, we do not need to verify the hypothesis~\eqref{eq:hypothesis-negative},
since the inequality itself can be shown analytically. Indeed, since $a=0$,
\begin{equation*}
1-P(a)=1-P(0)=1-(0.364\, P_4(0)+0.636\, P_5(0))=1-(0.364/4+0.636/5) = 0.7818.
\end{equation*}

Thus, we have the desired result:
\begin{theorem}\label{thm:maxagree}
For \maxagree, Algorithm~\ref{alg:maxagree}
has an expected approximation ratio of $0.7818$.
\end{theorem}

We complement this approximation guarantee with two near-optimality results in later sections.
In Section~\ref{sec:raghavendra-gap}, we study its quality from a complexity-theoretic perspective and show, under the UGC,
an inapproximability threshold of $0.802+\epsilon$.
In Section~\ref{sec:UB}, we study the limitations of rounding families under the edge-wise analysis and show that,
within a broad finite family of rounding schemes, no mixture can certify an approximation ratio larger than $0.78187$.

\subsection{\maxagreek}
\maxagreek\ is a variant of \maxagree\ and is obtained from the general clustering problem by taking $\mathcal{F}$ to be the family of all clusterings of $V$ with at most $K$ clusters, for a fixed integer $2\leq K\leq n$.
For $K=2$, \maxagreek\ admits the Goemans--Williamson approximation ratio $\alpha_{\mathrm{GW}}\approx0.8785$ for Max Cut~\cite{Giotis+06_j}.
This guarantee is optimal under the Unique Games Conjecture, since Max Cut is the special case in which $E=E_-$.
We therefore focus on $K\ge3$.
Let $a_K=-1/(K-1)$. We use $\mathrm{SDP}(a)$ with $a=a_K$, which coincides with the SDP relaxation used by \citet{Swamy04} for \maxagreek\ with $K\leq 5$.
Our algorithm runs the $k$-spokes roundings for all $k=1,\dots, K$ and returns the clustering with the largest objective value.
The algorithm is summarized in Algorithm~\ref{alg:maxagree-k}.

\begin{algorithm}[h]
\SetKwInput{Input}{Input}
\caption{Our algorithm for \maxagreek}\label{alg:maxagree-k}
\Input{$V$, $E_+$, $E_-$, $w_{ij}>0$ for $\{i,j\}\in E_+\cup E_-$}
Solve $\mathrm{SDP}(a_K)$ and obtain an optimal solution $(\bm{v}^*_i)_{i\in V}$\;
\textbf{for} $k=1,\dots, K$ \textbf{do} $\mathcal{C}_k\leftarrow k\text{-spokes}((\bm{v}^*_i)_{i\in V})$\;
$k^*\leftarrow \argmax_{k\in \{1,\dots, K\}}\textsf{val}(\mathcal{C}_k)$\;
\Return $\mathcal{C}_{k^*}$.
\end{algorithm}

Our goal in this section is to prove that for each $K=3,\dots, 16$,
Algorithm~\ref{alg:maxagree-k} has an expected approximation ratio $\gamma_K$ presented in Table~\ref{tab:maxagree-k}.
To this end, for each $K=3,\dots, 16$,
we consider the randomized mixture with probabilities $\bm{\lambda}^{(K)}=\left(\lambda^{(K)}_k\right)_{k\in [K]}$ presented in Table~\ref{tab:maxagree-k}.
Since our algorithm only improves the objective value over the mixtures, it suffices to prove an approximation ratio of $\gamma_K$ for the mixtures.

\begin{table}[t]
\centering
\caption{Approximation ratios of Algorithm~\ref{alg:maxagree-k} for \maxagreek.
The best known approximation ratios are $0.77$ for $K=3,4,5$ and $0.7666$ for $K\geq 6$~\cite{Swamy04}.
The mixture is the one that achieves the approximation ratio, where only positive elements are displayed and $\lambda^{(K)}_k$ is written as $\lambda_k$, for simplicity.
Note that $U^\mathrm{cert}_K$ is a certified upper bound on the approximation ratio that can be established
by any mixture of the $k$-spokes roundings for $1\le k\le K$ within the edge-wise analysis framework; see Section~\ref{sec:UB}.}
\label{tab:maxagree-k}
\begin{tabular}{rccc}
\toprule
$K$ & $\gamma_K$ & Mixture $\bm{\lambda}^{(K)}$ & $U^\mathrm{cert}_K$\\
\midrule
$3$
& $0.8151$
& $(\lambda_2,\lambda_3)=(0.123,0.877)$
& $0.81523$\\

$4$
& $0.8025$
& $(\lambda_3,\lambda_4)=(0.720,0.280)$
& $0.80254$\\

$5$
& $0.7964$
& $(\lambda_3,\lambda_4)=(0.453,0.547)$
& $0.79652$\\

$6$
& $0.7934$
& $(\lambda_3,\lambda_4)=(0.291,0.709)$
& $0.79342$\\

$7$
& $0.7914$
& $(\lambda_3,\lambda_4)=(0.182,0.818)$
& $0.79155$\\

$8$
& $0.7902$
& $(\lambda_3,\lambda_4)=(0.100,0.900)$
& $0.79029$\\

$9$
& $0.7893$
& $(\lambda_3,\lambda_4)=(0.040,0.960)$
& $0.78939$\\

$10$
& $0.7885$
& $(\lambda_4,\lambda_5)=(0.986,0.014)$
& $0.78866$\\

$11$
& $0.7878$
& $(\lambda_4,\lambda_5)=(0.925,0.075)$
& $0.78790$\\

$12$
& $0.7871$
& $(\lambda_4,\lambda_5)=(0.878,0.122)$
& $0.78728$\\

$13$
& $0.7866$
& $(\lambda_4,\lambda_5)=(0.832,0.168)$
& $0.78678$\\

$14$
& $0.7863$
& $(\lambda_4,\lambda_5)=(0.800,0.200)$
& $0.78636$\\

$15$
& $0.7859$
& $(\lambda_4,\lambda_5)=(0.769,0.231)$
& $0.78601$\\

$16$
& $0.7856$
& $(\lambda_4,\lambda_5)=(0.742,0.258)$
& $0.78571$\\

\bottomrule
\end{tabular}
\end{table}

Fix $K \in \{3,\dots, 16\}$.
Let $P(\rho)$ be the collision probability of the mixture for $K$, and for $M\in \mathbb{N}_0$,
let $\underline P_M(\rho)$ and $\overline P_M(\rho)$ denote the certified lower and upper bounds associated with this mixture, as defined in Corollary~\ref{cor:theory-P-LB-UB}.
By Lemma~\ref{lem:machinery}, it remains to verify the hypotheses~\eqref{eq:hypothesis-positive} and \eqref{eq:hypothesis-negative} with some $M\in \mathbb{N}_0$, $a=a_K$, and $\gamma=\gamma_K$.
We verify them in exact rational arithmetic (using Bernstein certificates for the hypothesis for the positive-edge inequality):
\begin{proposition}
For each $K=3,\dots, 16$, with the exact rational lower and upper bounds on $c_{k,m}$ generated as described in Section~\ref{subsec:coefficient-implementation},
\begin{equation*}
\underline{P}_{32}(\rho)\geq \gamma_K s^+_{a_K}(\rho) = \frac{\gamma_K(\rho-a_K)}{1-a_K}\quad \text{for all }\, \rho \in [a_K,1],
\end{equation*}
and
\begin{equation*}
1-\overline{P}_{32}(a_K)\geq \gamma_K,
\end{equation*}
where $\gamma_K$ is that presented in Table~\ref{tab:maxagree-k}.
\end{proposition}

Thus, we have the desired result:
\begin{theorem}\label{thm:maxagree-k}
For \maxagreek, Algorithm~\ref{alg:maxagree-k} has an expected approximation ratio of $\gamma_K$ presented in Table~\ref{tab:maxagree-k}.
\end{theorem}

The improvements in approximation ratios achieved by our algorithm are substantial, and are considerably larger than the improvement obtained for \maxagree.
The most notable case is $K=3$, for which our algorithm achieves an approximation ratio of $0.8151$,
substantially improving upon the previous state-of-the-art ratio of $0.77$ due to \citet{Swamy04}.

The mixtures $\bm{\lambda}^{(K)}$ exhibit an interesting and intuitive trend.
For $K=3$, the mixture is supported on the $2$-spokes and $3$-spokes roundings.
As $K$ increases, the support shifts first to the $3$-spokes and $4$-spokes roundings
and eventually to the $4$-spokes and $5$-spokes roundings,
while within each regime the probability mass gradually shifts toward the rounding with the larger number of spokes.
This is consistent with the mixture used for \maxagree.

Finally, the certified upper bounds $U_K^{\mathrm{cert}}$ show that
there is little room to improve the approximation ratios $\gamma_K$ by changing only the mixture of the \(k\)-spokes roundings
within the same edge-wise analysis framework. For details, see Section~\ref{sec:UB}.

\subsection{Max \texorpdfstring{$K$}{K}-Cut}\label{sec:max-k-cut}

Max $K$-Cut is the special case of the general clustering problem
obtained by taking $E_+=\emptyset$, $E_-=E$, and $\mathcal{F}$ to be the family of all clusterings of $V$ with at most $K$ clusters,
for a fixed integer $2\leq K\leq n$.
Since the case $K=2$ is equivalent to Max Cut, we henceforth focus on $K\geq 3$.
Without loss of generality, we assume $K\leq n$.
For this problem, we revisit the $K$-Cut algorithm presented by \citet*{Frieze+97}:
in our context, the algorithm solves $\mathrm{SDP}(a_K)$ and simply applies the $K$-spokes rounding (see Algorithm~\ref{alg:max-k-cut}).
\begin{algorithm}[h]
\SetKwInput{Input}{Input}
\caption{The $K$-Cut algorithm~\cite{Frieze+97}}\label{alg:max-k-cut}
\Input{$V$, $E_-$, $w_{ij}>0$ for $\{i,j\}\in E_-$.}
Solve $\mathrm{SDP}(a_K)$ and obtain an optimal solution $(\bm{v}^*_i)_{i\in V}$\;
$\mathcal{C}_{K}\leftarrow \text{$K$-spokes}((\bm{v}^*_i)_{i\in V})$\;
\Return $\mathcal{C}_{K}$.
\end{algorithm}

Our goal in this section is to significantly advance our understanding of the approximation guarantee of this algorithm.
\citet{Frieze+97} showed that for every $K$, the algorithm admits an expected approximation ratio of
\begin{equation*}
\alpha^\mathrm{FJ}_K \coloneqq \min_{\rho\in [a_K,1)}A_K(\rho),
\end{equation*}
where
\begin{equation*}
A_K(\rho)=\frac{K(1-P_K(\rho))}{(K-1)(1-\rho)}\quad \text{for } \rho\in [a_K,1).
\end{equation*}

We first prove that for every $K\geq 3$, the minimum of $A_K(\rho)$ is achieved at the left endpoint $\rho=a_K$, yielding a simpler representation of $\alpha_K^{\mathrm{FJ}}$. This result was previously known only for $K=3$~\cite{Goemans+04,deKlerk+04} and has remained open for $K\geq 4$.
\begin{theorem}\label{thm:max-k-cut-approx-ratio}
For every $K\geq 3$, the minimum of $A_K(\rho)$ is achieved at the left endpoint $\rho=a_K$. Consequently, it holds that
\begin{equation*}
\alpha^\mathrm{FJ}_K = 1 - P_K(a_K).
\end{equation*}
\end{theorem}

\begin{proof}
Recall $R_k(\rho)$ defined in Lemma~\ref{lem:negative-ratio-monotonicity}. Then we have
\begin{equation*}
A_K(\rho)=\frac{K}{K-1}R_K(\rho).
\end{equation*}
Lemma~\ref{lem:negative-ratio-monotonicity} guarantees that $R_K(\rho)$ is nondecreasing over $\rho\in [-1/2,1)$. Since $-1/2\leq a_K$, so is $A_K(\rho)$ over $\rho \in [a_K,1)$. Therefore, its minimum is achieved at the left endpoint $\rho=a_K$.
Using this, we have
\begin{equation*}
\alpha^\mathrm{FJ}_K = A_K(a_K) = \frac{K(1-P_K(a_K))}{(K-1)(1-a_K)} = 1 - P_K\left(a_K\right).
\end{equation*}
This proves the theorem.
\end{proof}

Based on this, we can prove the tightness of $\alpha^\mathrm{FJ}_K$, meaning that it is impossible to certify an approximation ratio better than $\alpha^\mathrm{FJ}_K$ for Algorithm~\ref{alg:max-k-cut}. This was previously known only for $K=3$~\cite{Goemans+04}.
Note that the case $K=2$ is more classic~\cite{Karloff99}.
\begin{theorem}\label{thm:max-k-cut-FJ-tightness}
For every $K\geq 3$, the approximation ratio $\alpha^\mathrm{FJ}_K$ is tight for Algorithm~\ref{alg:max-k-cut}.
\end{theorem}

\begin{proof}
Consider the $K$-node complete graph with normalized edge weights, i.e., $w_{ij}=1$ for every edge $\{i,j\}$.
Clearly, the optimal value for this instance is $\binom{K}{2}$.
For $\mathrm{SDP}(a_K)$ on this instance, let us consider the solution $(\bm{v}'_i)_{i\in V}$ such that
\begin{equation*}
\bm{v}'_i = \sqrt{\frac{K}{K-1}}\left(\bm{\mathrm{e}}_i-\frac{1}{K}\bm{1}\right)\quad \text{for } i\in V,
\end{equation*}
where $\bm{\mathrm{e}}_i$ is the $i$-th standard basis vector of $\mathbb{R}^K$ and $\bm{1}$ is the all-ones vector in $\mathbb{R}^K$.
These vectors form the $(K-1)$-dimensional regular simplex:
$\bm{v}'_i\cdot \bm{v}'_i=1$ for all $i\in V$ and $\bm{v}'_i\cdot \bm{v}'_j= a_K$ for all distinct $i,j\in V$,
meaning that the solution $(\bm{v}'_i)_{i\in V}$ is feasible for $\mathrm{SDP}(a_K)$.
Indeed, the solution is optimal: its SDP objective value is $\binom{K}{2} s^-_{a_K}(a_K) = \binom{K}{2}$,
a trivial upper bound on the SDP optimal value.
Hence, we see that any optimal solution $(\bm{v}^*_i)_{i\in V}$ satisfies $\bm{v}^*_i\cdot \bm{v}^*_j=a_K$;
otherwise, it would not be able to achieve the optimal value $\binom{K}{2}$, a contradiction.
Thus, the expected objective value achieved by Algorithm~\ref{alg:max-k-cut} is $\binom{K}{2}(1-P_K(a_K))$,
which equals $\binom{K}{2}\alpha^\mathrm{FJ}_K$ by Theorem~\ref{thm:max-k-cut-approx-ratio}.
Since the optimal value is $\binom{K}{2}$, the ratio between the expected objective value of the algorithm and the optimal value is exactly $\alpha^\mathrm{FJ}_K$.
This proves the theorem.
\end{proof}

In view of the above tightness result, it is all the more important to determine the value of $\alpha^\mathrm{FJ}_K$ precisely.
In what follows, we establish rigorous lower and upper bounds $\underline{\alpha}^\mathrm{FJ}_K$ and $\overline{\alpha}^\mathrm{FJ}_K$, respectively,
as presented in Table~\ref{tab:max-k-cut}.
Using Theorem~\ref{thm:max-k-cut-approx-ratio} and Proposition~\ref{prop:theory-Pk-LB-UB}, certifying a rigorous lower bound $\underline{\alpha}^\mathrm{FJ}_K$ reduces to verifying the hypothesis~\eqref{eq:hypothesis-negative} with some $M\in \mathbb{N}_0$, $a=a_K$, and $\gamma=\underline \alpha^\mathrm{FJ}_K$.
A rigorous upper bound can similarly be certified using the lower bound $\underline{P}_{K,M}(a_K)$ from Proposition~\ref{prop:theory-Pk-LB-UB}.
We summarize these two conditions in the following proposition:
\begin{proposition}
Fix $K\geq 3$ and $M\in \mathbb{N}_0$. Suppose that for each $m=0,\dots, M$, we have certified lower and upper bounds $\underline{c}_{K,m}\leq c_{K,m}\leq \overline{c}_{K,m}$.
Recall $\underline{P}_{K,M}(\rho)$ and $\overline{P}_{K,M}(\rho)$ defined in Proposition~\ref{prop:theory-Pk-LB-UB}.
If
\begin{equation}\label{eq:hypothesis-max-k-cut}
\underline{\alpha}^\mathrm{FJ}_K \leq 1 - \overline{P}_{K,M}(a_K) \quad
\text{and}\quad 1 - \underline{P}_{K,M}(a_K) \leq \overline{\alpha}^\mathrm{FJ}_K,
\end{equation}
then it holds that
\begin{equation*}
\underline{\alpha}^\mathrm{FJ}_K\leq \alpha^\mathrm{FJ}_K\leq \overline{\alpha}^\mathrm{FJ}_K.
\end{equation*}
\end{proposition}

\begin{proof}
By Proposition~\ref{prop:theory-Pk-LB-UB}, we have $\underline{P}_{K,M}(a_K)\leq P_K(a_K)\leq \overline{P}_{K,M}(a_K)$.
Hence, the hypotheses in~\eqref{eq:hypothesis-max-k-cut} imply $\underline{\alpha}^\mathrm{FJ}_K \leq 1 - P_{K}(a_K) \leq \overline{\alpha}^\mathrm{FJ}_K$,
and by Theorem~\ref{thm:max-k-cut-approx-ratio}, $\underline{\alpha}^\mathrm{FJ}_K\leq \alpha^\mathrm{FJ}_K\leq \overline{\alpha}^\mathrm{FJ}_K$.
\end{proof}

\begin{table}[t]
\centering
\caption{Rigorous lower and upper bounds $\underline{\alpha}^\mathrm{FJ}_K$ and $\overline{\alpha}^\mathrm{FJ}_K$ on the tight approximation ratio $\alpha^\mathrm{FJ}_K$ of Algorithm~\ref{alg:max-k-cut}. The best previously known explicit approximation ratios from \cite{Goemans+04,Frieze+97} are raw values reported there. The best previously known ratios from \cite{Newman18} are the rounded-down values computed using the closed-form representation of the collision probability of the rounding scheme proposed there. The numerical estimates of $\alpha^\mathrm{FJ}_K$ are obtained through numerical integrations~\cite{deKlerk+04}, and thus are not necessarily rigorous.}\label{tab:max-k-cut}
\begin{tabular}{rcclc}
\toprule
$K$ & $\underline{\alpha}^\mathrm{FJ}_K$ & $\overline{\alpha}^\mathrm{FJ}_K$ & Previous best ratio & Numerical estimate \\
\midrule
$3$
& $0.836008$
& $0.836009$
& $0.836008$~\cite{Goemans+04}
& $0.836008$~\cite{deKlerk+04} \\

$4$
& $0.857487$
& $0.857488$
& $0.850304$~\cite{Frieze+97}
& $0.857487$~\cite{deKlerk+04} \\

$5$
& $0.876609$
& $0.876610$
& $0.874243$~\cite{Frieze+97}
& $0.876610$~\cite{deKlerk+04} \\

$6$
& $0.891534$
& $0.891535$
& $0.877011$~\cite{Newman18}
& $0.891543$~\cite{deKlerk+04} \\

$7$
& $0.903259$
& $0.903260$
& $0.889378$~\cite{Newman18}
& $0.903259$~\cite{deKlerk+04} \\

$8$
& $0.912664$
& $0.912665$
& $0.899750$~\cite{Newman18}
& $0.912664$~\cite{deKlerk+04} \\

$9$
& $0.920366$
& $0.920367$
& $0.908479$~\cite{Newman18}
& $0.920367$~\cite{deKlerk+04} \\

$10$
& $0.926786$
& $0.926787$
& $0.926642$~\cite{Frieze+97}
& $0.926788$~\cite{deKlerk+04} \\

$11$
& $0.932222$
& $0.932223$
& $0.922227$~\cite{Newman18}
& NA \\

$12$
& $0.936885$
& $0.936886$
& $0.927708$~\cite{Newman18}
& NA \\

$13$
& $0.940931$
& $0.940932$
& $0.932486$~\cite{Newman18}
& NA \\

$14$
& $0.944476$
& $0.944477$
& $0.936684$~\cite{Newman18}
& NA \\

$15$
& $0.947609$
& $0.947610$
& $0.940400$~\cite{Newman18}
& NA \\

$16$
& $0.950398$
& $0.950399$
& $0.943710$~\cite{Newman18}
& NA \\

\bottomrule
\end{tabular}
\end{table}

We verify the hypotheses in~\eqref{eq:hypothesis-max-k-cut} in exact rational arithmetic:
\begin{proposition}
For each $K=3,\dots, 16$, with the exact rational lower and upper bounds on $c_{K,m}$ generated as described in Section~\ref{subsec:coefficient-implementation},
\begin{equation*}
\underline{\alpha}^\mathrm{FJ}_K \leq 1 - \overline{P}_{K,32}(a_K) \quad
\text{and}\quad 1 - \underline{P}_{K,32}(a_K) \leq \overline{\alpha}^\mathrm{FJ}_K,
\end{equation*}
where $\underline{\alpha}^\mathrm{FJ}_K$ and $\overline{\alpha}^\mathrm{FJ}_K$ are those presented in Table~\ref{tab:max-k-cut}.
\end{proposition}

Thus, we have the desired result:
\begin{theorem}\label{thm:max-k-cut}
For each $K=3,\dots, 16$, it holds that $\underline{\alpha}^\mathrm{FJ}_K\leq \alpha^\mathrm{FJ}_K\leq \overline{\alpha}^\mathrm{FJ}_K$, where $\underline{\alpha}^\mathrm{FJ}_K$ and $\overline{\alpha}^\mathrm{FJ}_K$ are those presented in Table~\ref{tab:max-k-cut}.
\end{theorem}

As in Table~\ref{tab:max-k-cut}, the gap between our rigorous lower and upper bounds $\underline{\alpha}^\mathrm{FJ}_K$ and $\overline{\alpha}^\mathrm{FJ}_K$ is at most $10^{-6}$, thereby pinning down the tight approximation ratio $\alpha^\mathrm{FJ}_K$ to within $10^{-6}$.
In particular, our lower bounds improve upon the previous state-of-the-art explicit approximation ratios for $K=4,\dots, 16$.
It is worth noting that in the literature, the numerical estimates of $\alpha^\mathrm{FJ}_K$ given by \citet*{deKlerk+04} have often been treated as rigorous lower bounds on $\alpha^\mathrm{FJ}_K$. However, for $K=6$ and $10$, their reported estimates exceed our rigorous upper bounds $\overline{\alpha}^\mathrm{FJ}_K$, confirming that these values should be regarded only as numerical approximations rather than certified lower bounds.

\subsection{Modularity Maximization}\label{sec:modularity}

Let $G=(V,E)$ be a simple undirected graph consisting of $n=|V|$ nodes and $m=|E|$ edges. Assume that $G$ is nonempty, i.e., $|E|\geq 1$.
We denote by $A$ the adjacency matrix of $G$, that is, $A_{ij}=1$ if $\{i,j\}\in E$ and $A_{ij}=0$ otherwise.
We also denote by $d_i$ the degree of $i\in V$, i.e., $d_i=\sum_{j\in V}A_{ij}$.
Let $\mathcal{C}$ be a clustering of $V$.
The modularity of $\mathcal{C}$, introduced by \citet{Newman+04}, can be written as
\begin{equation*}
Q(\mathcal{C})=\frac{1}{2m}\sum_{(i,j)\in V^2}\left(A_{ij}-\frac{d_id_j}{2m}\right)\1[\mathcal{C}(i)=\mathcal{C}(j)],
\end{equation*}
where $V^2$ denotes $V\times V$.
Modularity maximization is the problem to find a clustering $\mathcal{C}$ of $V$ that maximizes $Q(\mathcal{C})$.
Let $\mathrm{OPT}$ be the optimal value of the problem.

Let us introduce some notation. For simplicity, we define
\begin{equation*}
q_{ij}=\frac{A_{ij}}{2m}-\frac{d_id_j}{4m^2}\quad \text{for each } (i,j)\in V^2.
\end{equation*}
Then the modularity can be rewritten as
\begin{equation*}
Q(\mathcal{C})=\sum_{(i,j)\in V^2}q_{ij}\,\1[\mathcal{C}(i)=\mathcal{C}(j)].
\end{equation*}
Moreover, the set $V^2$ can be divided into two disjoint subsets:
\begin{equation*}
V^2_{\geq 0}=\{(i,j)\in V^2 : q_{ij}\geq 0\}\quad \text{and}\quad V^2_{<0}=\{(i,j)\in V^2 : q_{ij}<0\}.
\end{equation*}
Clearly, we have $\sum_{(i,j)\in V^2}q_{ij}=0$, so
\begin{equation*}
q := \sum_{(i,j)\in V^2_{\geq 0}}q_{ij} = \sum_{(i,j)\in V^2_{<0}}-q_{ij}.
\end{equation*}
It is easy to see that $0<q<1$.

We employ the following standard SDP relaxation used by \citet*{Dinh+15} and \citet*{Kawase+21}:
\begin{alignat*}{4}
&\text{maximize} &\quad &\sum_{(i,j)\in V^2} q_{ij}(\bm{v}_i\cdot \bm{v}_j)\\
&\text{subject to} &&\bm{v}_i\cdot \bm{v}_i = 1 &&\text{for all } i\in V,\\
&                  &&\bm{v}_i\cdot \bm{v}_j \geq 0 &&\text{for all } i,j\in V,\\
&                  &&\bm{v}_i\in \mathbb{R}^n &&\text{for all } i\in V.
\end{alignat*}
Let $\mathrm{OPT}_\mathrm{SDP}$ be the optimal value of this SDP,
and let $(\bm{v}^*_i)_{i\in V}$ be its optimal solution. Define
\begin{equation*}
z^*_+=\frac{1}{q}\sum_{(i,j)\in V^2_{\geq 0}}q_{ij}(\bm{v}^*_i\cdot \bm{v}^*_j)\quad
\text{and}\quad
z^*_-=\frac{1}{q}\sum_{(i,j)\in V^2_{<0}}q_{ij}(\bm{v}^*_i\cdot \bm{v}^*_j).
\end{equation*}
Clearly, we have $0\leq z^*_+\leq 1$ and $-1\leq z^*_-\leq 0$.
Moreover, $\mathrm{OPT}_\mathrm{SDP}=q(z^*_+ + z^*_-)\geq \mathrm{OPT}$.

We now present our algorithm.
It first solves the SDP relaxation and obtains an optimal solution $(\bm{v}^*_i)_{i\in V}$.
Then, it applies the $7$-spokes rounding to $(\bm{v}^*_i)_{i\in V}$ and outputs the resulting clustering $\mathcal{C}_7$.
The algorithm is presented in Algorithm~\ref{alg:modularity}.
\begin{algorithm}[h]
\SetKwInput{Input}{Input}
\caption{Our algorithm for modularity maximization}\label{alg:modularity}
\Input{$G=(V,E)$}
Solve the SDP relaxation and obtain an optimal solution $(\bm{v}^*_i)_{i\in V}$\;
$\mathcal{C}_7\leftarrow \text{$7$-spokes}((\bm{v}^*_i)_{i\in V})$\;
\Return $\mathcal{C}_7$.
\end{algorithm}

Our goal in this section is to prove that Algorithm~\ref{alg:modularity} has an expected additive approximation error of $0.3790$.
Unlike our analyses for the previous problems, the analysis here is not edge-wise.
Indeed, we use the following key lemma, which is the $k$-spokes rounding counterpart of \cite[Lemma~3]{Kawase+21}:
\begin{lemma}\label{lem:modularity-key}
For any integer $k\geq 1$, it holds that
\begin{equation*}
\E[Q(\mathcal{C}_{k})] \geq \mathrm{OPT} - q \left(z^*_+ - P_k(z^*_+) + \frac{1}{k} \right).
\end{equation*}
\end{lemma}

\begin{proof}
By the Hermite expansion~\eqref{eq:Pk-hermite-series}, we know that $P_k(\rho)$ is convex over $[0,1]$. Then we have
\begin{align*}
\E[Q(\mathcal{C}_k)]
&=\sum_{(i,j)\in V^2}q_{ij}P_k(\bm{v}^*_i\cdot \bm{v}^*_j)\\
&=q\left(\sum_{(i,j)\in V^2_{\geq 0}}\frac{q_{ij}}{q}P_k(\bm{v}^*_i\cdot \bm{v}^*_j) + \sum_{(i,j)\in V^2_{<0}}\frac{q_{ij}}{q}P_k(\bm{v}^*_i\cdot \bm{v}^*_j)\right)\\
&\geq q\left(P_k(z^*_+) - \sum_{(i,j)\in V^2_{<0}}\left(-\frac{q_{ij}}{q}\right)P_k(\bm{v}^*_i\cdot \bm{v}^*_j)\right)\\
&\geq q\left(P_k(z^*_+) - \sum_{(i,j)\in V^2_{<0}}\left(-\frac{q_{ij}}{q}\right)\left(P_k(0)+(1-P_k(0))(\bm{v}^*_i\cdot \bm{v}^*_j)\right)\right)\\
&= q\left(P_k(z^*_+) - P_k(0) + (1-P_k(0))z^*_- \right),
\end{align*}
where the first inequality follows from Jensen's inequality,
and the second inequality follows from the fact that the convex function never goes beyond the straight line connecting its extreme points.
We then relate this to $\mathrm{OPT}_\mathrm{SDP}$ (and thus to $\mathrm{OPT}$):
\begin{align*}
&q\left(P_k(z^*_+) - P_k(0) + (1-P_k(0))z^*_- \right) \\
&= \left(\mathrm{OPT}_\mathrm{SDP}-q(z^*_+ + z^*_-)\right) + q\left(P_k(z^*_+) - P_k(0) + (1-P_k(0))z^*_- \right) \\
&= \mathrm{OPT}_\mathrm{SDP} - q\left(z^*_+ - P_k(z^*_+) + P_k(0) + P_k(0)z^*_- \right) \\
&\geq \mathrm{OPT}_\mathrm{SDP} - q\left(z^*_+ - P_k(z^*_+) + P_k(0) \right) \\
&\geq \mathrm{OPT} - q\left(z^*_+ - P_k(z^*_+) + \frac{1}{k} \right),
\end{align*}
where the first inequality follows from the nonpositivity of $z^*_-$ and the positivity of $P_k(0)=1/k$.
This completes the proof.
\end{proof}

This lemma suggests that the additive approximation error $0.3790$ can be certified using the lower bound $\underline P_{7,M}(\rho)$ from Proposition~\ref{prop:theory-Pk-LB-UB}:
\begin{proposition}
Fix $M\in \mathbb{N}_{0}$. Suppose that for each $m=0,\dots, M$, we have certified lower bounds $\underline{c}_{7,m}\leq c_{7,m}$.
Recall $\underline{P}_{k,M}(\rho)$ defined in Proposition~\ref{prop:theory-Pk-LB-UB}.
If
\begin{equation}\label{eq:modularity-hypothesis}
x - \underline{P}_{7,M}(x) + \frac{1}{7} \leq 0.3790 \quad \text{for all } x \in [0,1],
\end{equation}
then it holds that
\begin{equation*}
\E[Q(\mathcal{C}_7)] > \mathrm{OPT} - 0.3790.
\end{equation*}
\end{proposition}

\begin{proof}
By Proposition~\ref{prop:theory-Pk-LB-UB}, we have $P_7(x)\geq \underline{P}_{7,M}(x)$ for all $x\in [0,1]$.
Together with Lemma~\ref{lem:modularity-key}, we have
\begin{equation*}
\E[Q(\mathcal{C}_7)] \geq \mathrm{OPT} - q\left(z^*_+ - P_7(z^*_+) + \frac{1}{7} \right)
\geq \mathrm{OPT} - q\left(z^*_+ - \underline{P}_{7,M}(z^*_+) + \frac{1}{7} \right).
\end{equation*}
Hence, the hypothesis~\eqref{eq:modularity-hypothesis} implies $\E[Q(\mathcal{C}_7)]\geq \mathrm{OPT}-0.3790\,q>\mathrm{OPT}-0.3790$.
\end{proof}

We verify the hypothesis~\eqref{eq:modularity-hypothesis} in exact rational arithmetic using Bernstein certificates:
\begin{proposition}
With the exact rational lower bounds on $c_{k,m}$ generated as described in Section~\ref{subsec:coefficient-implementation},
\begin{equation*}
x - \underline{P}_{7,32}(x) + \frac{1}{7} \leq 0.3790 \quad \text{for all } x \in [0,1].
\end{equation*}
\end{proposition}

Thus, we have the desired result:
\begin{theorem}\label{thm:modularity}
Algorithm~\ref{alg:modularity} has an expected additive approximation error of $0.3790$ for modularity maximization.
\end{theorem}

Finally, we note that as long as the analysis relies on Lemma~\ref{lem:modularity-key} followed only by the bound $q<1$,
Algorithm~\ref{alg:modularity} is nearly optimal within the family consisting of the $k$-spokes roundings for $1\leq k\leq 32$
and the $k$-hyperplane roundings for $k\geq 1$.
For any rounding with collision probability $P(\rho)$ that is convex over $\rho \in [0,1]$, define
\begin{equation*}
\mathrm{error}_P(x)=x-P(x)+P(0)\quad \text{for } x\in [0,1].
\end{equation*}
This is the additive approximation error derived by applying the analysis in Lemma~\ref{lem:modularity-key} to the general rounding.
To show the near optimality of Algorithm~\ref{alg:modularity}, since $\mathrm{error}_P$ is linear in $P$,
it suffices to find a single point $x'$ at which every rounding in the family has a large error close to $0.3790$.
We use $x'=4/5$.

For the $k$-spokes roundings, $P_k(0)=1/k$.
With the exact rational lower and upper bounds on $c_{k,m}$
generated as described in Section~\ref{subsec:coefficient-implementation},
we can verify the following in exact rational arithmetic:
\begin{equation*}
\mathrm{error}_{P_k}(4/5) = \frac45-P_k(4/5)+\frac1k \geq \frac45-\overline{P}_{k,24}(4/5)+\frac1k \ge 0.3785 \quad \text{for each } k=1,\ldots,32.
\end{equation*}
For the $k$-hyperplane roundings, letting $P_k^{\mathrm{GW}}(x)$ be its collision probability,
\begin{equation*}
P_k^{\mathrm{GW}}(x)
=
\left(1-\frac{\arccos x}{\pi}\right)^k\quad \text{and}\quad  P_k^{\mathrm{GW}}(0)=2^{-k}.
\end{equation*}
Since $\cos(\pi/5)=(1+\sqrt5)/4>4/5$, we have
\begin{equation*}
1-\frac{\arccos(4/5)}{\pi}<\frac45,
\end{equation*}
and therefore,
\begin{equation*}
\mathrm{error}_{P_k^{\mathrm{GW}}}(4/5) > \frac45-\left(\frac45\right)^k+2^{-k} \ge \frac{41}{100} >0.3785\quad \text{for every }k\geq 1.
\end{equation*}
Thus, for any mixture over the $k$-spokes roundings with $1\le k\le 32$ and the $k$-hyperplane roundings with $k\ge 1$, we have
\begin{equation*}
\mathrm{error}_P(4/5) \ge 0.3785.
\end{equation*}
Consequently, under the analysis of Lemma~\ref{lem:modularity-key} followed only by the bound $q<1$,
no such mixture can certify an additive approximation error smaller than $0.3785$,
verifying the near-optimality of Algorithm~\ref{alg:modularity} within the family under the same analysis.

\section{Inapproximability of \maxagree\ under the UGC}
\label{sec:raghavendra-gap}
We follow the approach outlined in Section~\ref{subsec:raghavendra-overview}
and give the details here.
Write $\operatorname{SDP}_{\mathrm{std}} = \operatorname{SDP}(0)$ for the standard
relaxation~\eqref{eq:maxagree_sdp} with $a=0$.  It maximizes
\[
  \sum_{ij\in E_+}w_{ij}\langle v_i,v_j\rangle
  +\sum_{ij\in E_-}w_{ij}\bigl(1-\langle v_i,v_j\rangle\bigr)
\]
over unit vectors $(v_i)_{i\in V}$ such that $\langle v_i,v_j\rangle\geq0$
for every pair of nodes.

Throughout this section, we allow nonnegative edge weights and assume
$E_+\cup E_-=\binom{V}{2}$, adding zero-weight edges on missing pairs.
This leaves the value of every clustering unchanged, and it matches
$\operatorname{SDP}_{\mathrm{std}}$, which constrains every pair of nodes,
not only the edges.  Accordingly, the canonical SDP below has a local
distribution for every pair of nodes.  The convention does not change the
problem, but it does change the canonical SDP in Raghavendra's original form,
which has local distributions only on the pairs of positive weight.
To establish an unconditional integrality gap, we need to take care of this distinction
(see Section~\ref{subsec:ugc-gap-transfer}).

For a maximization relaxation, we use the convention
\[
  \operatorname{gap}(\operatorname{SDP})
  :=\sup_I\frac{\operatorname{SDP}(I)}{\operatorname{OPT}(I)}\geq 1,
\]
where $I$ ranges over instances. Thus,
it corresponds to an integrality ratio of
$1/\operatorname{gap}$.

\begin{theorem}[SDP integrality gap and UGC hardness]
\label{thm:maxagree-sdp-gap-hardness}
Let $\operatorname{SDP}_{\mathrm{std}}$ be the standard SDP relaxation in~\eqref{eq:maxagree_sdp},
and let
$
 \alpha_{\mathrm{SDP}}
 :=\frac1{\operatorname{gap}(\operatorname{SDP}_{\mathrm{std}})}.
$
Then
    $
 \alpha_{\mathrm{SDP}}\leq 0.802.
 $

Moreover, for every constant $\varepsilon>0$, it is UG-hard to approximate
weighted \maxagree\ within a factor $\alpha_{\mathrm{SDP}}+\varepsilon$.  In particular, it is UG-hard
to achieve a factor $0.802+\varepsilon$.
\end{theorem}

\subsection{Reduction to a Fixed Number of Clusters}
For an integer $k\geq 2$, let $\operatorname{OPT}_k(I)$ be the optimum
of $\maxagree[k]$, and let $\operatorname{OPT}(I)$ be the unrestricted
optimum.
The restriction to a fixed alphabet loses only a factor $1-1/k$:

\begin{lemma}
\label{lem:hash-clusters}
For every \maxagree\ instance $I$ and every $k\geq 2$,
\[
  \operatorname{OPT}_k(I)
  \geq \left(1-\frac1k\right)\operatorname{OPT}(I).
\]
\end{lemma}

\begin{proof}
Independently hash every cluster of a given clustering $\mathcal C$ into one
of $k$ buckets, and merge clusters that receive the same bucket.
Let $\operatorname{hash}_k(\mathcal C)$ be the resulting clustering.
Let $A_+$ and $A_-$ be, respectively, the weights of the positive and negative edges
satisfied by $\mathcal C$.  Every satisfied positive edge remains satisfied.
The endpoints of a satisfied negative edge are in distinct original clusters,
and their buckets differ with probability $1-1/k$.  Edges not satisfied by
$\mathcal C$ can only increase the resulting value.  Consequently,
\[
  \mathbb E[\mathsf{val}(\operatorname{hash}_k(\mathcal C))]
  \geq A_++\left(1-\frac1k\right)A_-
  \geq \left(1-\frac1k\right)\mathsf{val}(\mathcal C).
\]
Applying this to an optimal clustering yields the result.
\end{proof}
The procedure can be derandomized: assign the clusters of an optimal \maxagree\ solution to $k$ buckets, each time choosing the bucket with the least negative edge weight to previously assigned clusters.

\subsection{The Canonical SDP for \texorpdfstring{$\maxagree[k]$}{MaxAgree[k]}}
\label{subsec:canonical-sdp}
We use the canonical CSP relaxation
of~\citet{Raghavendra08}, called the LC relaxation
in~\citet[Section~4.5]{Raghavendra09Thesis}, with the complete weighted graph
convention above: every pair of nodes, including pairs of zero weight, is
treated as a constraint and receives a local distribution.  There is one vector $u_{i,a}$ for each node $i$ and
label $a\in[k]$.  By~\citet[Observations~4.5.1--4.5.2]{Raghavendra09Thesis},
we may assume that the label vectors at each node are pairwise
orthogonal and sum to a common unit vector $u_0$.
Since each constraint involves two variables, its local distribution is
determined by the label-vector inner products.  Eliminating these
probability variables gives the equivalent vector program
\begin{alignat}{3}
\operatorname{SDP}_k(I)=\max\quad
 &\sum_{ij\in E_+}w_{ij}\sum_{a\in[k]}
       \langle u_{i,a},u_{j,a}\rangle
 +\sum_{ij\in E_-}w_{ij}\sum_{\substack{a,b\in[k]\\a\neq b}}
       \langle u_{i,a},u_{j,b}\rangle                                     \notag\\
\text{subject to}\quad
 &\|u_0\|^2=1,                                                             \notag\\
 &\sum_{a\in[k]}u_{i,a}=u_0
 &&(i\in V),                                                               \label{eq:basic-csp-sdp}\\
 &\langle u_{i,a},u_{i,b}\rangle=0
 &&(i\in V,\ a\neq b),                                                     \notag\\
 &\langle u_{i,a},u_{j,b}\rangle\geq0
 &&(ij\in\binom{V}{2},\ a,b\in[k]).                                      \notag
\end{alignat}
Conversely, given vectors feasible for~\eqref{eq:basic-csp-sdp}, the numbers
$
 \mu_{ij}^{a,b}:=\langle u_{i,a},u_{j,b}\rangle
 $
are nonnegative, sum to $\|u_0\|^2=1$, and have marginals $\|u_{i,a}\|^2$ and
$\|u_{j,b}\|^2$.  Together with the orthogonality of the label vectors at each
node, this is exactly what the LC relaxation requires of the local
distribution $\mu_{ij}$ of the pair $ij$.  The positive and negative payoffs
are respectively $\sum_a\mu_{ij}^{a,a}$ and $\sum_{a\neq b}\mu_{ij}^{a,b}$.

\paragraph{Symmetrization.}
The SDP admits a useful exact reduction to one scalar $y_{ij}$ per pair.
Both payoffs are invariant under applying the same permutation to all $k$
labels.  Take the normalized direct sum $\bigoplus$ of the $k!$ permutations
for each vector:
\[
  u'_{i,a}=\frac1{\sqrt{k!}}\bigoplus_{\pi\in S_k}u_{i,\pi(a)},\qquad
  u'_{0}=\frac1{\sqrt{k!}}\bigoplus_{\pi\in S_k}u_{0}.
\]
Their inner products are invariant under label permutations.  Hence, fixing
two distinct labels $1$ and $2$,
\begin{equation*}%
  \langle u'_{i,a},u'_{j,b}\rangle
  =\frac1{k!}\sum_{\pi\in S_k}
  \langle u_{i,\pi(a)},u_{j,\pi(b)}\rangle =
  \begin{cases}
  \langle u'_{i,1},u'_{j,1}\rangle ,&a=b,\\
  \langle u'_{i,1},u'_{j,2}\rangle ,&a\ne b.
  \end{cases}
\end{equation*}
This symmetrization preserves feasibility and the objective function value.
Moreover, for all $i, j \in V$,
\[
  1=\|u'_0\|^2
  =\Bigl\langle\sum_{a}u'_{i,a},\sum_{b}u'_{j,b}\Bigr\rangle
  =k\langle u'_{i,1},u'_{j,1}\rangle+k(k-1)\langle u'_{i,1},u'_{j,2}\rangle.
\]
Define $y_{ij}:=k\langle u'_{i,1},u'_{j,1}\rangle\in[0,1]$.  Then $y_{ii}=1$
and
\begin{equation}
\label{eq:symmetric-label-inner-products}
  \langle u'_{i,a},u'_{j,b}\rangle=
  \begin{cases}
  \frac{y_{ij}}{k},&a=b,\\
  \frac{1-y_{ij}}{k(k-1)},&a\ne b.
  \end{cases}
\end{equation}
Let
\[
  v_i=\sqrt{\frac{k}{2}}\,(u'_{i,1}-u'_{i,2}).
\]
Using~\eqref{eq:symmetric-label-inner-products},
\[
  \langle v_i,v_j\rangle
  =\frac{k}{2}\left(\frac{2y_{ij}}k
  -\frac{2(1-y_{ij})}{k(k-1)}\right)
  =\frac{ky_{ij}-1}{k-1},
  \qquad\text{that is,}\qquad
  y_{ij}=\frac1k+\Bigl(1-\frac1k\Bigr)\langle v_i,v_j\rangle.
\]
In particular, $\langle v_i,v_j\rangle\geq-1/(k-1)$ and
$\langle v_i,v_i\rangle=1$.  With $a=-1/(k-1)$, the payoffs
in~\eqref{eq:maxagree_sdp} are $s^+_a(\rho)=\frac1k+(1-\frac1k)\rho$ and
$s^-_a(\rho)=(1-\frac1k)(1-\rho)$, so $s^+_a(\langle v_i,v_j\rangle)=y_{ij}$
and $s^-_a(\langle v_i,v_j\rangle)=1-y_{ij}$.  Thus the $v_i$ are feasible
for $\operatorname{SDP}(-1/(k-1))$, the relaxation used by our $\maxagree[k]$
algorithm, and have the same value.

Conversely, let $(v_i)_{i\in V}$, with $v_i\in\mathbb R^n$, be feasible for
$\operatorname{SDP}(-1/(k-1))$, and put
$y_{ij}=\frac1k+(1-\frac1k)\langle v_i,v_j\rangle$, so that
$0\leq y_{ij}\leq1$.  Let $e_1,\ldots,e_k$ be the standard basis of
$\mathbb R^k$, put $\bar e=\frac1k\sum_ae_a$, and let $u_0$ be a unit vector
orthogonal to $\mathbb R^k\otimes\mathbb R^n$ (here $\otimes$ is the tensor
product).  The vectors
\begin{equation}
\label{eq:label-vectors-from-sdp-a}
  u_{i,a}=\frac1k\,u_0+\frac1{\sqrt k}\,(e_a-\bar e)\otimes v_i
  \qquad(i\in V,\ a\in[k])
\end{equation}
satisfy $\sum_au_{i,a}=u_0$, because $\sum_a(e_a-\bar e)=0$, and, since
$\langle e_a-\bar e,e_b-\bar e\rangle=\mathbf1[a=b]-\frac1k$,

\begin{equation}\label{eq:uiajb}
\begin{aligned}
  \langle u_{i,a},u_{j,b}\rangle
  &=\frac1{k^2}+\frac1k\Bigl(\mathbf1[a=b]-\frac1k\Bigr)
    \langle v_i,v_j\rangle.
\end{aligned}
\end{equation}
For $i=j$ this equals $\mathbf1[a=b]/k$, so the label vectors of each node
are orthogonal; for $i\neq j$ it equals $y_{ij}/k$ if $a=b$ and
$(1-y_{ij})/(k(k-1))\geq0$ otherwise.  Hence these vectors are feasible
for~\eqref{eq:basic-csp-sdp}, and their value
$\sum_{ij\in E_+}w_{ij}y_{ij}+\sum_{ij\in E_-}w_{ij}(1-y_{ij})$ is the value
of $(v_i)$ in $\operatorname{SDP}(-1/(k-1))$.  Together with the forward
direction, this shows that $\operatorname{SDP}_k(I)$ is the
optimum of $\operatorname{SDP}(-1/(k-1))$ on $I$.

In terms of $Y=(y_{ij})$ and the all-ones matrix $J_n$, the same program
reads
\begin{alignat}{3}
\operatorname{SDP}_k(I)=\max\quad
 &\sum_{ij\in E_+}w_{ij}y_{ij}
  +\sum_{ij\in E_-}w_{ij}(1-y_{ij})                                      \notag\\
\text{subject to}\quad
 &y_{ii}=1 &&(i\in V),                                                     \label{eq:reduced-basic-sdp}\\
 &y_{ij}\geq0 &&(ij\in\binom{V}{2}),                                     \notag\\
 &kY-J_n\succeq0,                                                         \notag
\end{alignat}
because $(kY-J_n)/(k-1)$ is exactly the Gram matrix of the $v_i$: its
diagonal is $1$, and its entries are at least $-1/(k-1)$ precisely when
$y_{ij}\geq0$.

\paragraph{SDP comparison.}
Write $\operatorname{SDP}_{\mathrm{std}}(I)$ for the optimal value
$\mathrm{OPT}_{\mathrm{SDP}}$ of the standard SDP in~\eqref{eq:maxagree_sdp}
with parameter $a=0$.
The only difference from~\eqref{eq:reduced-basic-sdp} is that the standard
SDP requires $Y\succeq0$ in place of $kY-J_n\succeq0$.
The following comparison quantifies the effect of this additional
constraint, which vanishes as $k\to\infty$.

\begin{lemma}
\label{lem:basic-sdp-sandwich}
For every \maxagree\ instance $I$ and every $k\geq2$,
\[
  \left(1-\frac1k\right)\operatorname{SDP}_{\mathrm{std}}(I)
  \;\leq\;
  \operatorname{SDP}_k(I)
  \;\leq\;
  \operatorname{SDP}_{\mathrm{std}}(I).
\]
\end{lemma}

\begin{proof}
\emph{Upper bound.}  As noted above, the two programs differ only in their
positive semidefiniteness constraints, and $kY-J_n\succeq0$ implies
$Y\succeq J_n/k\succeq0$.  Hence every $Y$ feasible
for~\eqref{eq:reduced-basic-sdp} is feasible for
$\operatorname{SDP}_{\mathrm{std}}$, with the same value.

\emph{Lower bound.}  Let $Y$ be feasible for the
program defining $\operatorname{SDP}_{\mathrm{std}}$ and put
\[
  Y'=\left(1-\frac1k\right)Y+\frac1kJ_n,
  \qquad\text{that is,}\qquad
  y'_{ij}=y_{ij}+\frac{1-y_{ij}}{k}.
\]
Then $y'_{ii}=1$, $y'_{ij}\geq0$ on every pair, and
$kY'-J_n=(k-1)Y\succeq0$, so $Y'$ is feasible for
\eqref{eq:reduced-basic-sdp}.  On a positive edge
$y'_{ij}\geq(1-1/k)\,y_{ij}$, and on a negative edge
$1-y'_{ij}=(1-1/k)(1-y_{ij})$.  Hence the value of $Y'$ is at least
$(1-1/k)$ times the value of $Y$.
\end{proof}
The map $Y\mapsto Y'$ has a random hashing interpretation: if $y_{ij}$
is read as the probability that
$i$ and $j$ share a cluster, then $y'_{ij}$ is the probability that they
share a bucket after the clusters are hashed into $k$ buckets, and
\[
  \mu_{ij}(a,a)=\frac{y_{ij}}{k}+\frac{1-y_{ij}}{k^2},
  \qquad
  \mu_{ij}(a,b)=\frac{1-y_{ij}}{k^2}\quad(a\neq b),
\]
is the joint distribution of the two bucket labels; compare
\eqref{eq:symmetric-label-inner-products}.  As in
Lemma~\ref{lem:hash-clusters}, positive agreements are preserved and negative
agreements retain at least a $1-1/k$ fraction.

\subsection{UGC-Based Inapproximability}
\label{subsec:ugc-gap-transfer}
We prove the hardness statement of
Theorem~\ref{thm:maxagree-sdp-gap-hardness} in this subsection.
Raghavendra's canonical SDP has local distributions only on the pairs of
positive weight, which form the support of the instance.  Write
$\operatorname{SDP}^{\mathrm{cp}}_k(I)$, where ``$\mathrm{cp}$'' stands for constraint pairs,
for this relaxation, that is,~\eqref{eq:basic-csp-sdp} with cross-node nonnegativity required only
for pairs $ij$ of positive weight, and measure its integrality gap against
$\operatorname{OPT}_k$.  Equivalently, in~\eqref{eq:reduced-basic-sdp},
require $y_{ij}\geq0$ only on these pairs.  Every solution feasible for
$\operatorname{SDP}_k(I)$ is feasible for
$\operatorname{SDP}^{\mathrm{cp}}_k(I)$, with the same objective value, so
\[
 \operatorname{SDP}_k(I)\leq\operatorname{SDP}^{\mathrm{cp}}_k(I).
\]

Fix a weighted $\maxagree[k]$ instance $J$ of total constraint weight one
with $\operatorname{SDP}^{\mathrm{cp}}_k(J)\geq c$ and
$\operatorname{OPT}_k(J)\leq s$, where $0\leq s<c\leq1$.
Applied to the equality and inequality predicates on $[k]$,
Raghavendra's theorem~\citep{Raghavendra08} states that, for every
sufficiently small $\eta>0$, it is UG-hard to distinguish weighted
$\maxagree[k]$ instances $I$, again of total constraint weight one, with
\[
  \operatorname{OPT}_k(I)\geq c-\eta
  \qquad\text{from those with}\qquad
  \operatorname{OPT}_k(I)\leq s+\eta.
\]
Consequently, it is UG-hard to approximate $\maxagree[k]$ within any factor
strictly greater than $s/c$.  The same conclusion holds if the bound $c$
is established for $\operatorname{SDP}_k(J)$, since
$\operatorname{SDP}_k(J)\leq\operatorname{SDP}^{\mathrm{cp}}_k(J)$.
The instances $I$ may include constraints of both signs on a pair; merging them into one signed edge, as explained at the
end of the proof of Theorem~\ref{thm:gaussian-maxagree-gap}, only lowers the
ratio of the two thresholds, so the conclusion is unaffected.
Equality self-loops contribute a constant and inequality self-loops contribute
zero; removing them preserves approximation guarantees when the removed
constant is added back.

For unrestricted \maxagree, use the same signed graph $I$, with the same
edge weights, and allow any number of clusters.  By
Lemma~\ref{lem:hash-clusters},
$\operatorname{OPT}_k(I)\leq\operatorname{OPT}(I)\leq\operatorname{OPT}_k(I)/(1-1/k)$.
Thus the first case above gives $\operatorname{OPT}(I)\geq c-\eta$, whereas
the second gives $\operatorname{OPT}(I)\leq(s+\eta)/(1-1/k)$.
Letting $\eta\to0$, it is therefore UG-hard to approximate \maxagree\
within any factor strictly greater than $s/[c(1-1/k)]$.

\begin{proof}[Proof of the hardness statement in Theorem~\ref{thm:maxagree-sdp-gap-hardness}]
Fix $\alpha>\alpha_{\mathrm{SDP}}$.  By the
definition of the integrality gap, there is a finite weighted instance $I$
satisfying
\[
 \frac{\operatorname{OPT}(I)}{\operatorname{SDP}_{\mathrm{std}}(I)}<\alpha.
\]
Choose a sufficiently large constant $k$ so that
\[
 \frac{\operatorname{OPT}(I)}
      {(1-1/k)^2\operatorname{SDP}_{\mathrm{std}}(I)}<\alpha.
\]
Lemma~\ref{lem:basic-sdp-sandwich} and
$\operatorname{OPT}_k(I)\leq\operatorname{OPT}(I)$ imply
\[
 \frac{\operatorname{OPT}_k(I)}
      {(1-1/k)\operatorname{SDP}_k(I)}
 \leq
 \frac{\operatorname{OPT}(I)}
      {(1-1/k)^2\operatorname{SDP}_{\mathrm{std}}(I)}<\alpha.
\]
After normalizing the constraint weights, apply the gap-to-hardness
transfer above to this fixed instance and choose its additive error small
enough to preserve the strict inequality.  It is therefore UG-hard to achieve approximation factor
$\alpha$ for unrestricted \maxagree.  Taking
$\alpha=\alpha_{\mathrm{SDP}}+\varepsilon$ proves the claim.
\end{proof}

\subsection{A Gaussian Construction with Integrality Ratio at Most
\texorpdfstring{$0.802$}{0.802}}
\label{sec:gaussian-maxagree-gap}
\label{sssec:gaussian-gap-instance}
This subsection proves the unconditional integrality-gap bound in
Theorem~\ref{thm:maxagree-sdp-gap-hardness}.  We first construct a gap for
the relaxation on constraint pairs, then transfer it to
$\operatorname{SDP}_{\mathrm{std}}$ in
Theorem~\ref{thm:gaussian-maxagree-gap}.

Recall that $\gamma_d$ denotes the standard Gaussian probability measure on $\mathbb R^d$.
Write
$\operatorname{vol}_\gamma(A)=\gamma_d(A)$ for the Gaussian volume of a
measurable set $A\subseteq\mathbb R^d$.  We first define a 
continuous model with vertex set $\mathbb R^d$ and edge weights given by
measures on pairs of vertices.  We later obtain a finite weighted graph
by partitioning the vertex set into cells and aggregating these measures
over pairs of cells.

Recall that $\nu_\rho$ denotes the joint distribution of a pair of
$\rho$-correlated standard Gaussian vectors in $\mathbb R^d$.
In particular,
$\nu_0=\gamma_d\otimes\gamma_d$ (the product measure) describes independent
Gaussians.
Our instance has the following positive and negative edge-weight measures in $\mathbb R^{2d}$:
\begin{equation*}
 \mu^+=\frac47\nu_{3/4},\qquad
 \mu^-=\frac37\nu_0.
\end{equation*}
The weights and correlation are motivated by the parameter optimization
discussed below.
Equivalently, its objective is the expected payoff of the following
random constraint: with probability $4/7$, draw a $3/4$-correlated
Gaussian pair and ask for equal labels; with probability $3/7$, draw
independent Gaussian endpoints and ask for different labels.
This sampling rule describes the weight measures, rather than generating
a finite graph.

Both measures are symmetric, and their masses sum to one.  Since the
payoffs are symmetric in the endpoints, we may regard each sampled pair
as an undirected edge.

\paragraph{Integral solutions and noise stability of partitions.}
A $k$-clustering is a measurable labeling
$\ell:\mathbb R^d\to[k]$, giving rise to a partition
$\mathcal A=(A_1,\ldots,A_k)$ with $A_i=\ell^{-1}(i)$.

Write $p_i=\operatorname{vol}_\gamma(A_i)$, so $\sum_i p_i=1$.
The \emph{noise stability of the partition} is
the
probability that a  $\rho$-correlated perturbation of a random Gaussian vector
preserves its label in the partition:
\begin{equation*}
\begin{aligned}
 P_{\mathcal A}(\rho)
 =\Pr_{(X,Y)\sim\nu_\rho}
       [\ell(X)=\ell(Y)]
 =\sum_{i=1}^k
   \Pr_{(X,Y)\sim\nu_\rho}
       [X\in A_i,\ Y\in A_i].
\end{aligned}
\end{equation*}
Each summand is the Gaussian noise stability
$\operatorname{Stab}_\rho[\mathbf1_{A_i}]$ in the notation of
\citet[Chapter~11]{o2014analysis}.
At $\rho=0$, independence gives, for $p_i>0$,
$\Pr[Y\in A_i\mid X\in A_i]=p_i$, and hence
$P_{\mathcal A}(0)=\sum_i p_i^2$.

The agreement value is the probability that a sampled constraint is satisfied:
\begin{equation}
\label{eq:gaussian-gap-integral-objective}
\begin{aligned}
 A(\mathcal A)
 &=\int\mathbf1[\ell(x)=\ell(y)]\,d\mu^+(x,y)
   +\int\mathbf1[\ell(x)\ne\ell(y)]\,d\mu^-(x,y)\\
 &=\frac47P_{\mathcal A}(3/4)
   +\frac37\left(1-\sum_i p_i^2\right).
\end{aligned}
\end{equation}
The fractional objective below uses the same measures and normalization.
Positive edges favor partitions whose labels survive a correlated
perturbation.  Negative edges favor spreading the Gaussian mass over many
cells, so that independent endpoints rarely receive the same label.

\subsubsection{Borell's inequality and a partition-independent bound}
\label{sssec:gaussian-gap-borell}

Following \citet[Chapter~11]{o2014analysis}, for $\rho\in[-1,1]$ and
$p,q\in(0,1)$ define the \emph{Gaussian quadrant probability}, with
$\nu_\rho$ taken in dimension $d=1$,
\begin{equation*}
 \Lambda_\rho(p,q)=
 \Pr_{(U,V)\sim \nu_\rho }
       [U\le\Phi^{-1}(p),\ V\le\Phi^{-1}(q)],
 \qquad
 \Lambda_\rho(p)=\Lambda_\rho(p,p).
\end{equation*}
Extend these functions continuously to $p,q\in[0,1]$; in particular,
$\Lambda_\rho(0)=0$ and $\Lambda_\rho(1)=1$.
The diagonal value $\Lambda_\rho(p)$ is
$\operatorname{Stab}_\rho[\mathbf1_H]$ for any halfspace $H$ of Gaussian
volume $p$, such as $H=\{x:x_1\le\Phi^{-1}(p)\}$.

\begin{theorem}[Borell's Gaussian noise-stability inequality~\citep{Borell85}]
\label{thm:gaussian-gap-borell}
For every $d\ge1$, $0<\rho<1$, and measurable $A\subseteq\mathbb R^d$,
\[
 \Pr_{(X,Y)\sim\nu_\rho}
       [X\in A,Y\in A]
 =\operatorname{Stab}_\rho[\mathbf1_A]
 \le \Lambda_\rho(\operatorname{vol}_\gamma(A)).
\]
\end{theorem}

In words, among sets with a prescribed Gaussian measure, a halfspace is
most likely to contain both endpoints of a positively correlated Gaussian
pair.
This is a theorem about arbitrary measurable sets, regardless of dimension.
Among its many proofs, the reader may check the one in \citet[Theorem~1.1]{MosselNeeman15Noise}.

Applying it separately to all the cells gives
\begin{equation*}
 P_{\mathcal A}(\rho)
 =\sum_i\operatorname{Stab}_\rho[\mathbf1_{A_i}]
 \le\sum_i\Lambda_\rho(p_i).
\end{equation*}

\paragraph{Upper-bounding the total agreement.}
Suppose that, for a fixed $0<\rho<1$ and constants $\alpha,\beta\geq0$,
\begin{equation*}
 \Lambda_\rho(p)\leq\alpha p^2+\beta p
 \qquad(0\leq p\leq1).
\end{equation*}
Writing $S=\sum_i p_i^2$ and summing over the cells gives
$P_{\mathcal A}(\rho)\leq\alpha S+\beta$.
More generally, put positive weight $a$ on $\nu_\rho$ and negative weight
$b$ on $\nu_0$, where $a,b\geq0$ and $a+b=1$.  The agreement of any
partition then satisfies
\begin{equation*}
 A_{a,b,\rho}(\mathcal A)
 =aP_{\mathcal A}(\rho)+b(1-S)
 \leq b+a\beta+(a\alpha-b)S.
\end{equation*}
We may cancel the dependence on $S$ by
choosing $b=a\alpha$, that is,
         $
 a=\frac1{1+\alpha},\ b=\frac\alpha{1+\alpha},
 $
 which yields
\begin{equation}\label{eq:A_cancel}
 A_{a,b,\rho}(\mathcal A) \leq  \frac{\alpha+\beta}{1+\alpha}.
\end{equation}

The fractional solution constructed below has value tending to
$a\rho+b=(\rho+\alpha)/(1+\alpha)$, so the resulting inverse-gap bound is
$(\alpha+\beta)/(\rho+\alpha)$.  Only $\rho$ and $\alpha$ are free: the
weights are fixed by the cancellation, and the smallest admissible $\beta$
is $\beta^*(\rho,\alpha)=\sup_{0<p\leq1}(\Lambda_\rho(p)-\alpha p^2)/p$.
Numerical exploration gives about $0.80192$, near $\rho\approx0.748$ and
$\alpha\approx0.745$.  We take the nearby simpler choices
$\rho=\alpha=3/4$, which give $a=4/7$ and $b=3/7$; then
$\beta^*\approx0.452895$ numerically, attained at $p\approx0.195$.

We verify our choice by establishing the following quadratic upper bound on Gaussian
noise stability:

\begin{lemma}%
\label{lem:gaussian-gap-scalar}
For every $p\in[0,1]$,
\begin{equation*}
 \Lambda_{3/4}(p)\le\frac34p^2+\frac{453}{1000}\,p.
\end{equation*}
\end{lemma}
Hence, by~\eqref{eq:A_cancel}, we have for any dimension and any arbitrary partition $\mathcal A$:
\begin{equation}
\label{eq:gaussian-gap-soundness}
A(\mathcal A) \le \frac{1203}{1750} \approx 0.68743.
\end{equation}

We prove this quadratic bound in
Section~\ref{sssec:gaussian-gap-certificate} using the angular
representation of the Gaussian quadrant probability, Jensen's inequality,
and elementary estimates.

\subsubsection{The fractional solution and the resulting gap}
\label{sssec:gaussian-gap-fractional}

The partition bound above supplies the integral side of the gap.
For the fractional side, we specify a vector at every vertex
and restrict the pair measures to the region where the corresponding
label-vector inner products are nonnegative.  We analyze integrals
against these fixed measures before constructing a finite graph.
Discretization then preserves both bounds up to arbitrarily small errors.

Write $\operatorname{SDP}^{\mathrm{cp}}_{\mathrm{std}}(I)$ for the value of
$\operatorname{SDP}_{\mathrm{std}}$ when $\langle v_i,v_j\rangle\geq0$ is
required only for pairs $ij$ of positive weight.  The upper-bound argument in
Lemma~\ref{lem:basic-sdp-sandwich} also gives
$\operatorname{SDP}^{\mathrm{cp}}_k(I)\leq
\operatorname{SDP}^{\mathrm{cp}}_{\mathrm{std}}(I)$.

\begin{lemma}
\label{lem:gaussian-gap-finite}
Fix $k\ge2$.  For every $\varepsilon>0$, there is a finite weighted
\maxagree\ instance $J$, of total constraint weight one, such that
\[
 \operatorname{SDP}^{\mathrm{cp}}_{\mathrm{std}}(J)
 \geq\operatorname{SDP}^{\mathrm{cp}}_k(J)
 \ge D_k-\varepsilon,
 \qquad
 \operatorname{OPT}(J)\le s+\varepsilon,
\]
where
\begin{equation*}
 D_k=\frac67-\frac{2}{7k},
 \qquad s=\frac{1203}{1750}.
\end{equation*}
Here a pair may carry both an equality and an inequality constraint.
\end{lemma}

\begin{proof}
Fix $k$ and a sufficiently large dimension $d\ge2$.
Let $x$ denote a point of $\mathbb R^d$
and define, for every $x,y\ne0$,
\[
 r_x=\frac{x}{\|x\|},
 \qquad c_{xy}=\langle r_x,r_y\rangle,
 \qquad y_k(x,y)=\frac1k+\left(1-\frac1k\right)c_{xy}.
\]
For each $x\ne0$, define $u_{x,a}$ by~\eqref{eq:label-vectors-from-sdp-a}
with $v_i=r_x$.  These vectors are orthogonal at each vertex and sum to
$u_0$; their inner products are given by~\eqref{eq:uiajb}.
Since $y_k(x,y)\le1$, the latter are nonnegative exactly on the set
\begin{equation}
\label{eq:gaussian-gap-retention}
 \mathcal R_{k,d}
 :=\left\{(x,y)\in(\mathbb R^d\setminus\{0\})^2:
             c_{xy}\ge-\frac1{k-1}\right\}.
\end{equation}

\emph{Restricted measures. }%
Define the restricted pair measures by
\[
 \widetilde\mu^\pm(B)
 :=\mu^\pm(B\cap\mathcal R_{k,d})
 \qquad\text{for every measurable }B\subseteq\mathbb R^{2d}.
\]
In other words, we drop the CSP constraints on vertex pairs outside the region $\mathcal R_{k,d}$;
this defines new pair weights.
Note that removing constraints cannot increase any clustering's
unnormalized value.
For every $(x,y)\in\mathcal R_{k,d}$, the inner products form a
local distribution with uniform marginals and same-label probability
$y_k(x,y)$.  The vectors satisfy the vertex and pair constraints required
by $\operatorname{SDP}^{\mathrm{cp}}_k$.

\emph{Limiting values.}
The total mass of $\widetilde\mu^++\widetilde\mu^-$ and the 
fractional value are, respectively,
\[
 \begin{aligned}
 M_{k,d}&=\frac47\nu_{3/4}(\mathcal R_{k,d})
           +\frac37\nu_0(\mathcal R_{k,d}),\\
 F_{k,d}&=\int y_k(x,y)\,d\widetilde\mu^+(x,y)
       +\int(1-y_k(x,y))\,d\widetilde\mu^-(x,y).
 \end{aligned}
\]
For $(X,Y)\sim\nu_\rho$, the coordinate pairs $(X_\ell,Y_\ell)$ are
independent and identically distributed, so 
by the weak law of large numbers, $c_{X, Y}$ converges in probability to $\rho$ as $d \to \infty$.
For fixed $k$ and $\rho\ge 0$, it follows that
$\lim_{d\to\infty}\Pr_{(X,Y)\sim\nu_\rho}[(X,Y)\in\mathcal R_{k,d}] = 1$.
Therefore, 
$\lim_{d\to\infty}M_{k,d}=1$ as well.
On positive constraints, $c_{XY}\to3/4$, so the same-label payoff
$y_k(X,Y)$ tends to $1/k+(1-1/k)(3/4)=3/4+1/(4k)$.
On negative constraints, $c_{XY}\to0$, so the same-label probability tends
to $1/k$ and the different-label payoff $1-y_k(X,Y)$ tends to $1-1/k$.
Both payoffs lie in $[0,1]$ on $\mathcal R_{k,d}$.
Thus,

\begin{equation}
\label{eq:gaussian-gap-limiting-value}
 \lim_{d\to\infty}F_{k,d}
 =\frac47\left(\frac34+\frac1{4k}\right)
  +\frac37\left(1-\frac1k\right)
 =D_k.
\end{equation}
Since $\widetilde\mu^\pm\le\mu^\pm$ and agreement payoffs are nonnegative,
the agreement integral of any partition against $\widetilde\mu^\pm$ is
at most its value $A(\mathcal A)$ against the original measures, hence at
most $s$ by~\eqref{eq:gaussian-gap-soundness}.
Its ratio to the fractional value is therefore at most
$s/F_{k,d}\to s/D_k$.

These bounds give the gap before discretization; it remains to realize
them, up to small errors, in a finite instance.

\emph{Discretization.}
For fixed $d$ and any $\eta>0$, choose $R>0$ such that the closed ball
$B(0,R)=\{x\in\mathbb R^d:\|x\|\le R\}$
satisfies $\gamma_d(B(0,R))\ge1-\eta$.
Fix a mesh size $h>0$.  Intersect $B(0,R)$ with the half-open cubes
of a grid of side $h/\sqrt d$, and enumerate the nonempty, disjoint cells as
$C_1,\ldots,C_N$.
Each cell has diameter at most $h$;
choose a representative $x_i\in C_i\setminus\{0\}$ for each.
The graph has vertices $1,\ldots,N$.  For every $i<j$, define
\[
 W^\pm_{ij}
 =\mathbf1[(x_i,x_j)\in\mathcal R_{k,d}]\,
   \mu^\pm\bigl((C_i\times C_j)\cup(C_j\times C_i)\bigr).
\]
The masses here use the original measures $\mu^\pm$.  The indicator sets
both weights to zero unless the representatives satisfy
$c_{x_ix_j}\ge-1/(k-1)$.  We omit loops and pairs of total weight zero.  Thus there are at most $\binom N2$ pairs, each
possibly carrying both signs. Attach the vectors $u_{x_i,a}$ from~\eqref{eq:label-vectors-from-sdp-a}
to vertex $i$.  Every pair of positive weight belongs to
$\mathcal R_{k,d}$ at its representatives, so these vectors and their
local distributions are feasible for $\operatorname{SDP}^{\mathrm{cp}}_k$.

Every clustering of the finite vertices extends to a measurable partition
of $\mathbb R^d$, constant on each $C_i$ and arbitrary outside $B(0,R)$.
Its unnormalized finite objective is exactly the contribution of the
retained cell pairs to the full objective.  All other
contributions are nonnegative, so~\eqref{eq:gaussian-gap-soundness} bounds
the finite objective by $s$.  

\emph{Value after discretization.}
It remains to show that discretization loses little total weight and
fractional value.  The cells $C_i$ and their representatives $x_i$
depend on $h$.  For $(x,y)\in C_i\times C_j$, set
\[
 g_h(x,y)=\mathbf1[i\ne j]\,\mathbf1[(x_i,x_j)\in\mathcal R_{k,d}],
 \qquad
 y_h(x,y)=y_k(x_i,x_j).
\]
Since $W^\pm_{ij}$ collects the mass of both orderings of each cell pair,
the total weight $M_h$ and fractional value $F_h$ of the finite instance
are exactly
\[
 M_h=\int_{B(0,R)^2}g_h\,d(\mu^++\mu^-),
 \qquad
 F_h=\int_{B(0,R)^2}g_h\,y_h\,d\mu^+
   +\int_{B(0,R)^2}g_h\,(1-y_h)\,d\mu^-.
\]
In particular, $M_h\le1$.  Let $h\to0$ with $d$ and $R$ fixed.
Since cell diameters tend to zero, the representative of the cell containing
$x$ converges to $x$.  By continuity, $y_h(x,y)\to y_k(x,y)$ for $x,y\ne0$.
Moreover, $g_h(x,y)\to\mathbf1[(x,y)\in\mathcal R_{k,d}]$ whenever also
$x\ne y$ and $c_{xy}\ne-1/(k-1)$: eventually the points lie in different
cells and the representatives lie on the same side of the threshold as
$(x,y)$.  For $d\ge2$, the zero endpoints, diagonal, and threshold set
form a null set under both pair laws, which have densities on
$\mathbb R^{2d}$.  All integrands lie in $[0,1]$ and the measures are finite,
so by the dominated
convergence theorem, $M_h$ and $F_h$ tend to the integrals defining $M_{k,d}$ and
$F_{k,d}$, after restricting them to $B(0,R)^2$.
Each pair law has marginals $\gamma_d$ and therefore gives mass at most
$2\eta$ to the complement of $B(0,R)^2$; hence the restriction lowers
$M_{k,d}$ and $F_{k,d}$ by at most $2\eta$.

\emph{Conclusion.}
For every $0<\delta<1$, choosing $d$ sufficiently large, then $\eta$
sufficiently small, and finally $h$ sufficiently small gives
$1-\delta\le M_h\le1$ and $F_h\ge D_k-\delta$.  Dividing all weights by
$M_h$ gives an instance $J$ of total weight one.  The same vectors remain
feasible, and the bound on every clustering scales by $1/M_h$, so
\[
 \operatorname{SDP}^{\mathrm{cp}}_k(J)\ge\frac{F_h}{M_h}\ge D_k-\delta,
 \qquad
 \operatorname{OPT}(J)\le\frac{s}{M_h}\le\frac{s}{1-\delta}.
\]
Taking $\delta\le\varepsilon$ with $s/(1-\delta)\le s+\varepsilon$ gives
the claimed bounds, and
$\operatorname{SDP}^{\mathrm{cp}}_{\mathrm{std}}(J)\ge
\operatorname{SDP}^{\mathrm{cp}}_k(J)$ by the comparison preceding the
lemma.
\end{proof}

Letting $k\to\infty$, Lemma~\ref{lem:gaussian-gap-finite} shows that the
inverse gap of $\operatorname{SDP}^{\mathrm{cp}}_{\mathrm{std}}$ is at most
$s/(6/7)=401/500$, with explicit instances.  For
$\operatorname{SDP}_{\mathrm{std}}$, these vectors do not suffice: completing
$J$ with zero-weight edges brings back the nonnegativity constraints on the
deleted pairs, which are violated.  This is where the strong-gap theorem
enters.

\begin{theorem}[Integrality gap]
\label{thm:gaussian-maxagree-gap}
For every fixed $k\ge2$, requiring local distributions on all pairs does not
change the integrality gap,
$\operatorname{gap}(\operatorname{SDP}_k)
 =\operatorname{gap}(\operatorname{SDP}^{\mathrm{cp}}_k)$, and
\begin{equation}
\label{eq:gaussian-gap-finite-k-ratio}
 \frac1{\operatorname{gap}(\operatorname{SDP}_k)}
 \le \frac{s}{D_k}
 =\frac{401/500}{1-1/(3k)}.
\end{equation}
Moreover, unconditionally,
\begin{equation*}
 \operatorname{gap}(\operatorname{SDP}_{\mathrm{std}})
 \ge\frac{500}{401}.
\end{equation*}
\end{theorem}

\begin{proof}
This is the one place where we use the strong-gap theorem
of~\citet[Theorem~1.3]{RaghavendraSteurer09Gaps}.  Given an instance
of total constraint weight one on which $\operatorname{SDP}^{\mathrm{cp}}_k\geq c$
and $\operatorname{OPT}_k\leq s'$, it gives, for every $\eta>0$, a new
instance $I$ with $\operatorname{OPT}_k(I)\leq s'+\eta$, together with label
vectors of value at least $c-\eta$ and a local label distribution on every
pair of nodes, whose second moments agree with the inner products of the
label vectors.  (The theorem provides consistent local distributions on all
sets of at most $R$ nodes, for any fixed $R$; we only need $R=2$.)  The new instance is obtained by running Raghavendra's reduction on
an integrality-gap instance of Unique Games
(Section~\ref{subsec:raghavendra-overview}).

These label vectors and local distributions are feasible for
$\operatorname{SDP}_k$ on $I$.  Hence
$\operatorname{gap}(\operatorname{SDP}_k)\geq
\operatorname{gap}(\operatorname{SDP}^{\mathrm{cp}}_k)$, and the reverse
inequality holds because $\operatorname{SDP}_k\leq\operatorname{SDP}^{\mathrm{cp}}_k$
on every instance (Section~\ref{subsec:ugc-gap-transfer}).  Applying this to
the instances of Lemma~\ref{lem:gaussian-gap-finite}, with
$c=D_k-\varepsilon$ and $s'=s+\varepsilon$, and taking
$\varepsilon,\eta\to0$ yields~\eqref{eq:gaussian-gap-finite-k-ratio}.

For the same instances $I$, Lemma~\ref{lem:basic-sdp-sandwich} gives
$\operatorname{SDP}_{\mathrm{std}}(I)\geq\operatorname{SDP}_k(I)
\geq D_k-\varepsilon-\eta$, and Lemma~\ref{lem:hash-clusters} gives
$\operatorname{OPT}(I)\leq\operatorname{OPT}_k(I)/(1-1/k)
\leq(s+\varepsilon+\eta)/(1-1/k)$.  Hence
\[
 \operatorname{gap}(\operatorname{SDP}_{\mathrm{std}})
 \geq\left(1-\frac1k\right)\frac{D_k-\varepsilon-\eta}{s+\varepsilon+\eta},
\]
and letting $\varepsilon,\eta\to0$ and then $k\to\infty$ gives
$(6/7)/(1203/1750)=500/401$.

There is no difficulty in restricting to signed graphs with at most one
sign on each pair.  If a pair has weights $u,v\ge0$ of opposite signs,
subtract $\min(u,v)$ from both.  This removes the same constant contribution
from every integral solution and every feasible SDP solution.  If the
removed total constant is $C$, and the old optimum and SDP value
are $A<D$, then
\[
 \frac{A-C}{D-C}\le\frac{A}{D}.
\]
Here $C\le A$, so the denominator stays positive.  Thus this simplification
can only strengthen the gap.  The same observation applies after the
strong-gap transfer.
\end{proof}

Since $\alpha_{\mathrm{SDP}}=1/\operatorname{gap}(\operatorname{SDP}_{\mathrm{std}})$,
Theorem~\ref{thm:gaussian-maxagree-gap} proves the bound
$\alpha_{\mathrm{SDP}}\leq401/500=0.802$ in
Theorem~\ref{thm:maxagree-sdp-gap-hardness}.  Together with the hardness
statement proved in Section~\ref{subsec:ugc-gap-transfer}, this also
gives its $0.802+\varepsilon$ inapproximability conclusion.

\subsubsection{A quadratic bound on noise stability}
\label{sssec:gaussian-gap-certificate}

We prove Lemma~\ref{lem:gaussian-gap-scalar}.  Put $c=453/250$ and
$H(p)=4\Lambda_{3/4}(p)-3p^2-cp$; we must show that $H\leq0$ on $[0,1]$.
Let $U, G \sim \mathcal{N}(0, 1)$ be independent Gaussians and write
$V=\frac34U+\frac{\sqrt7}{4}G$. Then \[
 \Pr[V\le z\mid U=u]=\Phi\left(\frac{4z-3u}{\sqrt7}\right),
 \qquad
 \Pr[V\le z\mid U=z]=\Phi(z/\sqrt7).
\]
Put $F(a,b)=\Pr[U\le a,V\le b]$, so
$\Lambda_{3/4}(\Phi(z))=F(z,z)$.
By the chain rule and symmetry,
\[
\begin{aligned}
 \frac{d}{dz}\Lambda_{3/4}(\Phi(z))
 &=\partial_aF(z,z)+\partial_bF(z,z)\\
 &=\varphi(z)\Pr[V\le z\mid U=z]
   +\varphi(z)\Pr[U\le z\mid V=z]\\
 &=2\varphi(z)\Phi(z/\sqrt7).
\end{aligned}
\]
As $\Phi'(z)=\varphi(z)$,
the change of variables $p = \Phi(z)$ yields
$\Lambda'_{3/4}(p)=2\Phi(z/\sqrt7)$, and therefore
\[
 H'(p)=8\Phi(z/\sqrt7)-6p-c,
 \qquad
 H''(p)=\frac8{\sqrt7}e^{3z^2/7}-6.
\]
Since $H''$ is increasing in $|z|$, negative at $z=0$, and unbounded, $H$ is
convex on $[0,a]$, concave on $[a,1-a]$ and convex on $[1-a,1]$ for some
$a\in(0,\frac12)$.  Let $T$ be the tangent to $H$ at $p_0=\Phi(-r)$, where
$r=56/65$.  Since $r<1$ and $e^{3/7}<(1-\frac37)^{-1}=\frac74$, we have
$H''(p_0)<2\sqrt7-6<0$, so $p_0\in(a,1-a)$.  It suffices to show that
$T(0)<0$ and $T(1)<0$.  Then $T<0$ on $[0,1]$, so $H\leq T<0$ on $[a,1-a]$,
and on the two convex pieces $H$ lies below its chords, whose endpoint
values $H(0)=0$, $H(a)$, $H(1-a)$ and $H(1)=1-c$ are nonpositive.

Let $q_0=\Phi(-r/\sqrt7)$ and $D=p_0-\Lambda_{3/4}(p_0)=\Pr[U\leq-r<V]$.  Then
\[
 T(0)=4p_0-4D+3p_0^2-8p_0q_0,
 \qquad
 T(1)=T(0)+8q_0-6p_0-c.
\]
Instead of computing the bivariate probability $D$, we bound it from below.
With $\theta=\arccos\frac34$,
\[
 D=\frac1{2\pi}\int_0^\theta\exp\Bigl(-\frac{r^2}{1+\cos s}\Bigr)\,ds,
\]
since, as functions of $r$, both sides have derivative
$-2\varphi(r)\bigl(\Phi(r/\sqrt7)-\frac12\bigr)$ (for the right side,
differentiate under the integral and substitute $u=\tan(s/2)$) and vanish as
$r\to\infty$.
The average of $1/(1+\cos s)$ over $[0,\theta]$ is
$\tan(\theta/2)/\theta=1/(\theta\sqrt7)$, and $x\mapsto e^{-r^2x}$ is convex,
so Jensen's inequality gives
\[
 D\geq\frac\theta{2\pi}\exp\Bigl(-\frac{r^2}{\theta\sqrt7}\Bigr).
\]
The claim is thus reduced to two values of $\Phi$ and one elementary
number: $p_0\approx0.194471$, $q_0\approx0.372352$, and the right-hand side
above is approximately $0.078022$.  Substituting it for $D$ gives
$T(0)<-4\cdot10^{-5}$ and $T(1)<-5\cdot10^{-5}$.  The supplied verifier,
\texttt{certify\_\allowbreak{}gaussian\_\allowbreak{}maxagree\_\allowbreak{}bound.py},
certifies both inequalities in exact rational interval arithmetic, using
alternating Taylor series for $\Phi$, $\exp$ and $\arctan$, together with
$\theta=2\arctan(1/\sqrt7)$ and Machin's formula for $\pi$.

\section{Limitations of Rounding Families under the Edge-Wise Analysis}\label{sec:UB}

In this section, we prove the limitation results announced earlier for \maxagree\ and \maxagreek.
To this end, we develop a general framework for certifying an upper bound on the best approximation ratio that can be established by any mixture over a prescribed finite family of rounding schemes under the edge-wise analysis formalized in Section~\ref{subsec:approximations-general}.
Applying this framework to broad finite families of rounding schemes, we show that the approximation guarantees of our algorithms for \maxagree\ and \maxagreek\ are nearly optimal within the corresponding families and the edge-wise analysis framework.

Our approach is based on LP duality together with certified bounds on collision probabilities.
For the $k$-spokes roundings, the required bounds are provided by our Hermite-coefficient certification framework.
We first formulate the best approximation ratio certifiable by mixtures over a prescribed rounding family as a semi-infinite LP.
We then restrict the positive-edge and negative-edge inequalities to finitely many rational test points and relax the resulting constraints using certified lower and upper bounds on the collision probabilities.
Finally, we certify an upper bound on this finite relaxed LP by constructing a rational feasible solution to its dual.
Thus, although a floating-point LP solver is used to find a candidate dual solution, the final upper bound is verified entirely in exact rational arithmetic.

\subsection{An Upper-Bound Certification Framework}\label{subsec:UB_approach}

Fix a special case of the general clustering problem introduced in Section~\ref{subsec:approximations-general} 
and take a rational $a\in[-1/2,0]$ such that $\mathrm{SDP}(a)$ is a relaxation of the problem.
Let $\mathcal{R}$ be a finite family of rounding schemes such that each rounding scheme always returns a feasible solution when applied to an optimal solution to $\mathrm{SDP}(a)$.
For each $r\in\mathcal{R}$, let $P_r(\rho)$ denote the probability that two nodes whose corresponding SDP vectors have inner product $\rho$ are assigned to the same cluster by rounding scheme $r$.
Consider a mixture over $\mathcal{R}$ with probabilities $\bm{\lambda}=(\lambda_r)_{r\in \mathcal{R}}$, where $\lambda_r\geq 0$ for every $r\in \mathcal{R}$ and $\sum_{r\in\mathcal{R}}\lambda_r=1$. 
Letting $P(\rho)$ denote the collision probability of the mixture, we have 
\begin{equation*}
P(\rho)=\sum_{r\in\mathcal{R}}\lambda_r P_r(\rho).
\end{equation*}
Recall that 
\begin{equation*}
s_a^+(\rho) = \frac{\rho-a}{1-a}\quad \text{and}\quad s_a^-(\rho) = \frac{1-\rho}{1-a}\quad \text{for }\rho\in [a,1].
\end{equation*}

To certify an approximation ratio $\gamma$ using the edge-wise analysis, we need to show that
\begin{equation*}
P(\rho)\geq \gamma s_a^+(\rho)
\quad\text{and}\quad
1-P(\rho)\geq \gamma s_a^-(\rho)
\quad
\text{for all }\rho\in[a,1].
\end{equation*}
Here, unlike Lemma~\ref{lem:edge-wise-analysis}, the negative-edge inequality should be satisfied for all $\rho\in [a,1]$, 
since our focus is not only on $k$-spokes roundings. 
The best approximation ratio that can be certified for mixtures over $\mathcal{R}$ by the edge-wise analysis is represented as the optimal value of the following semi-infinite LP:
\begin{equation*}
\begin{aligned}
\text{maximize}\quad
& \gamma\\
\text{subject to}\quad
& P(\rho)\geq \gamma s_a^+(\rho)
&& \text{for all }\rho\in[a,1],\\
& 1-P(\rho)\geq \gamma s_a^-(\rho)
&& \text{for all }\rho\in[a,1],\\
& \sum_{r\in\mathcal{R}}\lambda_r=1,\\
& \lambda_r\geq 0
&& \text{for all }r\in\mathcal{R}.
\end{aligned}
\end{equation*}

Take
$\rho_1^+,\dots,\rho_{N_+}^+\in[a,1]$
and
$\rho_1^-,\dots,\rho_{N_-}^-\in[a,1]$.
Replacing the first two constraints in the above LP by
\begin{equation*}
\begin{aligned}
P(\rho_i^+)
&\geq
\gamma s_a^+(\rho_i^+)
&& \text{for all }i\in[N_+],\\
1-P(\rho_i^-)
&\geq
\gamma s_a^-(\rho_i^-)
&& \text{for all }i\in[N_-]
\end{aligned}
\end{equation*}
leads to a finite relaxed LP, whose optimal value is therefore an upper bound on the optimal value of the original semi-infinite LP.

For every $r\in\mathcal{R}$ and $i\in[N_+]$, suppose that we have a certified upper bound
$\overline{P}_r(\rho_i^+)$ on $P_r(\rho_i^+)$.
Similarly, for every $r\in\mathcal{R}$ and $i\in[N_-]$, suppose that we have a certified lower bound
$\underline{P}_r(\rho_i^-)$ on $P_r(\rho_i^-)$.
Then
\begin{equation*}
\overline{P}(\rho_i^+)
=
\sum_{r\in\mathcal{R}}
\lambda_r\overline{P}_r(\rho_i^+)
\end{equation*}
is a certified upper bound on $P(\rho_i^+)$, and 
\begin{equation*}
\underline{P}(\rho_i^-)
=
\sum_{r\in\mathcal{R}}
\lambda_r\underline{P}_r(\rho_i^-)
\end{equation*}
is a certified lower bound on $P(\rho_i^-)$.
Further replacing the finite constraints by
\begin{equation*}
\begin{aligned}
\overline{P}(\rho_i^+)
&\geq
\gamma s_a^+(\rho_i^+)
&& \text{for all }i\in[N_+],\\
1-\underline{P}(\rho_i^-)
&\geq
\gamma s_a^-(\rho_i^-)
&& \text{for all }i\in[N_-]
\end{aligned}
\end{equation*}
leads to a further relaxed LP.
Its optimal value is therefore still an upper bound on the approximation ratio that can be certified by the edge-wise analysis.

However, the numerical optimum returned by a standard LP solver does not by itself yield a rigorous upper bound, since it is subject to floating-point roundoff errors.
We instead certify an upper bound using the dual LP:
\begin{alignat*}{4}
&\text{minimize} &\ \ & U\\
&\text{subject to} &   & \sum_{i\in[N_+]} u_i s_a^+(\rho_i^+) + \sum_{i\in[N_-]} v_i s_a^-(\rho_i^-) =1,\\
&                  &   & U \geq \sum_{i\in[N_+]} \overline{P}_r(\rho_i^+)u_i + \sum_{i\in[N_-]} \left(1-\underline{P}_r(\rho_i^-)\right)v_i &\quad & \text{for all }r\in\mathcal{R},\\
&                  &   & u_i\geq 0 & & \text{for all }i\in[N_+],\\
&                  &   & v_i\geq 0 & & \text{for all }i\in[N_-].
\end{alignat*}
We first solve this dual LP using a standard LP solver and obtain a floating-point approximate optimal solution
$\left(
(u_i^*)_{i\in[N_+]},
(v_i^*)_{i\in[N_-]},
U^*
\right)$.
We then round the approximate values $u_i^*$ and $v_i^*$ to nonnegative rational numbers $\widetilde{u}_i$ and $\widetilde{v}_i$, respectively.
Let
\begin{equation*}
Z
=
\sum_{i\in[N_+]}
\widetilde{u}_i s_a^+(\rho_i^+)
+
\sum_{i\in[N_-]}
\widetilde{v}_i s_a^-(\rho_i^-),
\end{equation*}
which is expected to be nearly equal to $1$, so assume that $Z>0$.
Define
\begin{equation}\label{eq:dual-sol-normalization}
u_i^{\mathrm{cert}}
=
\frac{\widetilde{u}_i}{Z}
\ \ \text{for } i\in[N_+]
\quad\text{and}\quad
v_i^{\mathrm{cert}}
=
\frac{\widetilde{v}_i}{Z}
\ \ \text{for } i\in[N_-].
\end{equation}
Then
\begin{equation*}
\sum_{i\in[N_+]}
u_i^{\mathrm{cert}}s_a^+(\rho_i^+)
+
\sum_{i\in[N_-]}
v_i^{\mathrm{cert}}s_a^-(\rho_i^-)
=
1.
\end{equation*}
For each $r\in\mathcal{R}$, compute
\begin{equation*}
U_r
=
\sum_{i\in[N_+]}
\overline{P}_r(\rho_i^+)
u_i^{\mathrm{cert}}
+
\sum_{i\in[N_-]}
\left(
1-\underline{P}_r(\rho_i^-)
\right)
v_i^{\mathrm{cert}}.
\end{equation*}
Then
\begin{equation*}
U^{\mathrm{cert}}
=
\max_{r\in\mathcal{R}} U_r
\end{equation*}
is the objective value of a rational feasible solution to the dual LP.
By weak duality, $U^{\mathrm{cert}}$ is a certified upper bound on the optimal value of the relaxed primal LP, 
and therefore also on the best approximation ratio that can be certified for mixtures over $\mathcal{R}$ by the edge-wise analysis.

Note that when $\mathcal{R}$ consists solely of $k$-spokes roundings, the negative-edge part of the analysis can be simplified further.
Indeed, by Corollary~\ref{cor:negative-ratio-monotonicity}, the negative-edge inequality over $[a,1]$ is equivalent to the single endpoint condition $1-P(a)\geq\gamma$.
Thus, for a family consisting solely of $k$-spokes roundings, it suffices to use a single test point for the negative-edge constraint, i.e., $N_-=1$ and $\rho_1^-=a$.

\subsection{Implementation and Certificates}
\label{subsec:upper-bound-implementation}

We implement the upper-bound certification approach using rational grids.
For a grid denominator $N\in\mathbb{N}$ and a rational $a\in[-1/2,0]$, we use the test points
\begin{equation*}
\rho_i^+
=
a+\frac{i}{N}(1-a)
\quad\text{and}\quad
\rho_i^-
=
a+\frac{i-1}{N}(1-a)
\quad \text{for } i=1,\ldots,N.
\end{equation*}
Thus, the positive-edge constraints are tested on
\begin{equation*}
\left\{
a+\frac{1-a}{N},\,
a+\frac{2(1-a)}{N},
\dots,
1
\right\},
\end{equation*}
while the negative-edge constraints are tested on
\begin{equation*}
\left\{
a,\,
a+\frac{1-a}{N},
\dots,
1-\frac{1-a}{N}
\right\}.
\end{equation*}
The omitted endpoints correspond to trivial constraints, since $s_a^+(a)=0$ and $s_a^-(1)=0$.

The dual LP is formed using certified lower and upper bounds on the collision probability of each rounding scheme at these test points.
For the $k$-spokes rounding, whose collision probability is $P_k(\rho)$, we can use the certified lower and upper bounds derived in Proposition~\ref{prop:theory-Pk-LB-UB}. 
For the $k$-hyperplane rounding, the collision probability is
\begin{equation*}
P^\mathrm{GW}_k(\rho)
=
\left(
1-\frac{\arccos(\rho)}{\pi}
\right)^k,
\end{equation*}
so the required lower and upper bounds can be obtained directly by outward-rounded ball arithmetic.
Thus, the distinction between the two types of roundings enters only in the computation of certified bounds on their collision probabilities; the LP framework itself treats all rounding schemes uniformly.

We distinguish between certificate generation and certificate checking.
A standard floating-point LP solver is used only to find a candidate solution to the dual LP.
The resulting dual variables are then rounded to nonnegative rational numbers and normalized exactly as in \eqref{eq:dual-sol-normalization}. 
We then compute the right-hand side of each dual constraint in exact rational arithmetic and set $U^\mathrm{cert}$ to their maximum, 
thereby obtaining an exactly feasible rational dual solution. 
The certificate consists of the parameter $a$, the test points, the prescribed rounding family $\mathcal{R}$, certified lower and upper bounds on the corresponding collision probabilities, and the rationalized dual variables.
The checker verifies the normalization constraint and all dual inequalities in exact rational arithmetic.
Hence, the floating-point LP solver is not part of the verification.

\subsection{Applications to \maxagree\ and \maxagreek}\label{subsec:UB-applications}
We first apply the framework to \maxagree, where $a=0$. 
Here we set $M=24$ and $N=100{,}000$. 
For the finite family consisting of the $k$-spokes roundings for $1\leq k\leq 32$ and the $k$-hyperplane roundings for $1\leq k\leq 5$, 
the exact checker verifies $U^{\mathrm{cert}}<0.78187$. 
For Swamy's subfamily, consisting of the $6$-spokes rounding and the $2$-hyperplane rounding, it verifies $U^{\mathrm{cert}}<0.77222$.
This gives the following result:
\begin{theorem}\label{thm:UB}
For \maxagree, consider the finite family consisting of the $k$-spokes roundings for $1\leq k\leq 32$ and the $k$-hyperplane roundings for $1\leq k\leq 5$.
Within the edge-wise analysis framework, no mixture over this family can certify an approximation ratio larger than $0.78187$.
Moreover, for Swamy's subfamily, consisting of the $6$-spokes rounding and the $2$-hyperplane rounding, no mixture can certify an approximation ratio larger than $0.77222$.
\end{theorem}

The first bound shows that our mixture of the $4$-spokes and $5$-spokes roundings, which achieves an approximation ratio of $0.7818$, is nearly optimal within this family and analysis framework.
The second bound shows that, within the same framework, attaining our approximation ratio of $0.7818$, or any comparable improvement over Swamy's guarantee, 
requires moving beyond Swamy's subfamily.

We next apply the framework to \maxagreek, where $a=a_K=-1/(K-1)$.
We consider the family consisting of the $k$-spokes roundings for $1\leq k\leq K$, which is precisely the family used by our proposed algorithm.
Since this family consists solely of $k$-spokes roundings, the negative-edge constraint is put only at a single point. 
The exact checker verifies the following result: 
\begin{theorem}\label{thm:UB-maxagree-k}
For each $K=3,\dots,16$, within the edge-wise analysis framework, 
no mixture over the $k$-spokes roundings
for $1\leq k\leq K$ can certify an approximation ratio larger than $U^\mathrm{cert}_K$ 
presented in Table~\ref{tab:maxagree-k}.
\end{theorem}

\section{Related Work}\label{sec:related}

\paragraph{Hermite expansions and Gaussian rounding.}
Hermite expansions have long been used to analyze Gaussian rounding.
In particular, \citet[Lemma~5]{Frieze+97} already established that the
collision probability $P_k$ has a power-series expansion with nonnegative
coefficients summing to one; their proof reduces the first few coefficients
to one-dimensional Gaussian integrals, and they observe, without an explicit
formula, that the same could be done at every degree.
Lemma~\ref{lem:one-dimensional-hermite-coefficients} carries out this
reduction in closed form for all degrees.
Our framework develops this representation into a
systematic method for certifying higher-degree coefficients and bounding
the collision function.
The Hermite representation also yields the monotonicity property behind
our proof
of the conjecture of~\citet*{deKlerk+04}.

The closest work to ours methodologically is
the concurrent independent work
of~\citet{Heilman26Grothendieck}.
Building on the higher-dimensional
Krivine roundings of \citet*{Braverman+13},
the author improves the upper bound on Grothendieck's
constant by thresholding Hermite polynomials on a random plane.
His one-dimensional formula for the Hermite coefficients of planar
threshold functions~\citep[Section~16.1, Eqs.~(116)--(117)]{Heilman26Grothendieck} is analogous
to~\eqref{eq:f1-hat-one-dimensional} for our argmax cells.
Both approaches combine validated integration
on a bounded interval with analytic tail bounds.

This shared reduction serves different objectives.  Heilman designs
rounding maps for Grothendieck's inequality and controls the inverse of
their correlation function to certify the resulting bound.  We analyze the
family of $k$-way argmax partitions underlying Frieze--Jerrum rounding,
across different values of $k$ and the full range of correlations.  The
Hermite coefficients of these partitions are indexed by multi-indices in
dimension $k$, whose number grows combinatorially with $k$ and the degree:
for $k\le32$ and degrees up to $24$ there are about $10^{16}$ of them.  We
exploit coordinate symmetries to reduce this to about $10^6$ integrals and
to organize the computations into reusable rational bounds on $P_k$.  These bounds support the design of rounding mixtures,
certify their guarantees for several clustering objectives, and bound
what prescribed rounding families can achieve under the edge-wise
analysis.

\paragraph{Computer-assisted analysis and rounding mixtures.}
\citet{Zwick02} used interval arithmetic to certify analytic inequalities
arising in SDP rounding.  More recently,
\citet*{Brakensiek+26_sicomp} improved approximations for Max 2-AND and
Max Di-Cut by computationally discovering distributions of rounding
functions and certifying their guarantees.  The work
of~\citet*{Brakensiek+26_stoc} similarly improves the long-standing
$1.2965$-approximation for Multiway Cut~\cite{Sharma+14} to $1.2787$ using
mixtures of hundreds of rounding schemes.  Their search
is formulated as a game between a rounding scheme and an edge
configuration, and their analysis combines analytic arguments with
interval arithmetic.  Like ours, their algorithm improves a
long-standing guarantee through optimized mixtures of roundings, of the
simplex LP in their case and of an SDP in ours.

\paragraph{Max $K$-Cut.}
For $K=2$, the approximation ratio
$\alpha_{\mathrm{GW}}\approx0.878567$ of \citet{Goemans+95} is optimal
under the UGC~\cite{Khot+07}.  The hardness proof uses the Majority Is
Stablest theorem of \citet*{Mossel+10}.
For large $K$, \citet*{Khot+07} also establish UGC-hardness
asymptotically matching the Frieze--Jerrum guarantee~\cite{Frieze+97}; whereas
unconditionally, Max $K$-Cut is NP-hard to
approximate within $1-1/(34K)$~\cite{Kann+97}.
\citet*{deKlerk+04} give an alternative geometric analysis of
Frieze--Jerrum rounding, report improved numerical estimates for
$4\le K\le10$, and formulate the worst-case conjecture resolved here.
\citet{Newman18} gives an alternative rounding scheme with an explicit
collision probability for general $K$; our approach sharpens the analysis
of the original Frieze--Jerrum scheme.  \citet*{Shinde+21} pursue a
complementary goal: obtaining Max $K$-Cut and \maxagree\ algorithms using
$O(|V|+|E|)$ memory while retaining guarantees close to those
of the underlying SDP roundings.

\paragraph{Gaussian partitions and inapproximability.}
Isoperimetric methods underlie the Max-Cut integrality gap
of~\citet{FeigeSchechtman02}.  Our \maxagree\ construction applies Borell's
inequality~\cite{Borell85} to each cell separately, avoiding the need to
identify an optimal partition.  For \maxagree, the connection to the
standard SDP and UGC hardness uses the CSP framework and strong-gap results
of~\citet{Raghavendra08} and~\citet{RaghavendraSteurer09Gaps}, as explained
in Section~\ref{subsec:raghavendra-overview}.

\citet{IsakssonMossel12} connect Gaussian partition conjectures to
optimal approximation of Max $K$-Cut under the UGC.  Our endpoint
characterization shows that their standard simplex conjecture
at the single correlation $\rho=-1/(K-1)$
\citep[Conjecture~1.4(ii)]{IsakssonMossel12} would suffice, via their
reduction, to establish
matching UGC-hardness for the Frieze--Jerrum approximation ratio.
For $K=3$, where the endpoint characterization was already known,
\citet{Heilman26} recently proved the required stability inequality
(namely the three-candidate Plurality is Stablest conjecture) for
correlations in $[-1/2,2/5]$, and obtained matching UGC-hardness.

\paragraph{Correlation Clustering.}
The approximation and hardness results most directly relevant to
\maxagree\ and \maxagreek\ are reviewed in the introduction.  On unweighted
complete graphs, \maxagree\ admits a PTAS~\cite{Bansal+04}.
With a fixed number $K$ of clusters on unweighted complete graphs, both
\maxagreek\ and $\textsc{MinDisagree}[K]$ admit
PTASs~\cite{Giotis+06_j}. For the complementary \mindisagree\ objective, the general weighted case
admits an $O(\log n)$-approximation~\cite{Charikar+05,Demaine+06}.  Since
this case is equivalent to Multicut up to constant
factors~\cite{Demaine+06}, it is NP-hard to approximate within every
constant factor under the UGC~\cite{Chawla+06}.
For \mindisagree\ on unweighted complete graphs, progress from pivot algorithms and LP
rounding~\cite{Ailon+08,Chawla+15} to stronger
relaxations~\cite{Cohen-Addad+22,Cohen-Addad+23,Cao+24} has led to substantially
better guarantees.  Recent work includes faster algorithms for the cluster
LP~\cite{Cao+25}, sublinear-time approximation~\cite{Cohen-Addad+24}, and a
$(1.3865+\epsilon)$-approximation based on approximate dual separation and
generalized pivot rounding~\cite{eff_sep}.  Other directions include local
objectives~\cite{Charikar+17,Kalhan+19,Puleo+18}, query-efficient
algorithms~\cite{Garcia-Soriano+20,Kuroki+24},
and multilayer settings~\cite{Miyauchi+26}.

\paragraph{Modularity Maximization.}
The SDP approach to modularity maximization is related to
\maxagree: \citet*{Dinh+15} add a constant to the objective and apply
correlation-clustering rounding to obtain an additive approximation guarantee.
\citet*{Kawase+21} improve the guarantee by analyzing the additive error
directly.  Our use of Frieze--Jerrum rounding continues this line of work
with the same SDP relaxation.  Dense graphs admit additive approximation
schemes~\cite{DasGupta+13}. Practical approaches also include greedy
methods~\cite{Blondel+08,Clauset+04}, spectral
methods~\cite{Newman06,Richardson+09}, and mathematical
programming~\cite{Agarwal+08,Cafieri+14,Miyauchi+13}.
For broader background on practical community detection
algorithms, see \citet{Fortunato10} and \citet{Fortunato+16}.

\section*{Acknowledgments}
DGS is supported by the grant PID2024-155946NB-I00 funded by Ministerio de Ciencia, Innovación y Universidades (MICIU), Agencia Estatal de Investigación (AEI/10.13039/501100011033) and the European Social Fund Plus (ESF+).

\section*{AI Methodologies Statement}
We used OpenAI ChatGPT (GPT-5.6 Sol) to assist with generating portions of the implementation and verification code for the computer-assisted results.
The authors independently reviewed, tested, and validated all AI-assisted code and take full responsibility for the contents of the paper.

\bibliographystyle{abbrvnat}
\bibliography{main}

\begin{thebibliography}{75}
\providecommand{\natexlab}[1]{#1}
\providecommand{\url}[1]{\texttt{#1}}
\expandafter\ifx\csname urlstyle\endcsname\relax
  \providecommand{\doi}[1]{doi: #1}\else
  \providecommand{\doi}{doi: \begingroup \urlstyle{rm}\Url}\fi

\bibitem[Agarwal and Kempe(2008)]{Agarwal+08}
G.~Agarwal and D.~Kempe.
\newblock Modularity-maximizing graph communities via mathematical programming.
\newblock \emph{European Physical Journal B}, 66\penalty0 (3):\penalty0
  409--418, 2008.

\bibitem[Ailon et~al.(2008)Ailon, Charikar, and Newman]{Ailon+08}
N.~Ailon, M.~Charikar, and A.~Newman.
\newblock Aggregating inconsistent information: Ranking and clustering.
\newblock \emph{Journal of the ACM}, 55\penalty0 (5), 2008.

\bibitem[Bansal et~al.(2004)Bansal, Blum, and Chawla]{Bansal+04}
N.~Bansal, A.~Blum, and S.~Chawla.
\newblock Correlation clustering.
\newblock \emph{Machine Learning}, 56:\penalty0 89--113, 2004.

\bibitem[Blondel et~al.(2008)Blondel, Guillaume, Lambiotte, and
  Lefebvre]{Blondel+08}
V.~D. Blondel, J.-L. Guillaume, R.~Lambiotte, and E.~Lefebvre.
\newblock Fast unfolding of communities in large networks.
\newblock \emph{Journal of Statistical Mechanics: Theory and Experiment},
  2008:\penalty0 P10008, 2008.

\bibitem[Bonchi et~al.(2022)Bonchi, Garc{\'\i}a-Soriano, and Gullo]{Bonchi+22}
F.~Bonchi, D.~Garc{\'\i}a-Soriano, and F.~Gullo.
\newblock \emph{Correlation Clustering}, volume~19 of \emph{Synthesis Lectures
  on Data Mining and Knowledge Discovery}.
\newblock Springer, 2022.

\bibitem[Borell(1985)]{Borell85}
C.~Borell.
\newblock Geometric bounds on the {O}rnstein--{U}hlenbeck velocity process.
\newblock \emph{Zeitschrift f{\"u}r Wahrscheinlichkeitstheorie und verwandte
  Gebiete}, 70\penalty0 (1):\penalty0 1--13, 1985.

\bibitem[Brakensiek et~al.(2026{\natexlab{a}})Brakensiek, Huang, Potechin, and
  Zwick]{Brakensiek+26_sicomp}
J.~Brakensiek, N.~Huang, A.~Potechin, and U.~Zwick.
\newblock Separating {MAX} {2-AND}, {MAX DI-CUT}, and {MAX CUT}.
\newblock \emph{SIAM Journal on Computing}, 55\penalty0 (3):\penalty0
  FOCS23{\char45}268--FOCS23{\char45}308, 2026{\natexlab{a}}.

\bibitem[Brakensiek et~al.(2026{\natexlab{b}})Brakensiek, Huang, Potechin, and
  Zwick]{Brakensiek+26_stoc}
J.~Brakensiek, N.~Huang, A.~Potechin, and U.~Zwick.
\newblock Improved approximation algorithms for {M}ultiway {C}ut by large
  mixtures of new and old rounding schemes.
\newblock In \emph{STOC~'26: Proceedings of the 58th Annual ACM Symposium on
  Theory of Computing}, pages 1441--1452, 2026{\natexlab{b}}.

\bibitem[Brandes et~al.(2008)Brandes, Delling, Gaertler, Gorke, Hoefer,
  Nikoloski, and Wagner]{Brandes+08}
U.~Brandes, D.~Delling, M.~Gaertler, R.~Gorke, M.~Hoefer, Z.~Nikoloski, and
  D.~Wagner.
\newblock On modularity clustering.
\newblock \emph{IEEE Transactions on Knowledge and Data Engineering},
  20\penalty0 (2):\penalty0 172--188, 2008.

\bibitem[Braverman et~al.(2013)Braverman, Makarychev, Makarychev, and
  Naor]{Braverman+13}
M.~Braverman, K.~Makarychev, Y.~Makarychev, and A.~Naor.
\newblock The {Grothendieck} constant is strictly smaller than {Krivine}'s
  bound.
\newblock \emph{Forum of Mathematics, Pi}, 1:\penalty0 e4, 2013.

\bibitem[Cafieri et~al.(2014)Cafieri, Hansen, and Liberti]{Cafieri+14}
S.~Cafieri, P.~Hansen, and L.~Liberti.
\newblock Improving heuristics for network modularity maximization using an
  exact algorithm.
\newblock \emph{Discrete Applied Mathematics}, 163:\penalty0 65--72, 2014.

\bibitem[Cao et~al.(2024)Cao, Cohen-Addad, Lee, Li, Newman, and Vogl]{Cao+24}
N.~Cao, V.~Cohen-Addad, E.~Lee, S.~Li, A.~Newman, and L.~Vogl.
\newblock Understanding the cluster linear program for correlation clustering.
\newblock In \emph{STOC~'24: Proceedings of the 56th Annual ACM Symposium on
  Theory of Computing}, pages 1605--1616, 2024.

\bibitem[Cao et~al.(2025)Cao, Cohen-Addad, Lee, Li, Lolck, Newman, Thorup,
  Vogl, Yan, and Zhang]{Cao+25}
N.~Cao, V.~Cohen-Addad, E.~Lee, S.~Li, D.~R. Lolck, A.~Newman, M.~Thorup,
  L.~Vogl, S.~Yan, and H.~Zhang.
\newblock Solving the correlation cluster {LP} in sublinear time.
\newblock In \emph{STOC~'25: Proceedings of the 57th Annual ACM Symposium on
  Theory of Computing}, pages 1154--1165, 2025.

\bibitem[Charikar et~al.(2005)Charikar, Guruswami, and Wirth]{Charikar+05}
M.~Charikar, V.~Guruswami, and A.~Wirth.
\newblock Clustering with qualitative information.
\newblock \emph{Journal of Computer and System Sciences}, 71\penalty0
  (3):\penalty0 360--383, 2005.

\bibitem[Charikar et~al.(2017)Charikar, Gupta, and Schwartz]{Charikar+17}
M.~Charikar, N.~Gupta, and R.~Schwartz.
\newblock Local guarantees in graph cuts and clustering.
\newblock In \emph{IPCO~'17: Proceedings of the 19th Conference on Integer
  Programming and Combinatorial Optimization}, pages 136--147, 2017.

\bibitem[Chawla et~al.(2006)Chawla, Krauthgamer, Kumar, Rabani, and
  Sivakumar]{Chawla+06}
S.~Chawla, R.~Krauthgamer, R.~Kumar, Y.~Rabani, and D.~Sivakumar.
\newblock On the hardness of approximating {M}ulticut and {S}parsest-{C}ut.
\newblock \emph{Computational Complexity}, 15\penalty0 (2):\penalty0 94--114,
  2006.

\bibitem[Chawla et~al.(2015)Chawla, Makarychev, Schramm, and
  Yaroslavtsev]{Chawla+15}
S.~Chawla, K.~Makarychev, T.~Schramm, and G.~Yaroslavtsev.
\newblock Near optimal {LP} rounding algorithm for correlation clustering on
  complete and complete $k$-partite graphs.
\newblock In \emph{STOC~'15: Proceedings of the 47th Annual ACM Symposium on
  Theory of Computing}, pages 219--228, 2015.

\bibitem[Clauset et~al.(2004)Clauset, Newman, and Moore]{Clauset+04}
A.~Clauset, M.~E.~J. Newman, and C.~Moore.
\newblock Finding community structure in very large networks.
\newblock \emph{Physical Review E}, 70:\penalty0 066111, 2004.

\bibitem[Cohen-Addad et~al.(2022)Cohen-Addad, Lee, and Newman]{Cohen-Addad+22}
V.~Cohen-Addad, E.~Lee, and A.~Newman.
\newblock Correlation clustering with {S}herali-{A}dams.
\newblock In \emph{FOCS~'22: Proceedings of the 63rd IEEE Annual Symposium on
  Foundations of Computer Science}, pages 651--661, 2022.

\bibitem[Cohen-Addad et~al.(2023)Cohen-Addad, Lee, Li, and
  Newman]{Cohen-Addad+23}
V.~Cohen-Addad, E.~Lee, S.~Li, and A.~Newman.
\newblock Handling correlated rounding error via preclustering: A
  1.73-approximation for correlation clustering.
\newblock In \emph{FOCS~'23: Proceedings of the 64th IEEE Annual Symposium on
  Foundations of Computer Science}, pages 1082--1104, 2023.

\bibitem[Cohen-Addad et~al.(2024)Cohen-Addad, Lolck, Pilipczuk, Thorup, Yan,
  and Zhang]{Cohen-Addad+24}
V.~Cohen-Addad, D.~R. Lolck, M.~Pilipczuk, M.~Thorup, S.~Yan, and H.~Zhang.
\newblock Combinatorial correlation clustering.
\newblock In \emph{STOC~'24: Proceedings of the 56th Annual ACM Symposium on
  Theory of Computing}, pages 1617--1628, 2024.

\bibitem[DasGupta and Desai(2013)]{DasGupta+13}
B.~DasGupta and D.~Desai.
\newblock On the complexity of {N}ewman's community finding approach for
  biological and social networks.
\newblock \emph{Journal of Computer and System Sciences}, 79\penalty0
  (1):\penalty0 50--67, 2013.

\bibitem[de~Klerk et~al.(2004)de~Klerk, Pasechnik, and Warners]{deKlerk+04}
E.~de~Klerk, D.~V. Pasechnik, and J.~P. Warners.
\newblock On approximate graph colouring and {MAX}-$k$-{CUT} algorithms based
  on the $\theta$-function.
\newblock \emph{Journal of Combinatorial Optimization}, 8\penalty0
  (3):\penalty0 267--294, 2004.

\bibitem[Demaine et~al.(2006)Demaine, Emanuel, Fiat, and Immorlica]{Demaine+06}
E.~D. Demaine, D.~Emanuel, A.~Fiat, and N.~Immorlica.
\newblock Correlation clustering in general weighted graphs.
\newblock \emph{Theoretical Computer Science}, 361\penalty0 (2):\penalty0
  172--187, 2006.

\bibitem[Dinh et~al.(2015)Dinh, Li, and Thai]{Dinh+15}
T.~N. Dinh, X.~Li, and M.~T. Thai.
\newblock Network clustering via maximizing modularity: Approximation
  algorithms and theoretical limits.
\newblock In \emph{ICDM~'15: Proceedings of the 2015 IEEE International
  Conference on Data Mining}, pages 101--110. IEEE, 2015.

\bibitem[Farouki(2012)]{bernstein_basis}
R.~T. Farouki.
\newblock The {B}ernstein polynomial basis: A centennial retrospective.
\newblock \emph{Computer Aided Geometric Design}, 29\penalty0 (6):\penalty0
  379--419, 2012.

\bibitem[Feige and Schechtman(2002)]{FeigeSchechtman02}
U.~Feige and G.~Schechtman.
\newblock On the optimality of the random hyperplane rounding technique for
  {MAX CUT}.
\newblock \emph{Random Structures \& Algorithms}, 20\penalty0 (3):\penalty0
  403--440, 2002.

\bibitem[Fortunato(2010)]{Fortunato10}
S.~Fortunato.
\newblock Community detection in graphs.
\newblock \emph{Physics Reports}, 486\penalty0 (3):\penalty0 75--174, 2010.

\bibitem[Fortunato and Hric(2016)]{Fortunato+16}
S.~Fortunato and D.~Hric.
\newblock Community detection in networks: {A} user guide.
\newblock \emph{Physics Reports}, 659:\penalty0 1--44, 2016.

\bibitem[Frieze and Jerrum(1997)]{Frieze+97}
A.~Frieze and M.~Jerrum.
\newblock Improved approximation algorithms for {MAX} $k$-{CUT} and {MAX}
  {BISECTION}.
\newblock \emph{Algorithmica}, 18\penalty0 (1):\penalty0 67--81, 1997.

\bibitem[Frieze and Kannan(1996)]{Frieze+96}
A.~Frieze and R.~Kannan.
\newblock The regularity lemma and approximation schemes for dense problems.
\newblock In \emph{FOCS~'96: Proceedings of the 37th Annual IEEE Symposium on
  Foundations of Computer Science}, pages 12--20, 1996.

\bibitem[Garc{\'\i}a-Soriano and Schohn(2026)]{eff_sep}
D.~Garc{\'\i}a-Soriano and A.~Schohn.
\newblock Approximate dual separation for the cluster {LP}: a 1.387
  approximation for correlation clustering.
\newblock \emph{arXiv preprint arXiv:2607.27829}, 2026.

\bibitem[Garc\'{\i}a-Soriano et~al.(2020)Garc\'{\i}a-Soriano, Kutzkov, Bonchi,
  and Tsourakakis]{Garcia-Soriano+20}
D.~Garc\'{\i}a-Soriano, K.~Kutzkov, F.~Bonchi, and C.~E. Tsourakakis.
\newblock Query-efficient correlation clustering.
\newblock In \emph{TheWebConf~'20: Proceedings of The Web Conference 2020},
  pages 1468--1478, 2020.

\bibitem[Garey and Johnson(1979)]{Garey+02}
M.~R. Garey and D.~S. Johnson.
\newblock \emph{Computers and Intractability}.
\newblock W.H. Freeman, 1979.

\bibitem[G{\"a}rtner and Matou{\v{s}}ek(2012)]{Gartner+12}
B.~G{\"a}rtner and J.~Matou{\v{s}}ek.
\newblock \emph{Approximation Algorithms and Semidefinite Programming}.
\newblock Springer, 2012.

\bibitem[Giotis and Guruswami(2006)]{Giotis+06_j}
I.~Giotis and V.~Guruswami.
\newblock Correlation clustering with a fixed number of clusters.
\newblock \emph{Theory of Computing}, 2:\penalty0 249--266, 2006.

\bibitem[Goemans and Williamson(1995)]{Goemans+95}
M.~X. Goemans and D.~P. Williamson.
\newblock Improved approximation algorithms for maximum cut and satisfiability
  problems using semidefinite programming.
\newblock \emph{Journal of the ACM}, 42\penalty0 (6):\penalty0 1115--1145,
  1995.

\bibitem[Goemans and Williamson(2004)]{Goemans+04}
M.~X. Goemans and D.~P. Williamson.
\newblock Approximation algorithms for {MAX-3-CUT} and other problems via
  complex semidefinite programming.
\newblock \emph{Journal of Computer and System Sciences}, 68\penalty0
  (2):\penalty0 442--470, 2004.

\bibitem[Grimmett and Stirzaker(2001)]{grimmet_prob}
G.~R. Grimmett and D.~R. Stirzaker.
\newblock \emph{Probability and Random Processes}.
\newblock Oxford University Press, 2001.
\newblock Third edition.

\bibitem[Heilman(2026{\natexlab{a}})]{Heilman26}
S.~Heilman.
\newblock Sharp hardness for {MAX-3-CUT} and {Q}uantum {MAX-CUT}.
\newblock \emph{arXiv preprint arXiv:2608.00333}, 2026{\natexlab{a}}.

\bibitem[Heilman(2026{\natexlab{b}})]{Heilman26Grothendieck}
S.~Heilman.
\newblock An upper bound on {Grothendieck}'s constant.
\newblock \emph{arXiv preprint arXiv:2606.00247}, 2026{\natexlab{b}}.

\bibitem[Heilman et~al.(2016)Heilman, Mossel, and
  Neeman]{HeilmanMosselNeeman16}
S.~Heilman, E.~Mossel, and J.~Neeman.
\newblock Standard simplices and pluralities are not the most noise stable.
\newblock \emph{Israel Journal of Mathematics}, 213\penalty0 (1):\penalty0
  33--53, 2016.

\bibitem[Isaksson and Mossel(2012)]{IsakssonMossel12}
M.~Isaksson and E.~Mossel.
\newblock Maximally stable {Gaussian} partitions with discrete applications.
\newblock \emph{Israel Journal of Mathematics}, 189\penalty0 (1):\penalty0
  347--396, 2012.

\bibitem[Johansson(2017)]{Johansson17}
F.~Johansson.
\newblock Arb: {E}fficient arbitrary-precision midpoint-radius interval
  arithmetic.
\newblock \emph{IEEE Transactions on Computers}, 66\penalty0 (8):\penalty0
  1281--1292, 2017.

\bibitem[Johansson(2018)]{Johansson18}
F.~Johansson.
\newblock Numerical integration in arbitrary-precision ball arithmetic.
\newblock In \emph{ICMS~'18: Proceedings of the 6th International Congress on
  Mathematical Software}, pages 255--263, 2018.

\bibitem[Kalhan et~al.(2019)Kalhan, Makarychev, and Zhou]{Kalhan+19}
S.~Kalhan, K.~Makarychev, and T.~Zhou.
\newblock Correlation clustering with local objectives.
\newblock In \emph{NeurIPS~'19: Proceedings of the 33rd Annual Conference on
  Neural Information Processing Systems}, pages 9341--9350, 2019.

\bibitem[Kann et~al.(1997)Kann, Khanna, Lagergren, and Panconesi]{Kann+97}
V.~Kann, S.~Khanna, J.~Lagergren, and A.~Panconesi.
\newblock On the hardness of approximating {M}ax $k$-{C}ut and its dual.
\newblock \emph{Chicago Journal of Theoretical Computer Science}, 1997\penalty0
  (2), 1997.

\bibitem[Karger et~al.(1998)Karger, Motwani, and Sudan]{Karger+98}
D.~Karger, R.~Motwani, and M.~Sudan.
\newblock Approximate graph coloring by semidefinite programming.
\newblock \emph{Journal of the ACM}, 45\penalty0 (2):\penalty0 246--265, 1998.

\bibitem[Karloff(1999)]{Karloff99}
H.~Karloff.
\newblock How good is the {G}oemans--{W}illiamson {MAX CUT} algorithm?
\newblock \emph{SIAM Journal on Computing}, 29\penalty0 (1):\penalty0 336--350,
  1999.

\bibitem[Kawase et~al.(2021)Kawase, Matsui, and Miyauchi]{Kawase+21}
Y.~Kawase, T.~Matsui, and A.~Miyauchi.
\newblock Additive approximation algorithms for modularity maximization.
\newblock \emph{Journal of Computer and System Sciences}, 117:\penalty0
  182--201, 2021.

\bibitem[Khot(2002)]{ugc_khot}
S.~Khot.
\newblock On the power of unique 2-prover 1-round games.
\newblock In \emph{STOC~'02: Proceedings of the 34th {A}nnual {ACM} {S}ymposium
  on {T}heory of {C}omputing}, pages 767--775, 2002.

\bibitem[Khot et~al.(2007)Khot, Kindler, Mossel, and O’Donnell]{Khot+07}
S.~Khot, G.~Kindler, E.~Mossel, and R.~O’Donnell.
\newblock Optimal inapproximability results for {MAX-CUT} and other 2-variable
  {CSP}s?
\newblock \emph{SIAM Journal on Computing}, 37\penalty0 (1):\penalty0 319--357,
  2007.

\bibitem[Kuroki et~al.(2024)Kuroki, Miyauchi, Bonchi, and Chen]{Kuroki+24}
Y.~Kuroki, A.~Miyauchi, F.~Bonchi, and W.~Chen.
\newblock Query-efficient correlation clustering with noisy oracle.
\newblock In \emph{NeurIPS~'24: Proceedings of the 38th Annual Conference on
  Neural Information Processing Systems}, 2024.

\bibitem[Miyauchi and Miyamoto(2013)]{Miyauchi+13}
A.~Miyauchi and Y.~Miyamoto.
\newblock Computing an upper bound of modularity.
\newblock \emph{European Physical Journal B}, 86:\penalty0 302, 2013.

\bibitem[Miyauchi et~al.(2026)Miyauchi, Adriaens, Bonchi, and
  Tatti]{Miyauchi+26}
A.~Miyauchi, F.~Adriaens, F.~Bonchi, and N.~Tatti.
\newblock Multilayer correlation clustering.
\newblock In \emph{AISTATS~'26: Proceedings of the 29th International
  Conference on Artificial Intelligence and Statistics}, 2026.

\bibitem[Mossel and Neeman(2015)]{MosselNeeman15Noise}
E.~Mossel and J.~Neeman.
\newblock Robust optimality of {Gaussian} noise stability.
\newblock \emph{Journal of the European Mathematical Society}, 17\penalty0
  (2):\penalty0 433--482, 2015.

\bibitem[Mossel et~al.(2010)Mossel, O'Donnell, and Oleszkiewicz]{Mossel+10}
E.~Mossel, R.~O'Donnell, and K.~Oleszkiewicz.
\newblock Noise stability of functions with low influences: Invariance and
  optimality.
\newblock \emph{Annals of Mathematics}, 171\penalty0 (1):\penalty0 295--341,
  2010.

\bibitem[Newman(2018)]{Newman18}
A.~Newman.
\newblock Complex semidefinite programming and {M}ax-$k$-{C}ut.
\newblock In \emph{SOSA~'18: Proceedings of the 1st Symposium on Simplicity in
  Algorithms}, pages 13:1--13:11, 2018.

\bibitem[Newman and Girvan(2004)]{Newman+04}
M.~E. Newman and M.~Girvan.
\newblock Finding and evaluating community structure in networks.
\newblock \emph{Physical Review E}, 69\penalty0 (2):\penalty0 026113, 2004.

\bibitem[Newman(2006)]{Newman06}
M.~E.~J. Newman.
\newblock Modularity and community structure in networks.
\newblock \emph{Proceedings of the National Academy of Sciences of the United
  States of America}, 103\penalty0 (23):\penalty0 8577--8582, 2006.

\bibitem[O'Donnell(2014)]{o2014analysis}
R.~O'Donnell.
\newblock \emph{Analysis of Boolean Functions}.
\newblock Cambridge University Press, 2014.

\bibitem[Puleo and Milenkovic(2018)]{Puleo+18}
G.~J. Puleo and O.~Milenkovic.
\newblock Correlation clustering and biclustering with locally bounded errors.
\newblock \emph{IEEE Transactions on Information Theory}, 64\penalty0
  (6):\penalty0 4105--4119, 2018.

\bibitem[Raghavendra(2008)]{Raghavendra08}
P.~Raghavendra.
\newblock Optimal algorithms and inapproximability results for every {CSP}?
\newblock In \emph{STOC~'08: Proceedings of the 40th Annual ACM Symposium on
  Theory of Computing}, pages 245--254, 2008.

\bibitem[Raghavendra(2009)]{Raghavendra09Thesis}
P.~Raghavendra.
\newblock \emph{Approximating {NP}-hard Problems: Efficient Algorithms and
  Their Limits}.
\newblock PhD thesis, University of Washington, 2009.

\bibitem[Raghavendra and Steurer(2009{\natexlab{a}})]{RaghavendraSteurer09Gaps}
P.~Raghavendra and D.~Steurer.
\newblock Integrality gaps for strong {SDP} relaxations of {Unique Games}.
\newblock In \emph{FOCS~'09: Proceedings of the 50th Annual IEEE Symposium on
  Foundations of Computer Science}, pages 575--585, 2009{\natexlab{a}}.

\bibitem[Raghavendra and
  Steurer(2009{\natexlab{b}})]{RaghavendraSteurer09Round}
P.~Raghavendra and D.~Steurer.
\newblock How to round any {CSP}.
\newblock In \emph{FOCS~'09: Proceedings of the 50th Annual IEEE Symposium on
  Foundations of Computer Science}, pages 586--594, 2009{\natexlab{b}}.

\bibitem[Richardson et~al.(2009)Richardson, Mucha, and Porter]{Richardson+09}
T.~Richardson, P.~J. Mucha, and M.~A. Porter.
\newblock Spectral tripartitioning of networks.
\newblock \emph{Physical Review E}, 80:\penalty0 036111, 2009.

\bibitem[Shamir et~al.(2004)Shamir, Sharan, and Tsur]{Shamir+04}
R.~Shamir, R.~Sharan, and D.~Tsur.
\newblock Cluster graph modification problems.
\newblock \emph{Discrete Applied Mathematics}, 144\penalty0 (1-2):\penalty0
  173--182, 2004.

\bibitem[Sharma and Vondr{\'a}k(2014)]{Sharma+14}
A.~Sharma and J.~Vondr{\'a}k.
\newblock Multiway cut, pairwise realizable distributions, and descending
  thresholds.
\newblock In \emph{STOC~'14: Proceedings of the 46th Annual ACM Symposium on
  Theory of Computing}, pages 724--733, 2014.

\bibitem[Shinde et~al.(2021)Shinde, Narayanan, and Saunderson]{Shinde+21}
N.~Shinde, V.~Narayanan, and J.~Saunderson.
\newblock Memory-efficient approximation algorithms for {M}ax-k-{C}ut and
  correlation clustering.
\newblock In \emph{NeurIPS~'21: Proceedings of the 35th Annual Conference on
  Neural Information Processing Systems}, pages 8269--8281, 2021.

\bibitem[Swamy(2004)]{Swamy04}
C.~Swamy.
\newblock Correlation clustering: Maximizing agreements via semidefinite
  programming.
\newblock In \emph{SODA~'04: Proceedings of the 15th Annual ACM--SIAM Symposium
  on Discrete Algorithms}, pages 526--527, 2004.

\bibitem[Tan(2008)]{note_inapp_maxagree}
J.~Tan.
\newblock A note on the inapproximability of correlation clustering.
\newblock \emph{Information Processing Letters}, 108\penalty0 (5):\penalty0
  331--335, 2008.

\bibitem[Vazirani(2001)]{Vazirani01}
V.~V. Vazirani.
\newblock \emph{Approximation Algorithms}.
\newblock Springer, 2001.

\bibitem[Williamson and Shmoys(2011)]{Williamson+11}
D.~P. Williamson and D.~B. Shmoys.
\newblock \emph{The Design of Approximation Algorithms}.
\newblock Cambridge University Press, 2011.

\bibitem[Zwick(2002)]{Zwick02}
U.~Zwick.
\newblock Computer assisted proof of optimal approximability results.
\newblock In \emph{SODA~'02: Proceedings of the 13th Annual ACM--SIAM Symposium
  on Discrete Algorithms}, pages 496--505, 2002.

\end{thebibliography}

\end{document}